%% file: main.tex
\documentclass[reqno]{article}

\usepackage[utf8]{inputenc}

\usepackage{amsmath, amsthm}
\usepackage{amsfonts}
\usepackage{amssymb}
\usepackage{amsxtra}
\usepackage{mathrsfs}
\usepackage{mathtools}

\usepackage{bbm}
\usepackage{bm}
\usepackage{color}
\usepackage{float}
\usepackage{graphicx}
\usepackage{caption}
\usepackage{subfigure}
\usepackage{cancel}
\usepackage{comment}
\usepackage{latexsym}
\usepackage{soul}
\usepackage{todonotes}
\usepackage{accents}

\usepackage{authblk}
\usepackage{natbib}

\usepackage{tikz}
\usetikzlibrary{calc,arrows.meta, decorations.pathreplacing, calligraphy}

\usepackage{slashed}
\usepackage{leftidx}
\usepackage{tensor}
\usepackage{epstopdf}

\usepackage{enumerate}
\usepackage{enumitem}

\usepackage{hyperref}

\newtheorem{theorem}{Theorem}[section]
\newtheorem{lemma}[theorem]{Lemma}
\newtheorem{corollary}[theorem]{Corollary}
\newtheorem{proposition}[theorem]{Proposition}
\newtheorem{conjecture}[theorem]{Conjecture}
\newtheorem{definition}[theorem]{Definition}

\theoremstyle{remark}
\newtheorem{remark}[theorem]{Remark}
\newcommand{\mr}{\mathring}
\numberwithin{equation}{section}

\usepackage[left=2cm,right=2cm,top=2.5cm,bottom=2.5cm]{geometry}

\makeatletter

\newcommand{\Rmnum}[1]{\expandafter\@slowromancap\romannumeral #1@}
\makeatother

\begin{document}

\title{Nonlinear stability of continuously self-similar naked singularities for the Einstein-scalar field equations \Rmnum{1}: main results}

\author[1]{Weihao Zheng\thanks{wz344@math.rutgers.edu}}
\affil[1]{\small Department of Mathematics, Rutgers University, Hill Center, 110 Frelinghuysen Road, Piscataway, NJ, USA}
\date{\today}

\maketitle
\begin{abstract}
This is the first part of a series of papers proving the nonlinear stability of a one-parameter family of continuously self-similar $C^{1,\alpha}$ naked singularity solutions, with $0<\alpha\ll1$, to the spherically symmetric Einstein–scalar field equations. The stability holds for initial perturbations lying in a small open neighborhood of the data generating these naked singularity solutions, measured in a localized Hölder topology. 

These continuously self-similar naked singularity spacetimes were previously constructed by Christodoulou [D. Christodoulou, Examples of naked singularity formation in the gravitational collapse of a scalar field, Ann. of Math. 140 (1994), 607–653], who also proved their instability to black hole formation under sufficiently rough perturbations [D. Christodoulou, The instability of naked singularities in the gravitational collapse of a scalar field, Ann. of Math. 149 (1999), 183–217], thereby verifying weak cosmic censorship within a rough functional framework.

In complete contrast, in this paper, we obtain the first nonlinear stability of these naked singularity spacetimes under general perturbations of the same regularity as the background. We rely on the linearized stability result established in the companion paper [J. Singh and W. Zheng, Nonlinear stability of continuously self-similar naked singularities for the Einstein-scalar field equations \Rmnum{2}: linearized stability]. Our result underscores the decisive role of the functional framework in formulating the Weak Cosmic Censorship conjecture. 
\end{abstract}
\section{Introduction}
\label{sec: whole intro}
The Weak Cosmic Censorship (WCC) conjecture in General Relativity asserts that, for \textit{generic\footnote{{The use of the term ‘generic’ is due to the well-known examples of naked singularities constructed by Christodoulou \cite{chris94}.}} regular} initial data, singularities arising in the maximal globally hyperbolic development (MGHD) of Einstein equations (coupled with a reasonable matter model) are hidden inside a black hole region \cite{penroseWCC,chris99B}, i.e., they are not naked; see Section~\ref{sec: infinite blue-shift} for the rigorous definition of naked singularities. 

However, examples of naked singularity spacetimes as solutions to the spherically symmetric Einstein-scalar field equations were rigorously constructed by Christodoulou \cite{chris94}. The Einstein-scalar field equations are a geometric system of PDEs describing the evolution of a Lorentzian spacetime $(\mathcal{M},g)$ and scalar field $\phi: \mathcal{M} \rightarrow \mathbb{R}$, taking the form of \begin{align}
    Ric_{\mu\nu} [g]&= \partial_{\mu}\phi\partial_{\nu}\phi,\label{eq: ESF 1}\\
    \Box_{g}\phi &= 0,\label{eq: ESF 2}
\end{align}
where $Ric$ is the Ricci curvature of the spacetime $(\mathcal{M},g)$.

One important feature of the family of naked singularity spacetimes $(\mathcal{M}_{k},g_{k},\phi_{k})$ constructed in \cite{chris94} is the so-called $k$-self-similarity (see Definition~\ref{intro def: k self similarity}). Later in the seminal work \cite{chris99}, Christodoulou also proved their instability to black hole formation under a class of rough perturbations, hence naked singularities are non-generic within this class of rough solutions. We now summarize the results in \cite{chris94,chris99} below.
\begin{theorem}[\cite{chris94,chris99}]
\label{thm: rough version of chris construction and instability}
There exists a family of naked singularity spacetimes $(\mathcal{M}_{k},g_{k},\phi_{k})$, parametrized by $k$ with $k^{2}\in(0,\frac{1}{3})$, solving the spherically symmetric Einstein-scalar field equations~\eqref{eq: ESF 1}--\eqref{eq: ESF 2}. These spacetimes $(\mathcal{M}_{k},g_{k},\phi_{k})$ exhibit $k$-self-similarity, and the initial data has only finite Hölder regularity $C^{1,k^{2}/(1-k^{2})}$. 

Moreover, for generic scalar field perturbations, whose derivatives belong to the function class of bounded variation (BV), the maximal globally hyperbolic development of the perturbed initial data contains a black hole region.
\end{theorem}

\noindent We shall revisit the above theorem of Christodoulou in a more detailed manner in Section~\ref{intro:sec: previous instability result}. Notably, the $C^{1,k^{2}/(1-k^{2})}$-Hölder regularity of the initial data in Theorem~\ref{thm: rough version of chris construction and instability} should be understood in a localized sense. More precisely, let $\mathcal{O}$ denote the naked singularity, $C_{out}$ the initial hypersurface, and $\underline{C}_{o}^{-}$ the past light cone emanating from the naked singularity (see Figure~\ref{fig:penrose}). Then the initial data are in fact \emph{smooth} on $C_{out}\backslash \underline{C}_{o}^{-}$, and possess Hölder regularity $C^{1,k^{2}/(1-k^{2})}$ only at the intersecting sphere $C_{out}\cap \underline{C}_{o}^{-}$ \footnote{Since there is no conical geometric notion of characterizing Hölder continuity directly on a manifold, we leave the precise definition of the corresponding function class of the initial data until~\eqref{eq: intro: precise definition of the regularity}, after we introduce the gauge choice.}. We refer to function spaces exhibiting this localized Hölder $C^{1,\alpha}$ regularity as $\mathcal{C}_{loc}^{1,\alpha}$ with $\alpha\in(0,1)$; see Definition~\ref{def: definition of the localized holder space} and Definition~\ref{def: localized Holder regularity general} for the precise definitions\footnote{We use slightly different notation here from that in Definition~\ref{def: definition of the localized holder space} and Definition~\ref{def: localized Holder regularity general}. The reader should think of $\mathcal{C}_{loc}^{1,\alpha}$ corresponding to the space $\mathcal{C}_{N}^{1+\alpha,\delta}$ in Definition~\ref{def: definition of the localized holder space} and $\bar{\mathcal{C}}_{N}^{1+\alpha,\delta}$ in Definition~\ref{def: localized Holder regularity general}, with any $N\geq 2$ and $\delta\in(0,1)$.} and the corresponding topology. In particular, the initial data of Christodoulou's naked singularity spacetime belong to $\mathcal{C}_{loc}^{1,k^{2}/(1-k^{2})}$.

While the Weak Cosmic Censorship conjecture concerns the non-genericity of \textit{all} possible naked singularities arising from initial data in some suitable regular function space, its validation in the specific context of the $k$-self-similar naked singularity spacetimes arising from initial data with BV derivatives in view of Theorem~\ref{thm: rough version of chris construction and instability} relies on two pillars:\begin{itemize}
    \item The local well-posedness for the spherically symmetric Einstein-scalar field equations for rough initial data with BV derivatives established in \cite{chris93}, which enables the study of the genericity of naked singularity spacetimes constructed in \cite{chris94} arising from initial data in this function class.
    \item The instability of these naked singularities in \cite{chris99} under perturbations whose derivatives are in the BV class.
\end{itemize}
 
However, a subtle regularity discrepancy exists in this framework. The derivatives of the initial perturbations considered in \cite{chris99} for the instability are in the BV class, which is rougher than the background $k$-self-similar initial data with $\mathcal{C}_{loc}^{1,k^{2}/(1-k^{2})}$ regularity. This motivates a refined formulation of the Weak Cosmic Censorship conjecture for spherically symmetric Einstein-scalar field equations,
depending on the underlying functional framework.
\begin{conjecture}[Weak Cosmic Censorship]
    For generic spherically symmetric initial data in the regularity class $\mathcal{B}$, singularities arising in the MGHD of the Einstein-scalar field equations are not naked.
\end{conjecture}
Theorem~\ref{thm: rough version of chris construction and instability} validates the above conjecture for regularity class $\mathcal{B}$ to be functions with BV derivatives. However, it is natural to ask: \textit{does the instability in Theorem~\ref{thm: rough version of chris construction and instability} persist, or are these naked singularities actually generic, if the regularity class $\mathcal{B}$ is the localized Hölder function space $\mathcal{C}_{loc}^{1,k^{2}/(1-k^{2})}$?}

In contrast to the instability result in Theorem~\ref{thm: rough version of chris construction and instability}, we show that stability holds for initial data in a small open neighborhood of the data generating these naked singularity solutions in Theorem~\ref{thm: rough version of chris construction and instability} in the topology of $\mathcal{C}_{loc}^{1,k^{2}/(1-k^{2})}$, \textbf{thereby falsifying the $\mathcal{C}_{loc}^{1,k^{2}/(1-k^{2})}$-version of the Weak Cosmic Censorship conjecture {in spherical symmetry}.}

\begin{theorem}[Informal Statement of Theorem~\ref{thm: true main theorem}]
    \label{thm:intro:1}
The $k$-self-similar naked singularity spacetimes $(\mathcal{M}_{k},g_{k},\phi_{k})$ constructed in Theorem~\ref{thm: rough version of chris construction and instability} with $0<k^{2}\ll1$ are nonlinearly stable as solutions to the spherically symmetric Einstein-scalar field equations for the scalar field initial data lying in a small open neighborhood of $\phi_{k}|_{C_{out}}$ in the topology of the localized Hölder function class $\mathcal{C}_{loc}^{1,k^{2}/(1-k^{2})}$.
\end{theorem}
\noindent For a detailed statement of the result and outline of the proof, see Section~\ref{sec: main result and proof outline}. 
\paragraph{Discussion on the formulation of the Weak Cosmic Censorship conjecture}
Theorem~\ref{thm:intro:1}, together with Christodoulou's seminal result on the instability of general naked singularities in the BV class~\cite{chris99} (see Section~\ref{sec: infinite blue-shift}), suggests that the precise formulation of the Weak Cosmic Censorship conjecture remains open to debate, and any successful resolution of the conjecture must identify an appropriate underlying functional framework.
\begin{itemize}
    \item On the one hand, Theorem~\ref{thm:intro:1} disproves the spherically symmetric Weak Cosmic Censorship conjecture in the function class $\mathcal{C}_{loc}^{1,k^{2}/(1-k^{2})}$, by showing the nonlinear stability of the $k$-self-similar naked singularity spacetime $(\mathcal{M}_{k},g_{k},\phi_{k})$ for scalar field perturbations in this class with $0<k^{2}\ll1$.

    \item On the other hand, Christodoulou~\cite{chris93} established a BV well-posedness theory for the spherically symmetric Einstein--scalar field equations and proved that, within this BV framework, any potential naked singularity is unstable under BV perturbations, leading to trapped surface formation~\cite{chris99}; see Section~\ref{sec: instability literature review} for details. Thus, the BV class provides a natural setting for Christodoulou's formulation of the instability mechanism. Since the initial data of the scalar field in the $k$-self-similar naked singularity spacetime also belongs to this BV class, Christodoulou's BV instability result particularly applies to the $k$-self-similar naked singularity spacetime.

    \item However, for the only known examples of $k$-self-similar naked singularities for the Einstein-scalar field equations, this BV framework appears slightly unnatural, as the initial data of the background spacetime actually have higher regularity and the relevant perturbations leading to the trapped surfaces considered by Christodoulou in~\cite{chris99} thus have lower regularity than the background. From this perspective, it might be more natural to consider initial data in $\mathcal{C}_{loc}^{1,k^{2}/(1-k^{2})}$ for the $k$-self-similar naked singularities, a regularity class under which the local well-posedness of the Einstein-scalar field equations also holds (see Remark~\ref{rmk: localwellposedness in localized holder}). \item  A third formulation of the Weak Cosmic Censorship conjecture in spherical symmetry can yet be considered in a regularity class intermediate between the BV class of Christodoulou~\cite{chris99} and the $\mathcal{C}_{loc}^{1,k^{2}/(1-k^{2})}$ considered in this paper. Indeed, the derivative of the initial data of the $k$-self-similar naked singularity spacetime $\phi_{k}|_{C_{out}}$ also belongs to a class of functions whose tangential derivatives lie in $C^{0,k^{2}/(1-k^{2})}\cap BV$, noting that \begin{equation*}
\mathcal{C}_{loc}^{0,k^{2}/(1-k^{2})}\subset C^{0,k^{2}/(1-k^{2})}\cap BV\subset BV,
    \end{equation*}
 a Hölder class that allows non-smooth behavior away from $C_{out}\cap \underline{C}_{o}^{-}$ (as opposed to the localized Hölder class $\mathcal{C}_{loc}^{1,k^{2}/(1-k^{2})}$). The validity of the Weak Cosmic Censorship conjecture for~\eqref{eq: ESF 1}--\eqref{eq: ESF 2} in spherical symmetry for initial data with derivatives in $C^{0,k^{2}/(1-k^{2})}\cap BV$ remains open at this time.
\end{itemize}

\paragraph{Exterior perturbations and interior perturbations}
In the Penrose diagram (see Figure~\ref{fig:penrose}) of the naked singularity spacetime $(\mathcal{M}_{k},g_{k},\phi_{k})$, we denote by $\mathcal{O}$ the point corresponding to the singularity. We call the past light cone emanating from the singular point $\mathcal{O}$ the \textit{singular horizon}. We refer to the future of the singular horizon as the \textit{exterior region} of the naked singularity spacetime, and to the past of the singular horizon as the \textit{interior region}. We call perturbations on the initial hypersurface supported in the exterior region \emph{exterior perturbations}. By the finite speed of propagation, such perturbations do not affect the geometry of the spacetime in the interior region. Perturbations that are nontrivial in the interior region on the initial hypersurface are referred to as \emph{interior perturbations}. Note that interior perturbations are not necessarily supported entirely within the interior region; the terminology is intended to emphasize the central role of the interior region in the analysis. Indeed, the main ingredient in the proof of Theorem~\ref{thm:intro:1} is a careful study of solutions in the interior region under interior perturbations.
\begin{figure}[htbp]
\centering
\input{figures/characteristic_naked_singularity}
\caption{Penrose diagram for naked singularity spacetimes}
\label{fig:penrose}
\end{figure}
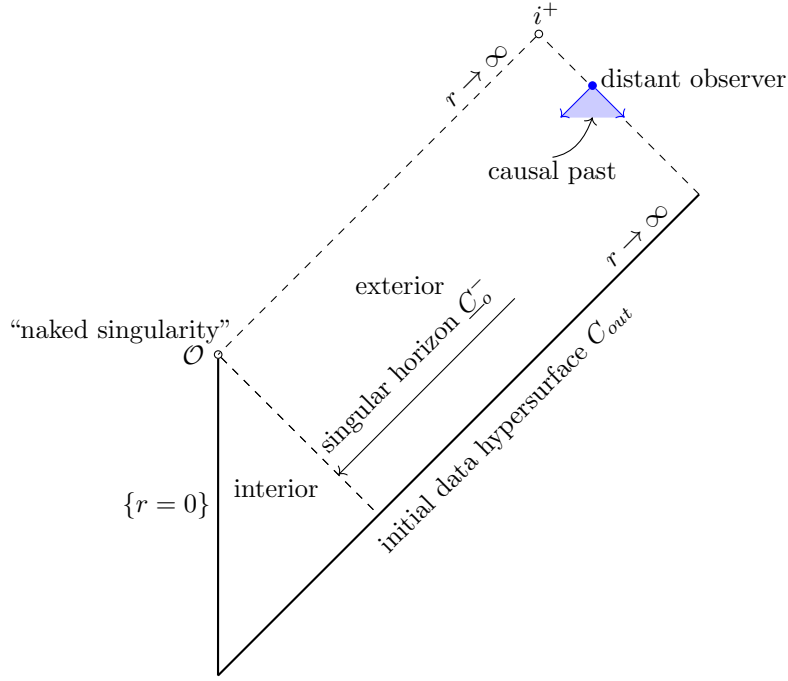

\paragraph{Previous nonlinear (in)stability results and their mechanism}
The original work of Christodoulou \cite{chris99} considered the instability of the $k$-self-similar naked singularities under exterior perturbations of the scalar field with BV derivatives. We now briefly review some later developments on the instability and the partial stability for these spacetimes; see Section~\ref{sec: instability literature review} for a more detailed review. \begin{itemize}
    \item(Progress on instability) Various types of low-regularity perturbations were studied and analogous instability results were established in \cite{an_highcodim,liuli,liuli_outsidesymm,li2025interior}, for perturbations of localized Hölder regularity $\mathcal{C}_{loc}^{1,\alpha}$ with $\alpha$ strictly below $k^{2}/(1-k^{2})$, which represents the regularity of the background; see Section~\ref{sec: instability literature review}. The underlying mechanism driving all of these instabilities is the so-called \textit{blue-shift instability}.
    \item(Partial progress on stability) On the other hand, the works of Singh \cite{singh1,singh2} proved that under \textit{exterior} spherical perturbations with localized Hölder regularity $\mathcal{C}_{loc}^{\alpha}$ for $\alpha\geq k^{2}/(1-k^{2})$, the $k$-self-similar naked singularities are in fact stable, showing that the blue-shift instability mechanism exploited in the previous instability results is sensitive to the regularity of the initial perturbations and is not triggered for \textit{exterior perturbations} of localized Hölder regularity $\mathcal{C}_{loc}^{1,\alpha}$ with $\alpha\geq k^{2}/(1-k^{2})$. This phenomenon is referred to as a \textit{high-regularity stabilizing effect}.
\end{itemize}
Both the blue-shift instability mechanism and the high-regularity stabilizing effect can already be understood at the linear level; see Section 1.3 of the companion paper \cite{zhenglinear} for a detailed discussion. However, high-regularity exterior perturbations alone are insufficient to reveal the stability properties of these naked singularities. To fully understand whether stability holds under perturbations of the localized Hölder regularity $\mathcal{C}_{loc}^{1,k^{2}/(1-k^{2})}$, it is necessary to consider general perturbations that are non-trivial in the interior region, which is the main purpose of Theorem~\ref{thm:intro:1}.
\paragraph{The notion of threshold regularity}
The stability result in this paper, together with previous instability results \cite{chris99,an_highcodim,liuli,liuli_outsidesymm,li2025interior}, highlights the crucial role played by the localized Hölder regularity $\mathcal{C}_{loc}^{1,k^{2}/(1-k^{2})}$ of the background $k$-self-similar naked singularity spacetime $(\mathcal{M}_{k},g_{k},\phi_{k})$, which we refer to as the threshold regularity. Now we can summarize them as follows: \begin{itemize}
    \item For \textit{generic} perturbations strictly below the threshold, the $k$-self-similar naked singularity spacetimes are unstable to black hole formation.
    \item For perturbations at or strictly above the threshold, the $k$-self-similar naked singularity spacetimes are stable.
\end{itemize}
\paragraph{Previous works on the interior of the naked singularity spacetime and the linearized stability}
The study of the (in)stability of the interior of the $k$-self-similar naked singularity spacetimes was initiated by Singh \cite{singh2}, where the linear wave equation $\Box_{g_{k}}\phi = 0$ on this background was analyzed. {For the linear wave equation on the $k$-self-similar naked singularity spacetime}, \cite{singh2} established stability for perturbations above the threshold and instability for perturbations below the threshold. This instability result was subsequently extended to the nonlinear Einstein-scalar field equations by Li \cite{li2025interior} for perturbations below the threshold, exploiting the blue-shift instability mechanism. However, the stability aspect of the linear wave equation alone is \textit{insufficient} to deduce an expectation for nonlinear stability under perturbations at or above the threshold, as it does not capture the influence of the background scalar field or the coupling between the geometry and the scalar field. Indeed, the analysis of the linearized Einstein-scalar field operator in the companion paper \cite{zhenglinear} reveals several additional instability mechanisms and uncovers new structures of the $k$-self-similar spacetimes that eliminate those mechanisms; see Section 1.3 and Section 3.7 of the companion paper~\cite{zhenglinear} for further discussion. The nonlinear stability result in this paper crucially relies on the linearized stability result in the companion paper \cite{zhenglinear}.

\paragraph{Methodology and the role of $0<k^{2}\ll1$}
The nonlinear argument in this paper heavily uses the linearized stability result in the companion paper~\cite{zhenglinear} (see also Theorem~\ref{thm: linearized result}), where $0<k^{2}\ll1$ plays an important role in the analysis of the corresponding unbounded non-self-adjoint linearized Einstein-scalar field operator. Although Christodoulou's original construction of $k$-self-similar naked singularities \cite{chris94} applies for all $0<k^{2}<1/3$, the restriction of $0<k^{2}\ll1$ in our main result comes from the linear analysis. Assuming the validity of the linearized result in~\cite{zhenglinear}, the nonlinear stability will hold for a larger range of $k$. We do not attempt to optimize the admissible range of $k$ in the present paper. It would be an interesting open problem to generalize the current nonlinear stability results to all $0<k^{2}<1/3$.

\paragraph{Relation to the critical collapse}
For asymptotically flat spherically symmetric solutions to the Einstein-scalar field equations, it is well known that the maximal globally hyperbolic development of small initial data will disperse to the Minkowski spacetime \cite{chris86}, while some large initial data will develop a black hole region \cite{chris91}. Numerical studies \cite{choptuik1,gund_understandingcritcollapse,gund_choptuikoutsidesymm} of the transition between the small-data dispersion and the large-data trapped surface formation in gravitational collapse have identified the existence of a smooth naked singularity spacetime, a phenomenon called the critical collapse; see also the review paper \cite{gund_review}. From a mathematical perspective, such a phenomenon can be formulated as follows. Fix a gauge choice on the initial hypersurface $C_{out}$ and consider the MGHD of a family of \textit{smooth} scalar field initial data $\lambda\phi_{0}|_{C_{out}}$. Then there should exist a critical value $\lambda = \lambda_{p}$, below which the corresponding solution will disperse to Minkowski, while above which the solution will form a black hole region. The critical value $\lambda =\lambda_{p}$ will correspond to a naked singularity spacetime; see Figure~\ref{fig:lambda-structure}. Although a rigorous mathematical construction of such a smooth naked singularity spacetime arising from critical collapse remains an open problem, it is natural to ask whether the known example constructed by Christodoulou \cite{chris94} can fit in this context. However, Theorem~\ref{thm:intro:1} proves that the $k$-self-similar naked singularity $(\mathcal{M}_{k},g_{k},\phi_{k})$ with $0<k\ll1$ is \textit{stable} under perturbations $\epsilon\phi_{k}|_{C_{out}}$ with $\vert\epsilon\vert\ll1$, showing that $(\mathcal{M}_{k},g_{k},\phi_{k})$ is \textit{not} critical with respect to the scaling of the initial data {$\lambda\phi_{k}$}; see Figure~\ref{fig:lambda-structure 2}. {More generally, for any spacetime solutions arising from a continuous one-parameter family of initial data $\phi_{\lambda}\in\mathcal{C}_{loc}^{1,k^{2}/(1-k^{2})}$ containing $\phi_{\lambda = 1} = \phi_{k}$, Theorem~\ref{thm:intro:1} showcases that $(\mathcal{M}_{k},g_{k},\phi_{k})$ is not critical. More precisely, there exists $\epsilon$ sufficiently small, such that for $\lambda\in(1-\epsilon,1+\epsilon)$, the corresponding spacetime arising from the scalar field initial data $\phi_{\lambda}$ has a naked singularity. }\begin{figure}[htbp]
\centering
\begin{tikzpicture}[>=Stealth]

\draw[->] (-0.3, 0) -- (9.5, 0) node[right] {$\lambda$};

\draw (0, 0.08) -- (0, -0.08) node[below] {$\lambda = 0$};
\draw (4.5, 0.08) -- (4.5, -0.08) node[below] {$\lambda_p$};
\draw (9.5, 0.08) -- (9.5, -0.08) node[below] {$\lambda \to \infty$};

\draw[decorate, decoration={calligraphic brace, amplitude=6pt, raise=4pt}, thick]
  (0, 0) -- (4.5, 0);
\node[above=14pt] at (2.25, 0) {disperses to Minkowski};

\draw[decorate, decoration={calligraphic brace, amplitude=6pt, raise=4pt}, thick]
  (4.5, 0) -- (9.5, 0);
\node[above=14pt] at (6.5, 0) {black hole formation};

\draw[->, thick] (4.5, 2.2) -- (4.5, 0.6);
\node[above] at (4.5, 2.2) {smooth naked singularity formation};

\end{tikzpicture}
\caption{The numerical expectation for initial data of the form $\lambda \phi_{0}$ for some smooth $\phi_{0}$.}
\label{fig:lambda-structure}
\end{figure}
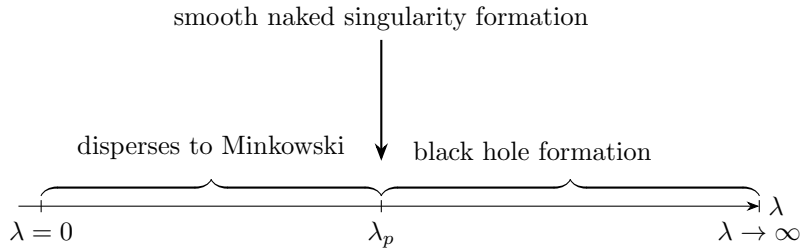
\begin{figure}[htbp]
\centering
\begin{tikzpicture}[>=Stealth]

\draw[->] (-0.3, 0) -- (9.5, 0) node[right] {$\lambda$};

\draw (0, 0.08) -- (0, -0.08) node[below] {$\lambda = 0$};
\draw (3, 0.08) -- (3, -0.08) node[below] {$\lambda = 1-\epsilon$};
\draw (6, 0.08) -- (6, -0.08) node[below] {$\lambda = 1+\epsilon$};
\draw (4.5, 0.08) -- (4.5, -0.08) node[below] {$\lambda = 1$};
\draw (9.5, 0.08) -- (9.5, -0.08) node[below] {$\lambda \to \infty$};
\draw (1, 0.08) -- (1, -0.08) node[below] {$\lambda = \lambda_{*}$};

\draw[decorate, decoration={calligraphic brace, amplitude=6pt, raise=4pt}, thick]
  (0, 0) -- (1, 0);
\node[above=9pt] at (0.5, 0) {disperses to Minkowski};


\draw[decorate, decoration={calligraphic brace, amplitude=6pt, raise=10pt}, thick]
  (3, 0) -- (6, 0);
\node[above=9pt] at (4.5, 0.5) {naked singularity};

\end{tikzpicture}
\caption{The rigorous result for initial data of the form $\lambda \phi_{k}\in \mathcal{C}_{loc}^{1,k^{2}/(1-k^{2})}$.}
\label{fig:lambda-structure 2}
\end{figure}
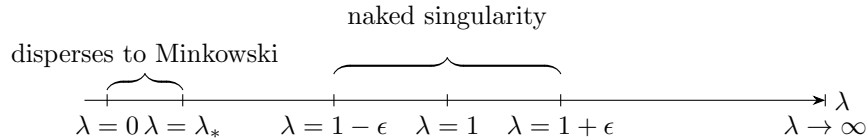

\paragraph{The black hole formation for large smooth perturbations}
{Although Theorem~\ref{thm:intro:1} particularly implies that the $k$-self-similar naked singularity spacetimes are stable under small smooth perturbations, a trapped surface is expected to arise when the size of the initial smooth perturbation is sufficiently large in the $\mathcal{C}_{loc}^{1,k^{2}/(1-k^{2})}$ norm. In particular, in Appendix~\ref{appdx: large smooth perturbation}, we will prove the following theorem, establishing the first black hole formation result under smooth perturbations.}
\begin{theorem}
\label{thm: smooth perturbation instability}
For $0<k^{2}<\frac{1}{3}$, there exists a smooth spherically symmetric function $f\in C^{\infty}(C_{out},\mathbb{R})$ supported on the exterior region of the initial hypersurface $C_{out}$, such that the MGHD of the initial data $\phi_{k}|_{C_{out}}+f$ for the Einstein-scalar field equations contains a trapped surface. 
\end{theorem}
{In fact, the smooth perturbation in Theorem~\ref{thm: smooth perturbation instability} is obtained by mollifying the initial perturbations in Theorem~\ref{thm: rough version of chris construction and instability}. Let $\chi$ be one of the initial perturbations in Theorem~\ref{thm: rough version of chris construction and instability} and $\chi_{\delta,\epsilon} = \epsilon\rho_{\delta}*\chi$, where $\rho_{\delta}$ is a mollifier and $\epsilon$ is a small constant. Then, for initial perturbations $f = \chi_{\delta,\epsilon}$ with $\delta$ sufficiently small, we will show that a trapped surface forms. The smallness of $\delta$ particularly implies the largeness of the initial perturbation in Theorem~\ref{thm: smooth perturbation instability} in $\mathcal{C}_{loc}^{1,k^{2}/(1-k^{2})}$ norm, thus does not contradict Theorem~\ref{thm:intro:1}, which is only true for small initial perturbations. On the other hand, for $\delta$ close to $1$, the $\mathcal{C}_{loc}^{1,k^{2}/(1-k^{2})}$ norm of $\chi_{\delta}$ will be small due to the presence of $\epsilon$. Therefore, for this specific one-parameter family of initial perturbations $\chi_{\delta}$, Theorem~\ref{thm:intro:1} and Theorem~\ref{thm: smooth perturbation instability} show that the $k$-self-similar naked singularity spacetime with $0<k^{2}\ll1$ is not critical for black hole formation; see Figure~\ref{fig:lambda-structure 3}.}
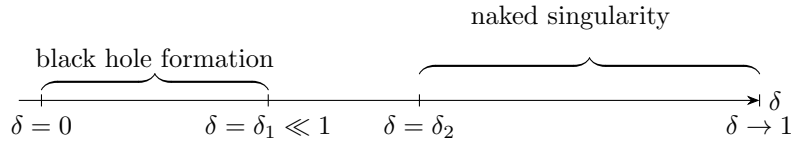
\begin{figure}[htbp]
\centering
\begin{tikzpicture}[>=Stealth]

\draw[->] (-0.3, 0) -- (9.5, 0) node[right] {$\delta$};

\draw (0, 0.08) -- (0, -0.08) node[below] {$\delta = 0$};
\draw (5, 0.08) -- (5, -0.08) node[below] {$\delta =\delta_{2}$};

\draw (9.5, 0.08) -- (9.5, -0.08) node[below] {$\delta \to 1$};
\draw (3, 0.08) -- (3, -0.08) node[below] {$\delta = \delta_{1}\ll1$};

\draw[decorate, decoration={calligraphic brace, amplitude=6pt, raise=4pt}, thick]
  (0, 0) -- (3, 0);
\node[above=9pt] at (1.5, 0) {black hole formation};


\draw[decorate, decoration={calligraphic brace, amplitude=6pt, raise=10pt}, thick]
  (5, 0) -- (9.5, 0);
\node[above=9pt] at (7, 0.5) {naked singularity};

\end{tikzpicture}
\caption{The rigorous result for initial perturbations of the form $\chi_{\delta}$.}
\label{fig:lambda-structure 3}
\end{figure}

\paragraph{Some related open problems}
In this paragraph, we list several open problems related to the present work.
\begin{itemize}
    \item Fix the initial gauge choice and consider the one-parameter family of initial data $\lambda\phi_{k}|_{C_{out}}$. The previous small-data dispersion results \cite{chris93,lukoh_bvscatter} suggest the existence of a threshold $\lambda_{*}$ such that, for all $\lambda\in(0,\lambda_{*})$ the MGHD of the initial data $\lambda\phi_{k}|_{C_{out}}$ disperses to Minkowski. On the other hand, the result in this paper shows that $\lambda\in(1-\epsilon,1+\epsilon)$ lies in the stability regime of the naked singularity; see Figure~\ref{fig:lambda-structure 2}. It is therefore natural to investigate the intermediate regime $\lambda\in[\lambda_{*},1-\epsilon]$ in Figure~\ref{fig:lambda-structure 2}, which is expected to exhibit a transition between dispersion and naked singularity formation.
    \item It is also of interest to understand the regime $\lambda>1+\epsilon$ in Figure~\ref{fig:lambda-structure 2}, at least for sufficiently large $\lambda$, where one expects the formation of trapped surfaces.
    \item{Motivated by the scenarios in Figure~\ref{fig:lambda-structure 3}, one can also study the transition phenomenon when $\delta\in[\delta_{1},\delta_{2}]$. }
    \item Although we show that the $k$-self-similar naked singularity spacetimes are stable for initial data in the localized Hölder $\mathcal{C}_{loc}^{1,k^{2}/(1-k^{2})}$, it remains an open problem to determine whether an instability leading to dispersion can occur within the broader BV framework.
\end{itemize}

\paragraph{Organization of the rest of the introduction}
 The remainder of the introduction aims to situate our result within the literature on the {Weak Cosmic Censorship conjecture} and singularity formation on hyperbolic PDEs. We introduce a rigorous definition of naked singularities in Section~\ref{sec: infinite blue-shift}, recall some properties of the $k$-self-similar naked singularity spacetimes constructed in \cite{chris94}, and review the previous (in)stability results on naked singularities in Section~\ref{sec: instability literature review}. In Section~\ref{intro:subsec:relatedworks}, we review some related works, particularly the naked singularity construction for Einstein equations in Section~\ref{sec: naked singularity construction} and (stable) singularity formation results for other types of nonlinear wave equations in Section~\ref{sec:self-similar blowup}.


\subsection{Globally naked singularities with infinite blue-shift}
\label{sec: infinite blue-shift}

We begin by giving a rigorous definition of a naked singularity spacetime, following the definition in \cite{yakov_rod23} and adapting it to the case of spherically symmetric Einstein-scalar field equations.

    

\begin{definition}
    \label{intro:def:1}
Let $(C_{out},g_{0},\phi_{0})$ denote suitably regular \footnote{In the context of \cite{chris99}, ``suitably regular'' means the function class of $(g_{0},\phi_{0})$ where their tangent derivatives on $C_{out}$ in the outgoing null direction lie in the BV function class; see Section~\ref{sec: recall the local well-posedness}. However, as discussed above, since we can also restrict the function class of the initial data to be Hölder regular, we do not provide the explicit meaning of ``regular'' here. } spherically symmetric, and complete asymptotically flat initial data of the Einstein-scalar field equations~\eqref{eq: ESF 1}--\eqref{eq: ESF 2} on an outgoing null hypersurface $C_{out}$. Further, let $(\mathcal{M},g,\phi)$ be the maximal globally hyperbolic development of this initial data. We say that $(\mathcal{M},g,\phi)$ contains a \textbf{globally naked singularity arising from collapse} if $(\mathcal{M},g,\phi)$ admits a future null infinity $\mathcal{I}^{+}$ which is future incomplete \footnote{One may refer to Definition 1.1 of \cite{yakov_rod23} for the definition of the future null infinity incompleteness.}.




    
\end{definition}
 
Under the spherically symmetric assumption of the initial data, the maximal globally hyperbolic development of this initial data is also spherically symmetric. Hence, we can consider the projection under the action of the $SO(3)$ group $\mathcal{Q}: = \mathcal{M}\backslash SO(3)$. We abuse the notation $(\mathcal{Q},g,\phi)$ to denote the solution under this projection. For any $p\in\mathcal{Q}$, let $r$ be the area radius of the preimage of $p$ under the $SO(3)$-projection and the center $\Gamma\subset\mathcal{Q}$ be the fixed point of the $SO(3)$ action where $r = 0$. Let $\mathcal{O}$ denote the naked singularity in $\mathcal{Q}$ and $C_{o}^{-}$ be the singular horizon; see Figure~\ref{fig:penrose} for the Penrose diagram of a naked singularity spacetime.

We give the following definition of a naked singularity with infinite blue-shift:\begin{definition}
\label{intro:def:1.5}
    A naked singularity has \textbf{infinite blue-shift} if 
    \begin{equation}
    \label{intro:eq:1}
        \int_{C_{o}^{-}} \frac{\mu}{1-\mu}\frac{1}{r}dr = \infty.
    \end{equation}
    where $\mu \doteq \frac{2m}{r} \geq 0$ is the mass ratio, and $m$ is the Hawking mass contained in the sphere of area radius $r$, defined by $m: = \frac{r}{2}\left(1-g(\nabla r,\nabla r)\right)$.
\end{definition}

\begin{remark}
\label{remark:intro:1}
    The BV local well-posedness theory in \cite{chris93} implies that for the regular points on the center $\Gamma$, the corresponding integral~\eqref{intro:eq:1} is finite. Hence, the infinite blue-shift implies that no BV extension of the spacetime can exist in a neighborhood of the singularity $\mathcal{O}$. Further assuming that the causal past of $\mathcal{O}$ is regular, we can consider $\mathcal{O}$ to be a \textit{first} singularity.
\end{remark}




\begin{remark}
\label{remark:intro:2}
The integral~\eqref{intro:eq:1} is closely related to the geometry along the singular horizon $C_{o}^{-}$. In particular, assuming that $L^{\prime}$ is the geodesic ingoing null normal vector field on $C_{o}^{-}$ and $L$ is the outgoing null vector field with $g(L,L^{\prime}) = -2$, then the infinite blue-shift condition~\eqref{intro:eq:1} is equivalent to $\nabla_{L}r\rightarrow 0$ when approaching the singularity $\mathcal{O}$.

\end{remark}

\subsection{$k$-self-similar naked singularities}
\label{sec: k-self-similarity}
In this section, we introduce the $k$-self-similar naked singularity constructed in \cite{chris94}. We first recall the following definition of the $k$-self-similarity, first introduced in \cite{chris94}.
\begin{definition}[$k$-self-similarity]
\label{intro def: k self similarity}
    Let $(\mathcal{M},g,\phi)$ be a solution to the Einstein-scalar field equations~\eqref{eq: ESF 1}--\eqref{eq: ESF 2}. We say that $(\mathcal{M},g,\phi)$ is $k$-self-similar if it possesses a vector field $K$ such that \begin{equation}
        \mathcal{L}_{K}g = 2g,\quad K\phi = -k,
    \end{equation}
    where $k\in\mathbb{R}$ is a free parameter.
\end{definition}
Under the assumption of the spherical symmetry and the $k$-self-similarity, Christodoulou successfully reduced the equation~\eqref{eq: ESF 1}--\eqref{eq: ESF 2} to a $2\times 2$ autonomous system of ODEs. A careful phase portrait analysis gives the existence of the $k$-self-similar naked singularity $(\mathcal{M}_{k},g_{k},\phi_{k})$ in the $k$ range $k^{2}\in(0,\frac{1}{3})$. For a detailed introduction on the construction, see Section~\ref{sec: pre on k naked singularity sapcetime}.

We briefly recall some basic properties of $(\mathcal{M}_{k},g_{k},\phi_{k})$. Recall the global double-null coordinates $(u,v,\theta,\varphi)$ \begin{equation*}
    g = -\frac{1}{2}\Omega^{2}du\otimes dv-\frac{1}{2}\Omega^{2}dv\otimes du+r^{2}(d\theta^{2}+\sin^{2}\theta d\varphi^{2}),
\end{equation*}
where $\Omega^{2}$ is called the lapse function. We can further assume the conformal Killing vector field $K$ takes the form \footnote{ A natural gauge choice of $K$ seems to be $K = u\partial_{u}+v\partial_{v}$. However, since the $k$-self-similar spacetime constructed in \cite{chris94} is of finite Hölder regularity, such a choice of $K$ will make the above double-null coordinates non-global.} of \begin{equation*}
    K = u\partial_{u}+(1-k^{2})v\partial_{v}.
\end{equation*}


A general continuous $k$-self-similar ansatz consistent with the symmetries of the Einstein-scalar field system is given by 
\begin{equation}
\label{intro:eq:2}
    r(u,v) = (-u)\mr{r}(z), \ \Omega^{2}(u,v) = (-u)^{k^{2}}\mr{\Omega}^{2}(z), \ \phi(u,v) = \mr{\phi}(z) - k \log(-u),
\end{equation}
where 
\begin{equation}
\label{intro:eq:3}
    z \doteq \frac{v}{(-u)^{1-k^2}}
\end{equation}
is a similarity coordinate. We can further fix the gauge choice of the center of this spherically symmetric spacetime to be $\Gamma:= \{r = 0\} = \{z=  -1\}$.

The unique vanishing point  $(u,v) = (0,0)$ of $K$ is called the scaling center, and it corresponds to the naked singularity $\mathcal{O}$ in $\mathcal{M}_{k}$. Under the double null coordinates, the singular horizon $C_{0}^{-}$ simply becomes $\{v = 0\}$. The exterior region corresponds to $\{v>0\}$, while the interior region corresponds to $\{v<0\}$. The initial hypersurface $C_{out}$ will become $\{u = -1\}$. The initial hypersurface is parametrized by $ v\in[-1,\infty)$. It follows from Christodoulou's original construction \cite{chris94} that the regularity of the initial data of $(\mathcal{M}_{k},g_{k},\phi_{k})$ is 
\begin{equation}
\label{eq: intro: precise definition of the regularity}
\begin{aligned}
    &r|_{u  =-1} = \mr{r}(v)\in C^{\infty}\left([-1,\infty)\backslash\{0\}\right)\cap C^{2}([-1,\infty)),\quad \Omega^{2}|_{u = -1}\in C^{\infty}\left([-1,\infty)\backslash\{0\}\right)\cap C^{1,\frac{k^{2}}{1-k^{2}}}([-1,\infty)),\\& \phi|_{u = -1}\in C^{\infty}\left([-1,\infty)\backslash\{0\}\right)\cap C^{1,\frac{k^{2}}{1-k^{2}}}([-1,\infty)).
\end{aligned}
\end{equation}
It follows from the construction that $\mathcal{O}$ is an infinite blue-shift singularity in the sense of Definition~\ref{intro:def:1.5}.

\subsection{Previous results on (in-)stability of naked singularity spacetimes} 
\label{sec: instability literature review}
In this section, we revisit the aforementioned instability and stability results on naked singularities in a more detailed manner. 

\subsubsection{Previous exterior instability results}
\label{intro:sec: previous instability result}
The first instability result of the $k$-self-similar naked singularity was established in \cite{chris99} by Christodoulou.

\begin{theorem}[\cite{chris99,liuli}]
    \label{thm:intro:2}
    Let $(\mathcal{M},g,\phi)$ denote a naked singularity spacetime with an infinite blue-shift singularity $\mathcal{O} = (0,0)$ in the sense of Definitions~\ref{intro:def:1} and Definition~\ref{intro:def:1.5}, arising from initial data for the scalar field\footnote{We suppress the initial data for $r(u,v)$ in this discussion.} $\partial_v\phi(-1,v) =h(v) \in \text{BV}$ on an outgoing null cone $\{u=-1\}$. Then there exist perturbations $f_0(v) \in \text{BV}, f_1(v) \in \text{AC}$ with support on the exterior region of the naked singularity $\{v\geq0\}$ such that the maximal development of data
    \begin{equation*}
        \partial_v\phi(-1,v) = h(v) + \lambda_0 f_0(v) + \lambda_1 f_1(v)
    \end{equation*}
    contains a trapped surface for all non-zero $\lambda_0, \lambda_1$ sufficiently small. 

    Let $\mathcal{E}_{*}$ be the set of all initial data of $\partial_{v}\phi$ which will develop naked singularities with infinite blue-shift. Then $\mathcal{E}_{*}$ has co-dimension at least 1 in the space of \text{BV} datasets.
\end{theorem}
The above theorem applies to a wide class of naked singularity spacetimes of the Einstein-scalar field equations satisfying the infinite blue-shift~\eqref{intro:eq:1}. In particular, it contains the $k$-self-similar naked singularity constructed in \cite{chris94}. When restricting the above result to a refined class of naked singularities, the argument in \cite{chris99} yields\footnote{We thank Carsten Gundlach for sharing this calculation with us.} the following, which improves the regularity of the perturbations:
\begin{corollary}[\cite{chris99}]
\label{thm:intro:3}
Assume $(\mathcal{M},g,\phi)$ satisfies, in addition to the assumptions of Theorem~\ref{thm:intro:2}, 
\begin{equation}
\label{intro:eq:8}
    \mu(u,0) = \mu_0 \in (0,1).
\end{equation}
Then there exist perturbations
\begin{equation*}
        \partial_v\phi(-1,v) = h(v) + \lambda f(v)
\end{equation*}
with support on $\{v\geq0\}$ leading to a trapped surface for all $\lambda$ sufficiently small, where $f(v)$ has Hölder regularity
\begin{equation*}
    f(v) \in C^{0,c(\mu_0)}
\end{equation*}
for an explicit constant
\begin{equation*}
    c(\mu_0) = \frac{\mu_0}{4(1+4\mu_0)}.
\end{equation*}

\end{corollary}
\begin{remark}
    For $k$-self-similar naked singularity spacetimes, on the singular horizon, the condition $\mu_0 = \frac{k^2}{1+k^2}$ holds. Hence, the perturbations in Corollary~\ref{thm:intro:3} have regularity $C^{1, \frac{k^2}{4(1+5k^2)}}.$ Since $\frac{k^2}{4(1+5k^2)} < \frac{k^2}{1-k^2},$ this regularity is strictly lower than the threshold regularity.
\end{remark}

It is worth emphasizing that the proof in \cite{chris99} does not aim to optimize the regularity of unstable perturbations for any particular class of naked singularities (e.g., self-similar spacetimes). In fact, for a general BV naked singularity, the geometry of the singular horizon can be significantly more complex than that of the $k$-self-similar spacetimes (cf. the behavior of $\mu$~\eqref{intro:eq:8}). It is therefore unlikely that a Hölder instability result holds without additional assumptions on the spacetime.

Modern works on WCC take inspiration from Christodoulou's spirit in spherical symmetry, exploiting the blue-shift amplification (see Section 1.3 in \cite{zhenglinear} for the details) of ingoing wave packets. The original proof in \cite{chris99} uses a contradiction argument, thus providing no information on the quantitative nature of trapped surfaces formation. Later in \cite{liuli}, Li--Liu gave another proof of the trapped surface formation under exterior BV perturbations, using a well-designed bootstrap argument. The work \cite{liuli_outsidesymm} considers the instability of spherically symmetric infinite blue-shift naked singularities under exterior perturbations outside of the spherical symmetry, significantly raising the co-dimension of the initial data set $\epsilon_{*}$.

Later works have primarily considered the (in)stability of the naked singularities in the self-similar regime, as the additional symmetry renders the (in)stability analysis more tractable. In \cite{singh2}, Singh showed the instability of the $k$-self-similar naked singularity under the exterior spherical perturbations of the scalar field strictly below the threshold regularity. An further showed in \cite{an_highcodim} that the $k$-self-similar naked singularity is unstable under exterior non-spherical BV perturbations, leading to the emergence of an anisotropic apparent horizon.

\subsubsection{Previous exterior stability results}
Rigorous mathematical study of the stability aspect of the $k$-self-similar naked singularity to the Einstein-scalar field equations was initiated in the work of Singh \cite{singh1}, where the author showed the $k$-self-similar naked singularity is nonlinearly stable under spherical exterior perturbation of $C^{2}$ regularity. This regularity has been improved in the later work of Singh \cite{singh2} to be any regularity at or above the threshold. More precisely, the work \cite{singh2} shows that the $k$-self-similar naked singularity is nonlinearly orbitally stable under any small spherical perturbations of regularity at the threshold and asymptotically stable under perturbations of regularity above the threshold. Moreover, the orbital stability result in~\cite{singh2} can further be upgraded to the asymptotic stability by the approach in the present paper; see Section~\ref{sec: main result and proof outline} for details.

Therefore, to address the issue of the validity of WCC under more regular functional frameworks, a careful study of the (in)stability of the $k$-self-similar naked singularity under general Hölder perturbations at or above the threshold is required. {This paper gives a complete answer in this direction, showing that the $k$-self-similar naked singularity spacetime is indeed stable for general perturbations at or above the threshold regularity. Moreover, the approach in this paper can further upgrade the threshold exterior orbital stability to asymptotic stability for general perturbations at the threshold.}

\subsubsection{Previous (in)stability works on interior perturbations}
\label{sec: intro: interior wave}
The motivation for studying the (in)stability of naked singularities under interior perturbations is twofold. 

First, in order to demonstrate the nonlinear stability of the $k$-self-similar naked singularity under perturbations at or above the threshold, it is necessary to consider general perturbations rather than restricting to exterior perturbations alone. In fact, the core difficulty of the stability problem lies in analyzing the interior region under interior perturbations.

Second, in a different context, for aforementioned naked singularities arising from critical collapse, numerics \cite{choptuik1,gund_understandingcritcollapse,gund_choptuikoutsidesymm} suggest that certain perturbations can lead to dispersion, producing a spacetime without any singularity. We note that such perturbations cannot be of exterior type, since exterior perturbations do not influence the presence or absence of the singularity. This fact showcases the importance of studying the interior region under interior perturbations. While the present paper focuses on the $k$-self-similar naked singularities, we hope that the approach developed here will provide insight into this phenomenon in future work.

The study of the (in)stability of the $k$-self-similar naked singularity under interior perturbations was initiated in the work of Singh \cite{singh2}, where a linear wave equation \begin{equation}
    r_{k}\Box_{g_{k}}\phi = r_{k}\partial_{u}\partial_{v}\phi+\partial_{u}r_{k}\partial_{v}\phi+\partial_{v}r_{k}\partial_{u}\phi = 0\label{eq: linear wave equation introduction}
\end{equation} 
on a fixed $k$-self-similar naked singularity spacetime was considered. The result of Singh \cite{singh2} shows that for perturbations \textit{above} the threshold, solutions to the linear wave equation exhibit better behavior than $\phi_{k}$, whereas for perturbations below the threshold, the behavior is worse. The nonlinear instability result \cite{li2025interior} of Li for interior perturbations is consistent with this latter linear instability aspect, showing that the $k$-self-similar naked singularity is unstable under perturbations supported in the interior region near the past light cone emanating from the singularity.

As we have mentioned above, the linear wave equation on the $k$-self-similar geometry is insufficient for the nonlinear stability for the following three reasons \begin{itemize}
    \item The linear wave equation does not capture the influence of the background scalar field $\phi_{k}$, whose regularity is one derivative less than $r_{k}$. {In fact, the analysis of the linear wave equation~\eqref{eq: linear wave equation introduction} in~\cite{singh2} crucially relies on the $C^{1,k^{2}/(1-k^{2})}$ regularity class of $\partial r_{k}$. However, in the study of the nonlinear stability, the term $\partial\phi_{k}$ appears in the {propagation} equation for the lapse function $\Omega^{2}$, and $\partial\phi_{k}$ is only $C^{0,k^{2}/(1-k^{2})}$, thereby invalidating the argument in~\cite{singh2}.}

    \item The linear wave equation fails to account for the interaction between the background geometry and scalar field when studying the nonlinear stability.
    \item The linear wave equation could not provide any insight for nonlinear perturbations at the threshold.
\end{itemize}
The first two points are resolved in the analysis of the linearized Einstein-scalar field operator in the companion paper~\cite{zhenglinear}, while the resolution of the last point lies in the scaling invariance of the Einstein-scalar field equations; see Section~\ref{sec: outline of the proof whole sec}.

\subsection{Related works}
\label{intro:subsec:relatedworks}

\subsubsection{Naked singularity constructions in general relativity}
\label{sec: naked singularity construction}
In the original formulation of WCC by Penrose \cite{penroseWCC}, it was conjectured that naked singularities in the Einstein equations (coupled with a reasonable matter model) do not exist. The construction of naked singularity solutions to the Einstein-scalar field equations in \cite{chris94,choptuik1} changes the paradigm drastically, and the word ``generic'' was added for the validity of WCC \cite{chris99B}. The original construction of Christodoulou heavily exploited the spherical symmetry and the exact $k$-self-similarity in the sense of Definition~\ref{intro def: k self similarity} and the expression~\eqref{intro:eq:2}. Under these symmetries, the unknown functions in the Einstein-scalar field equations become $\mr{r}(z)$, $\mr{\Omega}(z)$, and $\mr{\phi}(z)$, which dramatically reduces the complexity of the Einstein-scalar field equations. Moreover, by considering the algebraic combinations of $\mr{r}$, $\mr{\Omega}$, and $\mr{\phi}$, one can further derive a $2\times 2$ autonomous system of ODEs. Christodoulou conducted a careful phase portrait analysis of this system of ODEs and constructed the $k$-self-similar naked singularity spacetimes for $k^{2}\in(0,\frac{1}{3})$.

The original approach of Christodoulou is not immediately generalizable to other models, as it relies on a reduction to a $2\times 2$ system of ODEs under strong symmetry assumptions.

\paragraph{Construction of naked singularities for the Einstein vacuum equations}
For the Einstein vacuum equations, naked singularities cannot exist in the regime of spherical symmetry due to the Birkhoff theorem \cite{Wald1984}. 
Nonetheless, in the breakthrough of Rodnianski and Shlapentokh-Rothman \cite{yakov_rod23,yakov22}, a proper notion of the asymptotic $k$-self-similarity was introduced, and a family of asymptotically $k$-self-similar naked singularities was constructed for $0<k\ll 1$. The spacetimes share a limited Hölder regularity near the past light cone of the singularity $\{v = 0\}$ under the double-null gauge. The authors divided the construction into two parts: the construction of naked singularities exterior \cite{yakov_rod23}; and the construction of naked singularities interior \cite{yakov22}. In particular, the work \cite{yakov_rod23} provided a robust framework for constructing naked singularity exteriors and demonstrating directly from the construction that naked singularities constructed in \cite{yakov_rod23} are stable under the exterior perturbations of sufficiently high regularity. 

\paragraph{Construction of naked singularities for the Einstein--Euler equations}
Under exact self-similarity and spherical symmetry, an asymptotically flat naked singularity spacetime to Einstein--Euler was also constructed in \cite{guo_hadzic_jang21}. In contrast to the naked singularities of the Einstein-scalar field equations or the Einstein vacuum equations, the naked singularity spacetimes in \cite{guo_hadzic_jang21} have smooth initial data. However, this naked singularity is still believed to be unstable under gravitational perturbations. In the recent work of Li--Zhu \cite{li2025scalar}, the authors studied a toy model addressing this instability aspect, where they proved the instability of this naked singularity spacetime as a solution to the Einstein--Euler-scalar field equations under scalar field perturbations of finite Hölder regularity.

\paragraph{Critical collapse and discretely self-similar naked singularities}
For the aforementioned naked singularity spacetimes arising from gravitational collapse, numerical evidence suggests that they exhibit the so-called discrete self-similarity. Since it lies at the threshold of the dispersion and black hole formation in the process of the gravitational collapse, such a naked singularity has been numerically identified to be of codimension one in the moduli space of all regular initial data and has unstable directions leading to the dispersion. Recently, the exterior region corresponding to a discretely self-similar naked singularity to the Einstein-scalar field equations has been constructed \cite{cicortas_kehle24}.

\paragraph{Construction of naked singularities for Einstein--Maxwell-charged scalar field equations}
The mathematical construction of naked singularities in the aforementioned works relies on the self-similarity. However, if one considers the Einstein--Maxwell-charged scalar field equations \begin{align}
&Ric_{\mu\nu}(g)-\frac{1}{2}R(g)g_{\mu\nu} = T_{\mu\nu}^{EM}+T_{\mu\nu}^{SC},\label{eq:whole equation1}\\&
T_{\mu\nu}^{EM} = 2\left(g^{\alpha\beta}F_{\alpha\mu}F_{\beta\nu}-\frac{1}{4}F^{\alpha\beta}F_{\alpha\beta}g_{\mu\nu}\right),\quad \nabla^{\mu}F_{\mu\nu} = iq_{0}\left(\frac{\phi\overline{D_{\nu}\phi}-\overline{\phi}D_{\nu}\phi}{2}\right),\ F=dA,\\&
T^{SC}_{\mu\nu} = 2\left(\Re{\left(D_{\mu}\phi\overline{D_{\nu}\phi}\right)}-\frac{1}{2}\left(g^{\alpha\beta}D_{\alpha}\phi\overline{D_{\beta}\phi}\right)g_{\mu\nu}\right),\quad D_{\mu} = \nabla_{\mu}+iq_{0}A_{\mu},\\&
g^{\mu\nu}D_{\mu}D_{\nu}\phi = 0,\label{eq:whole equation5}
\end{align}
where $q_{0}$ is the scalar field charge, $F$ is the Maxwell field, and $A$ is the electromagnetic potential, the presence of the scalar field charge breaks the translation invariance. Therefore, one cannot directly impose the twisted self-similarity as in the work of Christodoulou \cite{chris94}. Nonetheless, using the fact that the back-reaction of the Maxwell field is always sub-leading, the author can construct a naked singularity solution to the Einstein--Maxwell-charged scalar field equations in the {upcoming work~\cite{zhengcharged}}.
\begin{theorem}
    There exists a family of asymptotically flat spherically symmetric $C^{1,\alpha_{k}}$ naked singularity solutions $(\mathcal{M}_{k},g_{k},\phi_{k},F_{k})$ to~\eqref{eq:whole equation1}--\eqref{eq:whole equation5}, parametrized by $k$ with $0<k\ll1$ and satisfying the upper bounds predicted by self-similarity. Moreover, on the past light cone emanating from the singularity, the Maxwell field has the {lower bound} \begin{equation*}
        \vert F\vert:=\left(g^{\alpha\beta}g^{\mu\nu}F_{\alpha\mu}F_{\beta\nu}\right)^{\frac{1}{2}}\gtrsim 1.
    \end{equation*}
\end{theorem}
The exterior stability of this naked singularity spacetime already follows from the construction. For general high-regularity perturbations, since the approach in the current paper crucially relies on the exact $k$-self-similarity, it remains {open} whether the method developed here can be applied to establish stability. We leave this question for future work.

\subsubsection{Self-similar blowup in nonlinear wave equations}
\label{sec:self-similar blowup}
A common feature of all the naked singularity spacetimes mentioned in Section~\ref{sec: naked singularity construction} is self-similarity, in the sense that the spacetime satisfies an upper bound predicted by the (twisted) self-similarity. More generally, the study of self-similar solutions to nonlinear wave equations is often viewed as a first step toward understanding their behavior in the large-data regime. Over the past decades, the construction of self-similar solutions to various nonlinear wave equations, along with the analysis of their stability and instability properties, has been extensively studied. We briefly summarize some of the progress below.
\paragraph{Self-similar blowup solutions to focusing semi-linear wave equations $\Box u = -\vert u\vert^{p-1}u$} The focusing semi-linear wave equations in $\mathbb{R}^{1+d}$ with $d\geq3$ take the form of
\begin{equation}
    \Box u: = -\partial_{t}^{2}u+\Delta u = -\vert u\vert^{p-1}u,\label{eq: semilinear wave equation}
\end{equation}
The self-similar solution $u(t,r)$ to the above focusing semi-linear wave equations takes the form of \begin{equation}
    u(t,r) = (T-t)^{-\frac{2}{p-1}}\varphi(\frac{r}{T-t}).\label{intro: eq: self-similar form}
\end{equation}
A spatially homogeneous self-similar blowup solution for the radial focusing semi-linear wave equation~\eqref{eq: semilinear wave equation} is: \begin{equation*}
    u_{T} = c_{p}(T-t)^{-\frac{2}{p-1}},
\end{equation*}
which blows up at $t = T$. This blowup solution has been proved to be stable for all dimensions and all values of $p>1$ \cite{ostermann2024stable,donninger2012stable,donninger2014stable,donninger2017stable,merle2003determination,merle2005determination,merle2007existence,merle2015stability,merle2016dynamics,donninger2016blowup,donninger2010nonlinear,chatzikaleas2019stable}.

Besides the aforementioned spatially homogeneous self-similar blowup solutions, countably many self-similar solutions with non-trivial $\varphi$ in~\eqref{intro: eq: self-similar form} have also been constructed \cite{bizon2007self,bizon2010self,dai2021self,kycia2011self}; see \cite{csobo2024blowup,chen2024co} for the progress on co-dimension one stability results on some of these solutions.

\paragraph{Self-similar blowup solutions to the supercritical wave map equations}
The supercritical wave map equations in $\mathbb{R}^{1+d}$ with $d\geq 3$ take the form of: \begin{equation}
    -\partial_{t}^{2}U+\Delta U = U\left((\partial_{t}U)^{2}-\left\vert\nabla_{x}U\right\vert^{2}\right),\label{eq: wave map equation}
\end{equation}
where $U:\mathbb{R}^{1+d}\rightarrow S^{d}$. Under the co-rotational setting and the spherical symmetry, the solution $U$ can be written as \begin{equation*}
    U(t,x) = \begin{bmatrix}
        \cos(u(t,r))\frac{x}{\vert x\vert}\\
        \sin(u(t,r))
    \end{bmatrix}.
\end{equation*}
Then the wave map equations will be reduced to the following equation for $u$: \begin{equation*}
    \partial_{t}^{2}u-\partial_{r}^{2}u-\frac{d-1}{r}\partial_{r}u +\frac{2\sin(2u)}{r^{2}}=0, 
\end{equation*}
which has an explicit self-similar blowup solution: \begin{equation*}
    u^{T} = 2\arctan\left(\frac{r}{\sqrt{d-2}(T-t)}\right).
\end{equation*}
A crucial step towards proving the stability of this co-rotational self-similar solution is to prove its mode stability; see \cite{donninger2011stable,donninger2012stable,costin2016proof,costin2017mode,chatzikaleas2017blowup,donninger2024stable,glogic2025globally,donninger2023optimal,donninger2025optimal,weissenbacher2025mode} for the stability results within the co-rotational perturbations and \cite{donninger2026blowup,donninger2026mode} for the most recent stability results outside of symmetry.

There are also many interesting works on the construction and the stability of self-similar blowup solutions to the nonlinear Schrödinger equations and Yang--Mills equations. We do not discuss these works in detail and only briefly mention them here. For the Yang--Mills equations, readers can refer to \cite{bizon2015generic,cazenave1998harmonic,bizon2006formation,bizon2000equivariant} for the construction of self-similar blowup solutions; \cite{donninger2014stable, costin2016stability,donninger2023globally,donninger2024stable} for the stability of those self-similar blowup solutions. For the nonlinear Schrödinger equations, readers can refer to \cite{donninger2024self} for the construction of self-similar blowup solutions with cubic nonlinearity and \cite{donninger2025self} for their stability.

Lastly, related to the critical collapse discussed in the previous section, we want to mention that a discretely self-similar solution to a wave maps type equation has been constructed, and its stability has been proved \cite{glogic2025existence}.

\subsubsection{Type II blowup in nonlinear wave equations}
While not directly related to the present work on the stability of (twisted) self-similar solutions to the Einstein-scalar field equations, we briefly highlight several notable results concerning Type II blow-up solutions to nonlinear wave equations. Roughly speaking, a Type II blow-up refers to singularity formation that does not occur at the self-similar (Type I) blow-up rate predicted by the scaling of the equation.

For the critical wave map equations (\eqref{eq: wave map equation} with $d = 2$), it is well-known \cite{shatah1994cauchy} that there is no self-similar blow-up solution in the co-rotational (equivariant) setting. However, two different kinds of Type II blow-up solutions under the co-rotational symmetry were constructed \cite{kriegerschlagtat_WM,raphael2012stable}. Their stability has been proved in \cite{kriegermiao_stabwavemap_corot,kriegermiaoschlag_stabwavemap_full} and \cite{raphael2012stable}, respectively. Under the $k$-equivariant setting, the work \cite{IgorSterb} constructed Type II blow-up solutions and showed their stability.

\subsection{Acknowledgments}
\label{intro:subsec:acknowledgments}
The author acknowledges the fundamental contributions from Jaydeep Singh, including his patient guidance, valuable discussions, and friendship. The author would like to thank his advisor, Maxime Van de Moortel, for his encouragement, kind support, careful reading of the manuscript, and many useful pieces of advice. The author would also like to thank Yakov Shlapentokh-Rothman, Igor Rodnianski, Mihalis Dafermos, and Jonathan Luk for their interest in this work and enlightening conversations. The author gratefully acknowledges the support
from the NSF Grant DMS-2247376.

\section{Preliminary}
\label{prelims:sec}

\subsection{Reduction of the Einstein-scalar field system under the spherical symmetry}
Recall that the Einstein-scalar field equations take the form of \begin{align}
Ric(g)_{\mu\nu}&= \partial_{\mu}\phi\partial_{\nu}\phi,\label{ric equation}\\
\Box_{g}\phi&=0,\label{scalar field equation}
\end{align}
where $g$ is the metric of the spacetime $\mathcal{M}$ and $\phi$ is the scalar field. We say that the spacetime $(\mathcal{M},g)$ is spherically symmetric if the metric $g$ can induce a Lorentzian metric $g_{\mathcal{Q}}$ on the quotient manifold $\mathcal{Q}: = \mathcal{M}/SO(3)$ by the projection $p:\mathcal{M}\rightarrow\mathcal{Q}$. Let the boundary $\Gamma\subset \mathcal{Q}$ be the set of fixed points of the $SO(3)$ action, referred to as the center (axis) of $\mathcal{Q}$. Assume that on the quotient manifold $\mathcal{Q}$, there exist the so-called global double-null coordinates $(u,v)$ such that $g_{\mathcal{Q}}$ takes the form of \begin{equation*}
    g_{\mathcal{Q}} = -\frac{1}{2}\Omega^{2}(u,v)du\otimes dv-\frac{1}{2}\Omega^{2}(u,v)dv\otimes du,
\end{equation*}
where $\Omega^{2}$ is called the lapse function. For each point $q$ on $\mathcal{Q}$, we can define the associated area radius function by \begin{equation*}
    r(q): = \sqrt{\frac{Area(p^{-1}(q))}{4\pi}}.
\end{equation*}
Then, we can express the metric $g$ on $\mathcal{M}$ by
\begin{equation*}
g = -\frac{1}{2}\Omega^{2}du\otimes dv-\frac{1}{2}\Omega^{2}dv\otimes du+r^{2}(u,v)d\sigma^{2},
\end{equation*}
where $d\sigma^{2} = d\theta^{2}+\sin^{2}(\theta)d\varphi^{2}$ is the standard metric on a unit sphere. The coordinate system $(u,v,\theta,\varphi)$ is referred to as the double-null coordinates here.

Under the spherically symmetric assumption, the Einstein-scalar field equations~\eqref{ric equation}--\eqref{scalar field equation} can be reduced to\begin{align}
\partial_{u}\left(\frac{r_{u}}{\Omega^{2}}\right) &= -\frac{r}{\Omega^{2}}(\partial_{u}\phi)^{2},\label{eq:u-Ray equation}\\
\partial_{v}\left(\frac{r_{v}}{\Omega^{2}}\right)& = -\frac{r}{\Omega^{2}}(\partial_{v}\phi)^{2},\label{eq:v-Ray equation}\\
r\partial_{u}\partial_{v}r+\partial_{u}r\partial_{v}r &= -\frac{1}{4}\Omega^{2},\label{eq:wave equation for r}\\
\partial_{u}\partial_{v}\log(\Omega^{2})& = \frac{\Omega^{2}}{2r^{2}}+\frac{2r_{u}r_{v}}{r^{2}}-2\partial_{u}\phi\partial_{v}\phi,\label{eq:wave equation for omega}\\
r\partial_{u}\partial_{v}\phi +r_{u}\partial_{v}\phi+r_{v}\partial_{u}\phi &= 0\label{eq:wave equation for phi}.  
\end{align} 
The transport equations~\eqref{eq:u-Ray equation}--\eqref{eq:v-Ray equation} are called the Raychaudhuri equations, and the equations~\eqref{eq:wave equation for r}--\eqref{eq:wave equation for phi} are the wave equations for $r$, $\Omega^{2}$, and $\phi$, respectively. Note that by a simple algebraic computation, one can deduce the wave equation~\eqref{eq:wave equation for omega} for $\Omega$ from the equations~\eqref{eq:u-Ray equation}--\eqref{eq:v-Ray equation}, \eqref{eq:wave equation for r} and \eqref{eq:wave equation for phi}. We will not use the wave equation for $\Omega^{2}$ in the later study.

It is often useful to work with the Hawking mass $m$ and the mass ratio $\mu$, which are defined as follows:\begin{align}
m: &= \frac{r}{2}\left(1-g(\nabla r,\nabla r)\right),\\
\mu:&= \frac{2m}{r}.
\end{align}
It is straightforward to derive the transport equations for $m$ and $\mu$ from~\eqref{eq:u-Ray equation}--\eqref{eq:wave equation for phi}:\begin{align}
\partial_{u}m +\frac{r}{r_{u}}\left(\partial_{u}\phi\right)^{2}m &= \frac{1}{2}\frac{r^{2}}{r_{u}}\left(\partial_{u}\phi\right)^{2},\label{eq:u-equ for m}\\
\partial_{v}m+\frac{r}{r_{v}}\left(\partial_{v}\phi\right)^{2}m &= \frac{1}{2}\frac{r^{2}}{r_{v}}\left(\partial_{v}\phi\right)^{2},\label{eq:v-equ for m}\\
\partial_{u}\mu+\left(\frac{r_{u}}{r}+\frac{r}{r_{u}}\left(\partial_{u}\phi\right)^{2}\right)\mu&=\frac{r}{r_{u}}\left(\partial_{u}\phi\right)^{2},\label{eq:u-equ for mu}\\
\partial_{v}\mu+\left(\frac{r_{v}}{r}+\frac{r}{r_{v}}\left(\partial_{v}\phi\right)^{2}\right)\mu&=\frac{r}{r_{v}}\left(\partial_{v}\phi\right)^{2}.\label{eq:v-equ for mu}
\end{align}
For the Einstein-scalar field equations, instead of solving for $(g,\phi)$, it is equivalent to consider $(r,\phi,m)$. We call this the $(r,\phi,m)$-formulation of the Einstein-scalar field equations. Under this formulation, it suffices to consider equations~\eqref{eq:wave equation for r}, \eqref{eq:wave equation for phi}, and \eqref{eq:u-equ for m}--\eqref{eq:v-equ for m}.

\subsection{Christodoulou's naked singularity spacetime}
\label{sec: pre on k naked singularity sapcetime}
In this section, we briefly review the $k$-self-similar naked singularity spacetime solution to~\eqref{ric equation}--\eqref{scalar field equation}, originally constructed by Christodoulou \cite{chris94} under Bondi coordinates. For the convenience of the later discussion, we present the construction and its key properties in double-null coordinates. Further details can be found in Appendix B of \cite{singh1}.

Recall that under the double-null coordinates $(\hat{u},\hat{v},\theta,\varphi)$ with \begin{equation*}
    g = -\frac{1}{2}\Omega^{2}d\hat{u}\otimes d\hat{v}-\frac{1}{2}\Omega^{2}d\hat{v}\otimes d\hat{u}+r^{2}d\sigma^{2},
\end{equation*}
the Einstein-scalar field equations~\eqref{ric equation}--\eqref{scalar field equation} have the scaling and translation symmetry. More precisely, if $\left(r,\Omega^{2},\phi\right)$ is a solution to~\eqref{eq:u-Ray equation}--\eqref{eq:wave equation for phi}, then $\left(r_{a},\Omega_{a}^{2},\phi_{a,b}\right)$\begin{equation*}
    r_{a}(\hat{u},\hat{v}): = ar\left(\frac{\hat{u}}{a},\frac{\hat{v}}{a}\right),\quad \Omega_{a}^{2}(u,v): = \Omega^{2}\left(\frac{\hat{u}}{a},\frac{\hat{v}}{a}\right),\quad \phi_{a,b}(\hat{u},\hat{v}): = \phi\left(\frac{\hat{u}}{a},\frac{\hat{v}}{a}\right)+b
\end{equation*}
 is also a solution to~\eqref{ric equation}--\eqref{scalar field equation}. Taking $b = -k\log a$ and assuming the solution spacetime $(\mathcal{M},g,\phi)$ is invariant under scaling and translation \begin{equation*}
    r(\hat{u},\hat{v}) = ar\left(\frac{\hat{u}}{a},\frac{\hat{v}}{a}\right),\quad\Omega^{2}(\hat{u},\hat{v}) = \Omega^{2}\left(\frac{\hat{u}}{a},\frac{\hat{v}}{a}\right),\quad \phi(\hat{u},\hat{v}) = \phi\left(\frac{\hat{u}}{a},\frac{\hat{v}}{a}\right)+k\log a,
\end{equation*}
then the resulting solution $\left(r,\Omega^{2},\phi\right)$ to~\eqref{ric equation}--\eqref{scalar field equation}, known as the $k$-self-similar solution, takes the form of \begin{equation}
    r\left(\hat{u},\hat{v}\right) = \left(-\hat{u}\right)\mr{r}\left(\hat{z}\right),\quad \Omega^{2}\left(\hat{u},\hat{v}\right) = \mr{\Omega}^{2}\left(\hat{z}\right),\quad \phi(\hat{u},\hat{v}) = \mr{\phi}\left(\hat{z}\right)-k\log(-\hat{u}),\label{eq: k self similar form of the equations}
\end{equation}
where $\hat{z}$ is defined to be \begin{equation*}
    \hat{z}:= -\frac{\hat{v}}{\hat{u}}.
\end{equation*}

Alternatively and more geometrically, for the $k$-self-similarity, we have the following definition:
\begin{definition}[\textup{\cite{singh2}}]
Fix a parameter $k\in\mathbb{R}$. A spherically symmetric solution $\left(\mathcal{M},g,\phi\right)$ to the Einstein scalar field system is called $k$-self-similar if there exists a conformal Killing vector field $K$ such that \begin{equation*}
    \mathcal{L}_{K}g = 2g,\quad \mathcal{L}_{K}\phi = -k.
\end{equation*}
We define the scaling origin $\mathcal{O}$ as the point where $K$ vanishes.
\end{definition}
We can make the gauge choice such that the conformal Killing vector field $K$ takes the form of \begin{equation*}
    K = \hat{u}\partial_{\hat{u}}+\hat{v}\partial_{\hat{v}}.
\end{equation*}
Then, it is not hard to see that the coordinate $\hat{z}$ can be used to parametrize the integral curve of $K$. Moreover, the scaling origin $\mathcal{O}$ corresponds to $(0,0)$ under such a gauge choice. In addition, we can further fix the gauge freedom such that the center $\Gamma$ corresponds to \begin{equation*}
    \Gamma = \{\hat{z} = -1\}.
\end{equation*}
Then the region under consideration will become \begin{equation*}
    \{u<0,\hat{z}\geq-1\}.
\end{equation*}
The following proposition states a direct consequence of the $k$-self-similarity, in the spirit of~\eqref{eq: k self similar form of the equations}.
\begin{proposition}
    Assume that the spacetime solution $\left(\mathcal{M},g_{k},\phi_{k}\right)$ is a $k$-self-similar solution to~\eqref{ric equation}--\eqref{scalar field equation}, then $(r_{k},\Omega_{k}^{2},\phi_{k})$ takes the form of \begin{equation*}
        r_{k} = \left(-\hat{u}\right)\mr{r}(\hat{z}),\quad \Omega_{k}^{2} = \mr{\Omega}^{2}(\hat{z}),\quad \phi_{k} = \mr{\phi}(\hat{z})-k\log\left(-\hat{u}\right).
    \end{equation*}
    Denoted by $\nu_{k}$ and $\lambda_{k}$ the $\hat{u}$-derivative and $\hat{v}$-derivative of $r_{k}$ respectively, we have \begin{equation*}
        \nu_{k}(\hat{u},\hat{v})= -\mr{r}(\hat{z})+\hat{z}\frac{d\mr{r}}{d\hat{z}} =:\mr{\nu}(\hat{z}),\quad \lambda_{k}(\hat{u},\hat{v}) = \frac{d\mr{r}}{d\hat{z}}(\hat{z}) =:\mr{\lambda}(\hat{z}).
    \end{equation*}
    Moreover, for the Hawking mass $m_{k}$ and the mass ratio $\mu_{k}$, we have \begin{equation*}
        m_{k}(\hat{u},\hat{v}) = (-\hat{u})\mr{m}(\hat{z}),\quad \mu_{k}(\hat{u},\hat{v}) = \mr{\mu}(\hat{z}).
    \end{equation*}
    \label{prop: consequence of the k self similarity}
\end{proposition}
Substituting the form in Proposition~\ref{prop: consequence of the k self similarity} into~\eqref{eq:u-Ray equation}--\eqref{eq:wave equation for phi}, one can derive the ODE system for $(\mr{r},\mr{\Omega},\mr{\phi})$. In the seminal work of Christodoulou \cite{chris94}, a careful study of the ODE system has been conducted, and a naked singularity spacetime to the Einstein-scalar field equations was constructed within the range $k^{2}\in\left(0,\frac{1}{3}\right)$. Under such a construction, the scaling origin $\mathcal{O}$ corresponds to the naked singularity. The past light cone of the naked singularity $\{\hat{v} = 0\}$ is called the singular horizon. The region \begin{equation*}
    \mathcal{Q}^{(ex)}: = \{\hat{u}<0,\hat{z}\geq0\}
\end{equation*} 
is called the naked singularity spacetime exterior; the region \begin{equation*}
    \mathcal{Q}^{(in)}: = \{\hat{u}<0,-1\leq\hat{z}\leq0\}
\end{equation*}
is called the naked singularity spacetime interior. We denote by $\mathcal{Q}: =\mathcal{Q}^{(in)}\cup\mathcal{Q}^{(ex)}$ the whole region of the naked singularity spacetime. We denote the hypersurface $\{u = const\}$ by $\Sigma_{u_{0}}$. Since we particularly care about the interior region in this paper, we denote the hypersurface $\{u=const\}\cap \mathcal{Q}^{(in)}$ by $\Sigma_{u_{0}}^{(in)}$. We will not pursue the detailed construction here. In the following few subsections, we will list some useful properties of the naked singularity spacetime.

\subsubsection{Renormalized gauge}
\label{sec: renormalized gauge}
First, although the scaling conformal Killing vector field $K$ takes a simple form $K = \hat{u}\partial_{\hat{u}}+\hat{v}\partial_{\hat{v}}$ under the double-null coordinates $\left(\hat{u},\hat{v},\theta,\varphi\right)$, the ODE analysis shows that various quantities arising from the regular data at the center $\Gamma$ will become singular when approaching the singular horizon $\{\hat{v} = 0\}$. In particular, $\frac{d\mr{\phi}}{d\hat{z}}$, $\mr{\Omega}^{2}(\hat{z})$, and $\partial_{\hat{v}}r(\hat{z})$ will behave like $\vert\hat{z}\vert^{-k^{2}}$ when $\hat{z}\rightarrow 0$. Hence, the spacetime metric is not well-defined on the singular horizon. However, such a behavior is only a coordinate singularity, and can be compensated for by a suitable change of the coordinates.

Motivated by extending the spacetime geometry continuously to the whole region $\mathcal{Q}$, we define \begin{equation*}
    u = \hat{u},\quad (-v) = \left(-\hat{v}\right)^{q_{k}},\quad (-z) = \left(-\hat{z}\right)^{q_{k}},
\end{equation*}
where $q_{k}: = 1-k^{2}$ will be frequently used in the later analysis. Let $p_{k}: = \frac{1}{q_{k}}$. Then, we have \begin{equation*}
    \partial_{\hat{u}} = \partial_{u},\quad \partial_{\hat{v}} = q_{k}\left(-v\right)^{-k^{2}p_{k}}\partial_{v}.
\end{equation*}
Proposition \ref{prop: consequence of the k self similarity} can be easily adapted to this renormalized gauge. Without any confusion, and to avoid the tedious notation, we reuse the notations $\nu_{k}$ and $\lambda_{k}$ introduced in Proposition \ref{prop: consequence of the k self similarity} to denote the derivatives $\partial_{u}r$ and $\partial_{v}r$, respectively. To make things more precise, we have the following proposition: \begin{proposition}
    Assume that the spacetime solution $\left(\mathcal{M},g_{k},\phi_{k}\right)$ is a $k$-self-similar solution to \eqref{ric equation}-\eqref{scalar field equation} under the renormalized gauge $(u,v,\theta,\varphi)$ with \begin{equation*}
        g_{k}=-\frac{1}{2}\Omega_{k}^{2}du\otimes dv-\frac{1}{2}\Omega_{k}^{2}dv\otimes du+r^{2}d\sigma^{2}.
    \end{equation*}
    Then $\left(r_{k},\Omega_{k}^{2},\phi_{k}\right)$ takes the form of \begin{equation*}
        r_{k}(u,v) = (-u)\mr{r}(z),\quad \Omega_{k}^{2}(u,v) = (-u)^{k^{2}}\mr{\Omega}^{2}(z),\quad \phi_{k}(u,v) = \mr{\phi}(z)-k\log(-u).
    \end{equation*}
    Moreover, for $\nu_{k},\lambda_{k},m_{k},\mu_{k}$, we have \begin{align*}
        \nu_{k}(u,v) = \mr{\nu}(z) = -\mr{r}(z)+q_{k}z\frac{d\mr{r}}{dz},\quad \lambda_{k}(u,v) =p_{k}(-v)^{p_{k}k^{2}}\mr{\lambda}(z) = (-u)^{k^{2}}\frac{d\mr{r}}{dz}. 
    \end{align*}
    Since the Hawking mass $m_{k}$ and the mass ratio $\mu_{k}$ are gauge-invariant quantities, we have \begin{equation*}
        m_{k}(u,v) = (-u)\mr{m}(z),\quad \mu_{k}(u,v) = \mr{\mu}(z).
    \end{equation*}
\end{proposition}

An important feature of the naked singularity spacetime $\left(\mathcal{M},g,\phi\right)$ is that the metric $g$ and the scalar field $\phi$ are of low regularity. Using the subscript $k$ to denote the $k$-self-similar naked singularity spacetime, the following proposition collects the regularity properties and some useful quantitative behaviors of the spacetime geometry and the scalar field, the proof of which can be found in Appendix B of \cite{singh1}. \textbf{Since the Einstein-scalar field equations are translation and scaling invariant, to fix the gauge freedom of $(r_{k},\phi_{k},m_{p})$, throughout this paper, we assume $r(u,0) = (-u)$ and $\phi_{k}(u,0) = -k\log(-u)$.}
\begin{proposition}[\textup{\cite{chris94,singh2}}]
\label{prop: estimate on the exact k self similar spacetime}
Under the renormalized double-null coordinates $\left(u,v,\theta,\varphi\right)$, for $0<k^{2}<\frac{1}{3}$, the $k$-self-similar naked singularity spacetime $\left(\mathcal{M},g_{k},\phi_{k}\right)$ constructed in \cite{chris94} are of finite regularity nature. More precisely, we have \begin{enumerate}
    \item[(1)] The spacetime $\left(\mathcal{M},g_{k},\phi_{k}\right)$ is smooth in the region $$\{-1\leq u<0,v<0\}\cup\{-1\leq u<0,v>0\}.$$ In other words, the $k$-self-similar spacetime is smooth away from the singular horizon.
    \item[(2)] The low regularity nature of the spacetime only arises at the singular horizon $\{v = 0\}$. In the neighborhood $\{-1\leq u<0,\ -1\leq z\leq 1\}$ of the singular horizon, we have \begin{equation*}
        \mr{m},\ \mr{\Omega},\ \mr{\phi}\in C^{1,p_{k}k^{2}}\left([-1,1]\right),\quad \mr{r}\in C^{2,p_{k}k^{2}}([-1,1]),
    \end{equation*}
    with the estimates \begin{equation*}
        \left\vert\vert z\vert^{j-1-p_{k}k^{2}}\partial_{z}^{j+1}\mr{r}\right\vert+\left\vert\vert z\vert^{j-1-p_{k}k^{2}}\partial_{z}^{j}\mr{\Omega}\right\vert+\left\vert\vert z\vert^{j-1-p_{k}k^{2}}\partial_{z}^{j}\mr{m}\right\vert+\left\vert\vert z\vert^{j-1-p_{k}k^{2}}\partial_{z}^{j}\mr{\phi}\right\vert\lesssim_{k}1,\quad j\geq 2.
    \end{equation*}
    Moreover, the above finite regularity is sharp in the sense that \begin{align*}
        \frac{d\mr{\phi}}{dz}\sim\frac{1}{k},\quad 
        \frac{d^{2}\mr{\phi}}{dz^{2}}\sim \vert z\vert^{-1+p_{k}k^{2}},\quad z\rightarrow 0.
    \end{align*}
    \item[(3)] In the region $\mathcal{Q}^{(in)}$, we have the following estimates on the geometric quantities: \begin{align*}
        &\mr{r}(-1) = 0,\quad  \mr{r}\approx 1-\left\vert z\right\vert,\quad (-\nu_{k})\approx 1.\quad \lambda_{k}\approx (-u)^{k^{2}},\\&
        \Omega_{k}^{2}\approx (-u)^{k^{2}},\quad 0<\frac{\mr{m}}{\mr{r}^{3}}\lesssim k^{2},\quad 0<\frac{\mr{\mu}}{\mr{r}^{2}}\lesssim k^{2}.
    \end{align*}
    \item[(4)] In the region $\mathcal{Q}^{(in)}$, we have the following estimates on the quantities related to the scalar field: \begin{align*}
    \mr{\phi}^{\prime}(z = 0) = \frac{C}{k}, \quad \left\Vert\mr{\phi}^{\prime}\right\Vert_{L^{\infty}([-1,-\frac{1}{2}])}\lesssim k^{2},\quad 
    \left\Vert\vert z\vert^{\epsilon}\mr{\phi}^{\prime}\right\Vert_{L^{\infty}([-1,0])}\lesssim \frac{k}{\epsilon},\quad k\ll\epsilon<1.
    \end{align*}
    \item[(4)] As a consequence of Property $(4)$ and our gauge choice of $\phi_{k}$, we have \begin{equation*}
        \left\vert\mr{\phi}\right\vert\lesssim \frac{k}{\epsilon}\left\vert z\right\vert^{1-\epsilon}.
    \end{equation*}
    \item[(5)] For the higher order derivatives, we have \begin{align*}
    &\vert z\vert^{1+\epsilon}\left\vert \mr{\phi}^{\prime\prime}\right\vert\lesssim \frac{k^{2}}{\epsilon},\quad \vert z\vert^{\epsilon}\left\vert\frac{d^{2}\mr{r}}{dz^{2}}\right\vert+\vert z\vert^{1+\epsilon}\left\vert\frac{d^{3}\mr{r}}{dz^{3}} \right\vert\lesssim \frac{k^{2}}{\epsilon^{2}},\\&
    \left\vert z\right\vert^{\epsilon}\left\vert\frac{\partial_{z}\mr{\mu}}{(1+z)}\right\vert+\vert z\vert^{1+\epsilon}\left\vert\partial_{z}^{2}\mr{\mu}\right\vert\lesssim \frac{k^{2}}{\epsilon^{2}},\quad \vert z\vert^{\epsilon}\left\vert\frac{\partial_{z}\mr{m}}{(z+1)^{2}}\right\vert+\vert z\vert^{1+\epsilon}\left\vert\frac{\partial_{z}^{2}\mr{m}}{z+1}\right\vert\lesssim \frac{k^{2}}{\epsilon^{2}}.
    \end{align*}
\end{enumerate}
\end{proposition}

\subsection{Self-similar coordinates}
Motivated by the $k$-self-similarity and the renormalized gauge introduced in Section \ref{sec: pre on k naked singularity sapcetime}, we define the self-similar coordinates $(s,z)$ to be \begin{equation}
    s: = -\log(-u),\quad z = \frac{v}{(-u)^{q_{k}}}.
\end{equation}
Then under the self-similar coordinates, the center $\Gamma$ in the quotient spacetime $\mathcal{Q}$ corresponds to the curve $\{z = -1\}$, the singular horizon $\mathcal{H}$ corresponds to $\{z = 0\}$, and the interior region $\mathcal{Q}^{(in)}$ corresponds to $\{-1\leq z\leq0\}$. Moreover, a simple coordinate transform gives that \begin{align}
    &\partial_{u} = e^{s}\partial_{s}+q_{k}ze^{s}\partial_{z},\quad \partial_{v} = e^{q_{k}s}\partial_{z},\\&
    \partial_{s} =(-u)\partial_{u}+q_{k}(-v)\partial_{v},\quad \partial_{z} = (-u)^{q_{k}}\partial_{v}.
\end{align}

\subsection{Local well-posedness of the Einstein-scalar field equations}
\label{sec: recall the local well-posedness}
In this section, we recall the local well-posedness of the mixed characteristic-timelike initial value problem to the Einstein-scalar field equations established in \cite{chris93}.

Let $BV(U)$ and $AC(U)$ be the spaces of bounded variation functions and absolutely continuous functions on some set $U$, respectively. We denote the norm of total variation by $T.V.$. Let $\mathcal{Q}_{u_{0},v_{0}}$ be the region $\mathcal{Q}\cap\{u\leq u_{0},v\leq v_{0}\}$, and let $\lambda = \partial_{v}r$ and $\nu = \partial_{u}r$, respectively. Moreover, we fix the gauge choice of the axis to be $\Gamma = \{(-v) = (-u)^{q_{k}}\}$.

In \cite{chris93}, Christodoulou established the local well-posedness results on the spherically symmetric Einstein-scalar field equations for the initial data with BV derivatives and $C^{1}$ derivatives, respectively. First, we recall the following two definitions on the $BV$ and $C^{1}$ solutions.
\begin{definition}
    A BV solution  in a domain $\mathcal{Q}_{u_{0},v_{0}}$ consists of a triple $(r,\phi,m)$  solving \eqref{eq:u-Ray equation}-\eqref{eq:wave equation for phi}, satisfying the following:
    \begin{enumerate}
        \item[(1)]$ 
            \sup_{\mathcal{Q}_{u_{0},v_{0}}}(-\nu)<\infty,\quad \sup_{\mathcal{Q}_{u_{0},v_{0}}}\lambda^{-1}<\infty$.
            \item[(2)] $\lambda\in BV(\Sigma_{u}\cap\mathcal{Q}_{u_{0},v_{0}})$ uniformly in $u$, and $\nu\in BV(\underline{\Sigma}_{v}\cap \mathcal{Q}_{u_{0},v_{0}})$ uniformly in $v$.
            \item[(3)] $\phi(u,\cdot)\in AC(\Sigma_{u}\cap\mathcal{Q}_{u_{0},v_{0}})$ and $\phi(\cdot,v)\in AC(\underline{\Sigma}_{v}\cap\mathcal{Q}_{u_{0},v_{0}})$, with $T.V.$ norm uniform in $u$ and $v$, respectively.
            \item[(4)] The derivatives $\partial_{v}(r\phi)\in BV(\Sigma_{u}\cap\mathcal{Q}_{u_{0},v_{0}})$ uniformly and $\partial_{u}(r\phi)\in BV(\underline{\Sigma}_{v}\cap\mathcal{Q}_{u_{0},v_{0}})$ uniformly.
            \item[(5)] The solution is regular on the axis, in the sense that for any $(u,v)\in\Gamma$ \begin{align*}
                &r(u,v) = r\phi(u,v) = m(u,v) = 0,\\&
                \lim_{\delta\rightarrow0}\left(\lambda+p_{k}\vert u\vert^{k^{2}}\nu\right)(u,v+\delta) = 0,\quad \lim_{\delta\rightarrow0}\left(\partial_{v}(r\phi)+p_{k}\vert u\vert^{k^{2}}\partial_{u}(r\phi)\right)(u,v+\delta) = 0.
            \end{align*}
    \end{enumerate}
\end{definition}
\begin{definition}
    A $C^{1}$ solution in a domain $\mathcal{Q}_{u_{0},v_{0}}$ consists of a triple $(r,\phi,m)$ solving \eqref{eq:u-Ray equation}-\eqref{eq:wave equation for phi}, satisfying the following:\begin{enumerate}
        \item[(1)] $\sup_{\mathcal{Q}_{u_{0},v_{0}}}(-\nu)<\infty,\quad \sup_{\mathcal{Q}_{u_{0},v_{0}}}\lambda^{-1}<\infty$.
        \item[(2)] $\lambda$, $\nu$, $\partial_{u}(r\phi)$, and $\partial_{v}(r\phi)$ are in $C^{1}(\mathcal{Q}_{u_{0},v_{0}})$.
        \item[(3)] The solution is regular on the axis, in the sense that for any $(u,v)\in\Gamma$, we have \begin{align*}
            &r(u,v) = r\phi(u,v) = m(u,v) = 0,\\&
            \lim_{\delta\rightarrow 0}\left(\lambda+p_{k}\vert u\vert^{k^{2}}\nu\right)(u,v+\delta) = 0,\quad \lim_{\delta\rightarrow0}\left(\partial_{v}(r\phi)+p_{k}\vert u\vert^{k^{2}}\partial_{u}(r\phi)\right)(u,v+\delta) = 0,\\&
            \lim_{\delta\rightarrow 0}\left(\lambda+p_{k}\vert u\vert^{k^{2}}\nu\right)(u-\delta,v) = 0,\quad \lim_{\delta\rightarrow 0}\left(\partial_{v}(r\phi)+p_{k}\vert u\vert^{k^{2}}\partial_{u}(r\phi)\right)(u-\delta,v) = 0.
        \end{align*}
    \end{enumerate}
\end{definition}
Given the gauge choice of the center $\Gamma = \{(-v) = (-u)^{q_{k}}\}$ and the gauge choice of $r(-1,\cdot) = r_{k}(-1,\cdot)$, the only free initial data is $\partial_{v}(r\phi)$. Using the regularity of the center and the transport equation \eqref{eq:u-equ for m} for $m$, we can uniquely determine $r\phi$ and $m$ on the initial hypersurface $\Sigma_{-1}$. One corollary of the local well-posedness result established in \cite{chris93} adapting to our setting is that, for $\partial_{v}(r\phi)_{\Sigma_{-1}}\in BV$, there exist a domain of the form $\{-1\leq u\leq-1+\delta,\ -(-u)^{q_{k}}\leq v\}$ and a $BV$ solution $(r,\phi,m)$ to \eqref{eq:u-Ray equation}-\eqref{eq:wave equation for phi} in this domain solving the given initial data. Moreover, if the initial data $\partial_{v}(r\phi)|_{\Sigma_{-1}}\in C^{1}$, then the solution is also a $C^{1}$ solution.

If we further adapt to the setting of the $k$-self-similar naked singularity spacetime, then given the initial data $\partial_{v}(r\phi)|_{\Sigma_{-1}} = \partial_{v}(r_{k}\phi_{k})(-1,\cdot)$, the solution to \eqref{eq:u-Ray equation}-\eqref{eq:wave equation for phi} achieving this initial data is the $k$-self-similar naked singularity spacetime. Hence, to consider the perturbation of the initial data, it suffices to perturb $\partial_{v}(r\phi)|_{\Sigma_{-1}}$.

\subsection{Function spaces}
In this section, we introduce the following function spaces, which will be useful when imposing the interior initial data on the linearized system $\mathcal{P} = 0$ and in the nonlinear analysis for the interior region.
\begin{definition}
\label{def: definition of the localized holder space}
    Let $\alpha\in(1,2)$, $\gamma\in(-\frac{1}{2},\frac{1}{2})$, $\delta\in(0,1)$, and $I:=[-1,0]$. We can define \begin{align}
        \mathcal{C}^{\alpha,\delta}_{N}(I): &=\left\{f(z):I\rightarrow\mathbb{R}| f\in C_{z}^{N}(I\backslash\{0\})\cap C_{z}^{1}(I),\ \vert z\vert^{j-\alpha}\frac{d^{j}f}{dz^{j}}\in C_{z}^{0,\delta},\ 2\leq j\leq N\right \},
        \\
        H^{\gamma}_{N}(I):&=\left\{f(z): I\rightarrow \mathbb{R}|f\in C_{z}^{N}(I\backslash\{0\})\cap W_{z}^{1,2}(I),\ (-z)^{j-2+\gamma}\frac{d^{j}f}{dz^{j}}\in L^{2}(I),\ 2\leq j\leq N\right\},
    \end{align}
    associated with the norms \begin{align}
        \Vert f\Vert_{\mathcal{C}_{N}^{\alpha,\delta}}: &= \sum_{j = 0}^{1}\left\Vert\frac{d^{j}f}{dz^{j}}\right\Vert_{L^{\infty}(I)}+\sum_{j = 2}^{N}\left\Vert \vert z\vert^{j-\alpha}\frac{d^{j}f}{dz^{j}}\right\Vert_{C^{0,\delta}_{z}(I)},\\
        \left\Vert f\right\Vert_{H^{\gamma}_{N}}:&=\sum_{j = 0}^{1}\left\Vert\frac{d^{j}f}{dz^{j}}\right\Vert_{L^{2}(I)}+\sum_{j = 2}^{N}\left\Vert \vert z\vert^{j-2+\gamma}\frac{d^{j}f}{dz^{j}}\right\Vert_{L^{2}(I)},
    \end{align}
    where $\Vert\cdot\Vert_{C_{z}^{0,\delta}}$ means the standard Hölder norm. When $\delta = 0$, $C_{z}^{0,0}$ will be reduced to the standard continuous function space. We simplify the notation $\mathcal{C}_{N}^{\alpha,0}$ to be $\mathcal{C}_{N}^{\alpha}$.

    Finally, we define the function space 
    \begin{equation}
        H_{N,a}^{\gamma} = \left\{f:\mathbb{R}_{+}\times I\rightarrow\mathbb{R}|f(s,\cdot)\in H_{N}^{\gamma}(I),\ \sup_{s\geq0}e^{(1+a)s}\left\Vert f(s,\cdot)\right\Vert_{H_{N}^{\gamma}}<\infty\right\}.
    \end{equation}

    associated with the norm \begin{equation}
        \left\Vert f\right\Vert_{H_{N,a}^{\gamma}} = \sup_{s\geq0}e^{(1+a)s}\left\Vert f(s,\cdot)\right\Vert_{H_{N}^{\gamma}}.
    \end{equation}
\end{definition}
The following proposition is a direct consequence of Proposition~\ref{prop: estimate on the exact k self similar spacetime}

\begin{proposition}
The interior initial data $(r_{k},\Psi_{k},m_{k})|_{\Sigma_{-1}^{(in)}}$ of the $k$-self-similar naked singularity spacetime constructed by Christodoulou \cite{chris94} belongs to the function space $\mathcal{C}_{\infty}^{2-}\times\mathcal{C}_{\infty}^{p_{k},\delta}(I)\times\mathcal{C}_{\infty}^{p_{k},\delta}(I)$ for all $\delta\in(0,1-O(k))$.
\end{proposition}
An easy example of a function $f$ in the function space $\mathcal{C}_{N}^{\alpha,\delta}$ is $f = (-z)^{\alpha}$, where $f$ is smooth away from the point $z = 0$ and is of finite Hölder regularity at $z = 0$. The following decomposition lemma illustrates that the motivation of the function space $\mathcal{C}_{N}^{\alpha,\delta}$ is to generalize the function $(-z)^{\alpha}$.

\begin{lemma}{(\cite{singh2})}
\label{lemma: decomposition of the function spaces}
    For any $f\in\mathcal{C}_{N}^{\alpha,\delta}$, there exist a unique constant $c_{0}$ and a unique function $f_{1}\in\mathcal{C}_{N}^{\alpha+\delta^{\prime}}$ with $\delta^{\prime}<\min\{\delta,2-\alpha\}$ such that \begin{equation}
        f = c_{0}\vert z\vert^{\alpha}+f_{1}.\label{eq: decompose holder function}
    \end{equation}
    Moreover, we have the estimate \begin{equation}
        \left\vert c_{0}\right\vert+\left\Vert f_{1}\right\Vert_{\mathcal{C}^{\alpha+\delta^{\prime}}}\lesssim \left\Vert f\right\Vert_{\mathcal{C}_{N}^{\alpha}}.
        \label{eq: estimate on the decomposition}
    \end{equation}
    For $\alpha = \frac{1}{1-k^{2}}$, the function $\vert z\vert^{\alpha}$ in \eqref{eq: decompose holder function} can be replaced by $\mr{\phi}$ and the estimate \eqref{eq: estimate on the decomposition} continues to hold.
\end{lemma}
\begin{proof}
The proof can be found in Section 2.9 of \cite{singh2}.
\end{proof}

For the function spaces $\mathcal{C}_{N}^{\alpha}$ and $H_{N}^{\gamma}$, we can establish the following embedding lemma.
\begin{lemma}
\label{lemma: sobolev embedding}
For any $\gamma\in(-\frac{1}{2},\frac{1}{2})$, $0\leq\delta<1$, and $\alpha\in(1,\frac{3}{2}-\gamma]$ we have that $ \mathcal{C}_{N}^{\alpha}\hookrightarrow H_{N+1}^{\gamma}$.
\end{lemma}

\begin{proof}
    It suffices to show that for any $f\in H_{N+1}^{\gamma}$, we have that $$\left\Vert f\right\Vert_{\mathcal{C}_{N}^{\alpha}}\lesssim \left\Vert f\right\Vert_{H_{N+1}^{\gamma}}$$ for any $\alpha\in(1,\frac{3}{2}-\gamma]$.

    By the Sobolev embedding on $[-1,-\frac{1}{2}]$, we have that \begin{equation*}
    \sum_{j = 0}^{N}\left\Vert \frac{d^{j}f}{dz^{j}}\right\Vert_{L^{\infty}[-1,-\frac{1}{2}]}\lesssim \left\Vert f\right\Vert_{H_{N+1}^{\gamma}}.
\end{equation*}
Then for any $2\leq j\leq N$, by the fundamental theorem of calculus and Hölder inequality, we have that \begin{equation*}
    \left\vert \frac{d^{j}f}{dz^{j}}(z)-\frac{d^{j}f}{dz^{j}}(-\frac{1}{2})\right\vert\leq\left(\int_{-\frac{1}{2}}^{z}(-\widetilde{z})^{2\gamma+2j-2}\left(\frac{d^{j+1}f}{dz^{j+1}}\right)^{2}\right)\left(\int_{-\frac{1}{2}}^{z}(-\widetilde{z})^{-2\gamma-2j+2}\right)^{\frac{1}{2}}.
\end{equation*}
Hence, we have \begin{equation*}
    \left\vert(-z)^{j-\alpha}\frac{d^{j}f}{dz^{j}}\right\vert\leq \left\vert(-z)^{\gamma+j-\frac{3}{2}}\frac{d^{j}f}{dz^{j}}\right\vert\lesssim \left\Vert f\right\Vert_{H_{N+1}^{\gamma}}.
\end{equation*}
This concludes the proof.
\end{proof}
The function spaces defined in Definition \ref{def: definition of the localized holder space} will be used to describe the initial interior perturbations. For perturbations on the whole initial hypersurface $\{u = -1,\ v\geq -1\}$, we have the following definition, capturing the finite localized Hölder regularity in the exterior region as well as the asymptotic flatness of the initial data.
\begin{definition}
\label{def: localized Holder regularity general}
    Let $\alpha\in(p_{k},\frac{3}{2})$, and $\widetilde{I} = [-1,1]$. We define the function spaces \begin{align}
        &\widetilde{\mathcal{C}}_{N}^{\alpha,\delta}: = \{f(z):[-1,\infty)\rightarrow\mathbb{R}|f\in C_{z}^{N}(\widetilde{I}\backslash\{0\})\cap C_{z}^{1}(\widetilde{I}),\ \vert z\vert^{j-\alpha}\frac{d^{j}f}{dz^{j}}\in C_{z}^{0,\delta},\ 2\leq j\leq N\},\\&
        AF: = \{f:[-1,\infty)\rightarrow \mathbb{R}|f\in C^{2}([1,\infty)), \sup_{v\in[1,\infty)}\left\vert\partial_{v}^{i}f\right\vert\lesssim v^{-1-i},\ i = 1,2\}.
    \end{align}
    Let $$\bar{\mathcal{C}}_{N}^{\alpha,\delta}: =\widetilde{\mathcal{C}}_{N}^{\alpha,\delta}\cap AF$$ associated with the norm: \begin{equation}
        \left\Vert f\right\Vert_{\bar{\mathcal C}_{N}^{\alpha,\delta}}: = \sum_{j = 0}^{1}\left\Vert\frac{d^{j}f}{dz^{j}}\right\Vert_{C^{0,\delta}([-1,1])}+\sum_{j = 2}^{N}\left\Vert\vert z\vert^{j-\alpha}\frac{d^{j}f}{dz^{j}}\right\Vert_{C^{0,\delta}([-1,1])}+\sum_{j = 1}^{2}\left\Vert v^{1+j}\partial_{v}^{j}f\right\Vert_{L^{\infty}([1,\infty))}.
    \end{equation}
\end{definition}
\begin{remark}
\label{rmk: localwellposedness in localized holder}
    Fixing the gauge choice of $r$ on $\Sigma_{-1}$ and imposing the scalar field initial data $\phi|_{\Sigma_{-1}}\in\bar{\mathcal C}_{N}^{\alpha,\delta}$, the local well-posedness of the Einstein-scalar field equations in this class follows similarly as in~\cite{chris93,lukoh_bvscatter}. We omit the proof in this paper.
\end{remark}
We prove the following proposition, showing that the space $\bar{\mathcal{C}}^{\alpha,\delta}_{N}$ is a Banach space.
\begin{proposition}
    The function space $\bar{\mathcal{C}}_{N}^{\alpha,\delta}$ is a Banach space for any $N\geq 2$, $\alpha\in(p_{k},\frac{3}{2})$ and $\delta\in(0,1)$.
\end{proposition}
\begin{proof}
    Assume that $\{f_{n}\}$ is a Cauchy sequence in $\bar{\mathcal{C}}_{N}^{\alpha,\delta}$. Then for any non-zero $z\in\widetilde{I}$, there exists $\epsilon$ small enough such that $x\in\widetilde{I}\backslash[-\epsilon,\epsilon]$. Then we have that $f_{n}$ is a Cauchy sequence in $C^{N,\delta}(\widetilde{I}\backslash[-\epsilon,\epsilon])$, where $C^{N,\delta}$ is the standard Hölder space. Therefore, we have that $f_{n}\rightarrow f$ in $C^{N,\delta}$ for any $x\in\widetilde{I}\backslash\{0\}$. Then we have $f\in C_{z}^{N}(\widetilde{I}\backslash\{0\})$. Similarly, we can show that $f_{n}\rightarrow f$ in $C^{2}$ on any compact subset of $[1,\infty)$. It is straightforward to check that $f\in\bar{\mathcal{C}}_{N}^{\alpha,\delta}$ and $\left\Vert f_{n}-f\right\Vert_{\bar{\mathcal{C}}_{N}^{\alpha,\delta}}\rightarrow 0$ by definition. This concludes the proof.
    \end{proof}

The choice of the interior initial perturbation $(r_{p},\Psi_{p})|_{\Sigma_{-1}^{(in)}}$ for the fully nonlinear equations will be \begin{equation}
    (r_{p},\Psi_{p})|_{\Sigma_{-1}^{(in)}} = (0,\epsilon\Psi_{p}^{0}),\label{eq: the choice of the initial data}
\end{equation} 
where $\Psi_{p}^{0}\in\mathcal{C}_{N}^{p_{k},\delta}$, for any value of $\delta\in(0,1)$ such that $\phi_{k}|_{s = 0}\in\mathcal{C}^{p_{k},\delta}$ and $N\geq 6$.

\subsection{Formulation of the nonlinear evolution in the interior region}
\label{sec: setup of the nonlinear stability}
As we have discussed in Section \ref{sec: whole intro}, the nonlinear stability argument crucially relies on the linearized stability. The linearized Einstein-scalar field equations around the $k$-self-similar spacetime are a system of linear equations for the geometric quantities and the scalar field, with coefficients depending on the background. For the technical reason to make the coefficients in the linearized operator independent of $s$, instead of considering $(r,\phi,m)$ as the formulation of the Einstein-scalar field equations, we consider the $(r,\Psi,m)$-formulation, where $\Psi$ is the renormalized radiation field, taking the form of \begin{equation*}
    \Psi = r(\phi+k\log(-u)).
\end{equation*}
Then, we only need to consider the following four equations under the self-similar coordinates as the Einstein-scalar field equations \begin{align}
    \partial_{s}\partial_{z}r+q_{k}z\partial_{z}^{2}r+q_{k}\partial_{z}r = &\frac{\mu}{1-\mu}\frac{(\partial_{s}+q_{k}z\partial_{z})r\partial_{z}r}{r},\label{eq: wave equation for r under self-similar coordinates}\\
    \partial_{s}\partial_{z}\Psi+q_{k}z\partial_{z}^{2}\Psi+q_{k}\partial_{z}\Psi+k\partial_{z}r = &\frac{\mu}{1-\mu}\frac{(\partial_{s}+q_{k}z\partial_{z})r\partial_{z}r}{r^{2}}\Psi,\label{eq: wave equation for psi under self-similar coordinates}\\
    \partial_{z}m+\frac{r}{\partial_{z}r}\left(\partial_{z}\left(\frac{\Psi}{r}\right)\right)^{2}m =&\frac{1}{2}\frac{r^{2}}{\partial_{z}r}\left(\partial_{z}\left(\frac{\Psi}{r}\right)\right)^{2},\label{eq: u transport equation for m under self-similar coordinates}\\
    \left(\partial_{s}+q_{k}z\partial_{z}\right)m+\frac{r}{(\partial_{s}+q_{k}z\partial_{z})r}\left[(\partial_{s}+q_{k}z\partial_{z})\left(\frac{\Psi}{r}\right)+k\right]^{2}m=&\frac{1}{2}\frac{r^{2}}{(\partial_{s}+q_{k}z\partial_{z})r}\left[(\partial_{s}+q_{k}z\partial_{z})\left(\frac{\Psi}{r}\right)+k\right]^{2}.\label{eq: v-transport equation for m under self-similar coordinates}
\end{align}
By the local well-posedness of the Einstein-scalar field equations, we will only use equations \eqref{eq: wave equation for r under self-similar coordinates}-\eqref{eq: u transport equation for m under self-similar coordinates}. We can also systematically write the Einstein-scalar field equations for $(r,\Psi,m)$ as $\mathcal{F}(r,\Psi,m) = 0$, where $\mathcal{F}$ represents the equations \eqref{eq: wave equation for r under self-similar coordinates}-\eqref{eq: u transport equation for m under self-similar coordinates} and $(r,\Psi,m)$ is the solution to the Einstein-scalar field equations with the perturbed initial data $(r,\Psi)|_{\Sigma_{-1}^{(in)}} = (r_{k},\Psi_{k}+\Psi_{p})$ for $\Psi_{p}\in\mathcal{C}_{N}^{p_{k},\delta}$. Then we have \begin{equation*}
    0 = \mathcal{F}(r_{k}+r_{p},\Psi_{k}+\Psi_{p},m_{k}+m_{p})-\mathcal{F}(r_{k},\Psi_{k},m_{k}),
\end{equation*}
where \begin{equation*}
    r_{p}: = r-r_{k},\quad \Psi_{p}: = \Psi-\Psi_{k},\quad m_{p}: = m-m_{k}.
\end{equation*}
Taylor expanding the above expression with respect to $(r_{p},\Psi_{p},m_{p})$ and denoting the linear part as $\mathcal{P}$, we have \begin{equation*}
    \mathcal{P}(r_{p},\Psi_{p},m_{p}) = \mathcal{N}(r_{p},\Psi_{p},m_{p};r_{k},\Psi_{k},m_{k}),
\end{equation*}
where $\mathcal{N}$ is the quadratic nonlinearities with coefficients depending on the background spacetime $(r_{k},\Psi_{k},m_{k})$. A direct computation shows that the linearized operator $\mathcal{P}$ takes the following form under the self-similar coordinates
\begin{align}
\mathcal{P}^{(1)}(r_{p},\Psi_{p},m_{p}) =& 
\partial_{s}\partial_{z}r_{p}+q_{k}z\partial_{z}^{2}r_{p}+\left(q_{k}-\frac{\mu_{k}}{1-\mu_{k}}\frac{\partial_{s}r_{k}+2q_{k}z\partial_{z}r_{k}}{r_{k}}\right)\partial_{z}r_{p}-\frac{\mu_{k}}{1-\mu_{k}}\frac{\partial_{z}r_{k}}{r_{k}}\partial_{s}r_{p}\nonumber\\&+\frac{\mu_{k}(2-\mu_{k})}{(1-\mu_{k})^{2}}\frac{(\partial_{s}r_{k}+q_{k}z\partial_{z}r_{k})\partial_{z}r_{k}}{r_{k}^{2}}r_{p}-\frac{2}{(1-\mu_{k})^{2}}\frac{(\partial_{s}r_{k}+q_{k}z\partial_{z}r_{k})\partial_{z}r_{k}}{r_{k}^{2}}m_{p},
\label{eq:wave eq for rp under self-similar coordinates}\allowdisplaybreaks\\[1em]
\mathcal{P}^{(2)}(r_{p},\Psi_{p},m_{p}) =& 
\partial_{s}\partial_{z}\Psi_{p}+q_{k}z\partial_{z}^{2}\Psi_{p}+q_{k}\partial_{z}\Psi_{p}-\frac{\mu_{k}}{1-\mu_{k}}\frac{(\partial_{s}r_{k}+q_{k}z\partial_{z}r_{k})\partial_{z}r_{k}}{r_{k}^{2}}\Psi_{p}\nonumber\\&-\frac{1}{r_{k}^{2}}\frac{\mu_{k}}{1-\mu_{k}}\Psi_{k}\partial_{z}r_{k}\partial_{s}r_{p}-\left(\frac{1}{r_{k}^{2}}\frac{\mu_{k}}{1-\mu_{k}}\Psi_{k}(\partial_{s}r_{k}+q_{k}z\partial_{z}r_{k})+k\right)\partial_{z}r_{p}\nonumber\\&+\left(\frac{\mu_{k}(3-2\mu_{k})}{r_{k}^{3}(1-\mu_{k})^{2}}\Psi_{k}(\partial_{s}r_{k}+q_{k}z\partial_{z}r_{k})\partial_{z}r_{k}\right)r_{p}-\frac{2\Psi_{k}(\partial_{s}r_{k}+q_{k}z\partial_{z}r_{k})\partial_{z}r_{k}}{r_{k}^{3}(1-\mu_{k})^{2}}m_{p},
\label{eq:wave eq for psip under self-similar coordinate}\allowdisplaybreaks\\[1em]\mathcal{P}^{(3)}(r_{p},\Psi_{p},m_{p})=&
\partial_{z}m_{p}+\frac{r_{k}}{\partial_{z}r_{k}}(\partial_{z}\phi_{k})^{2}m_{p}-(1-\mu_{k})\frac{r_{k}}{\partial_{z}r_{k}}\partial_{z}\phi_{k}\partial_{z}\Psi_{p}+(1-\mu_{k})\partial_{z}\phi_{k}\Psi_{p}\nonumber\\&+\left(\frac{1}{2}\left(\frac{r_{k}}{\partial_{z}r_{k}}\right)^{2}(\partial_{z}\phi_{k})^{2}+\frac{r_{k}}{\partial_{z}r_{k}}\partial_{z}\phi_{k}(\phi_{k}-ks)\right)(1-\mu_{k})\partial_{z}r_{p}\nonumber\\&-
\left(\frac{1}{2}\frac{r_{k}}{\partial_{z}r_{k}}(\partial_{z}\phi_{k})^{2}\mu_{k}+(1-\mu_{k})\partial_{z}\phi_{k}(\phi_{k}-ks)\right)r_{p}.\label{eq:v transport eq for m under self-similar coordinates}
\end{align}

To simplify the notation, we suppress the dependence on the background and denote the nonlinearities by $\mathcal{N}(r_{p},\Psi_{p},m_{p})$. The explicit structure of $\mathcal{N}$ is computed in Appendix~\ref{appendix: nonlinear equations}.
\subsection{Schematic notations}
Since the coefficients of the linearized operator $\mathcal{P}$ are rather complicated, in this section, we use the estimates in Proposition~\ref{prop: estimate on the exact k self similar spacetime} on the background $k$-self-similar spacetime to introduce some schematic notations. 
\begin{definition}
Let $\mathcal{F}_{a}$, $\mathcal{G}_{a}$, and $\mathcal{K}$ be three classes of functions independent of $s$ satisfying the following properties:
\begin{enumerate}
  \item[(a)] For any $F_{a}\in \mathcal{F}_{a}$, we have\begin{equation}
  \left\Vert F_{a}\right\Vert_{L^{\infty}([-1,-\frac{1}{2}])}\lesssim a,\quad \left\Vert z\vert^{\epsilon} F_{a}\right\Vert_{L^{\infty}\left([-1,0]\right)}\lesssim\frac{a}{\epsilon},\quad\epsilon>0.
  \end{equation}
  \item[(b)] For any $G_{a}\in\mathcal{G}_{a}$, we have $\frac{dG_{a}}{dz}\in\mathcal{F}_{a}$ and \begin{equation}
  \Vert G_{a}\Vert_{L^{\infty}([-1,0])}\lesssim a.
  \end{equation} 
  \item[(c)] For any $K\in\mathcal{K}$, we have $K(z)$ is bounded.
\end{enumerate}
In the following, without any confusion, we will simply use $F_{a}(z)$, $G_{a}(z)$, and $K(z)$ to denote any function in the function classes $\mathcal{F}_{a}$, $\mathcal{G}_{a}$, and $\mathcal{K}$, respectively. Moreover, $F_{a}^{n}$, $G_{a}^{n}$, and $K^{n}$ denote the product of any $n$ functions in $\mathcal{F}_{a}$, $\mathcal{G}_{a}$, and $\mathcal{K}$, respectively.
\label{def:schematic notation}
\end{definition}
\begin{proposition}
Adopting the notation in Definition \ref{def:schematic notation}, we can rewrite the equations \eqref{eq:wave eq for rp under self-similar coordinates}-\eqref{eq:v transport eq for m under self-similar coordinates} in the following forms:\begin{align}
&\partial_{s}\partial_{z}r_{p}+q_{k}z\partial_{z}^{2}r_{p}+(q_{k}+G_{k^{2}}(z))\partial_{z}r_{p}+G_{k^{2}}(z)\partial_{s}r_{p}+G_{k^{2}}(z)r_{p} = K(z)\frac{m_{p}}{\mr{r}^{2}},\label{eq:wave equation for r rho}\\&
\partial_{s}\partial_{z}\Psi_{p}+q_{k}z\partial_{z}^{2}\Psi_{p}+q_{k}\partial_{z}\Psi_{p}+G_{k^{2}}(z)\Psi_{p} = G_{k^{2}}(z)\partial_{s}r_{p}+(G_{k^{2}}(z)+k)\partial_{z}r_{p}+G_{k^{2}}(z)r_{p}\nonumber\\&\hspace{6.5cm}+\vert z\vert^{\frac{1}{2}}G_{k}(z)\frac{m_{p}}{\mr{r}^{2}},\label{eq:wave equation for Psi rho}\\&
\partial_{z}m_{p}+\left(\frac{\mr{r}}{\partial_{z}\mr{r}}(\partial_{z}\mr{\phi})^{2}\right)m_{p} = \mr{r}^{2}\left(F_{k^{2}}(z)\frac{r_{p}}{\mr{r}}+F_{k^{2}}(z)\partial_{z}r_{p}+F_{k}(z)\left(\partial_{z}\left(\frac{\Psi_{p}}{\mr{r}}\right)-\partial_{z}\left(\mr{\phi}\frac{r_{p}}{\mr{r}}\right)\right)\right).\label{eq:z-transport eq for mrho}
\end{align}
\end{proposition}
\begin{proof}
Using Proposition \ref{prop: estimate on the exact k self similar spacetime}, a straightforward computation gives the proof of this proposition.
\end{proof}
Since we will also use the second-order energy estimate, we state the following schematic equations after taking one $\partial_{z}$-derivative of \eqref{eq:wave equation for r rho}-\eqref{eq:z-transport eq for mrho}:
\begin{proposition}
Adopting the notation in Definition \ref{def:schematic notation}, commuting the equations \eqref{eq:wave equation for r rho}-\eqref{eq:wave equation for Psi rho} with $\partial_{z}$, we have the following schematic forms:\begin{align}
&\begin{aligned}
&\partial_{s}\partial_{z}^{2}r_{p}+q_{k}z\partial_{z}^{3}r_{p}+\left(2q_{k}+G_{k^{2}}(z)\right)\partial_{z}^{2}r_{p}\\ =& F_{k^{2}}(z)\partial_{z}r_{p}+F_{k^{2}}(z)\partial_{s}r_{p}+F_{k^{2}}(z)r_{p}+K(z)\partial_{z}\left(\frac{m_{p}}{\mr{r}^{2}}\right)+(K(z)+F_{k^{2}}(z))\frac{m_{p}}{\mr{r}^{2}},
\end{aligned}\label{eq:second order wave equation for r rho}\\[1em]
&\begin{aligned}
&\partial_{s}\partial_{z}^{2}\Psi_{p}+q_{k}z\partial_{z}^{3}\Psi_{p}+2q_{k}\partial_{z}^{2}\Psi_{p}+G_{k^{2}}(z)\partial_{z}\Psi_{p}+F_{k^{2}}(z)\Psi_{p}\\=& G_{k}(z)\partial_{z}^{2}r_{p}+F_{k^{2}}(z)\partial_{s}r_{p}+F_{k^{2}}(z)\partial_{z}r_{p}+F_{k^{2}}(z)r_{p}+\vert z\vert^{\frac{1}{2}}G_{k}(z)\partial_{z}\left(\frac{m_{p}}{\mr{r}^{2}}\right)+F_{k}(z)\frac{m_{p}}{\mr{r}^{2}}.
\end{aligned}\label{second order wave equation for psi rho}
\end{align}
\end{proposition}
For the convenience of the integration, we define \begin{equation*}
    \mathcal{R}(s_{0}):=\left\{-1\leq z\leq 0,\ 0\leq s\leq s_{0}\right\},\quad \mathcal{H}(s_{0}): = \left\{v = 0,\ -s_{0}\leq u\leq0\right\}.
\end{equation*}

\subsection{Linearized stability result for the $k$-self-similar naked singularity interior}
Since the nonlinear stability result crucially relies on the decay estimates for solutions to the linearized Einstein-scalar field equations $\mathcal{P}(r_{p},\Psi_{p},m_{p}) = 0$ established in \cite{zhenglinear}, we recall one of the main theorems in \cite{zhenglinear} here, for perturbations above the threshold. Let $\delta_{0}\in(0,1)$ be the real number such that $\phi_{k}|_{\Sigma_{-1}^{(in)}}\in\mathcal{C}^{p_{k},\delta_{0}}$. Then we have the following.
\begin{theorem}[Linearized stability for perturbations above the threshold in \cite{zhenglinear}]
\label{thm: linearized result}
    For $k$ sufficiently small, any $\alpha\in(p_{k},\frac{3}{2})$, and $\beta\in(\alpha+\frac{1}{4},2)$, let $(r_{p},\Psi_{p},m_{p})$ be the solution to the spherically symmetric Einstein-scalar field equations $\mathcal{P} = 0$ linearized around the $k$-self-similar naked singularity interior with the initial data $$(r_{p},\Psi_{p})|_{\Sigma_{-1}^{(in)}} = (r_{p}^{0},\Psi_{p}^{0})\in\mathcal{C}_{N}^{\beta}\times\mathcal{C}_{N}^{\alpha}$$ for $N\geq 5$. Then there exist constants $c_{\infty}$ depending on the initial data and $C$ depending on $k$, such that for any $\alpha^{\prime}\in(p_{k},\alpha)$, the solution $(r_{p},\Psi_{p},m_{p})$ satisfies the following bounds: \begin{align}
        \sum_{0\leq i+j\leq 1}\left\vert\partial_{u}^{i}\partial_{v}^{j}(\Psi_{p}-c_{\infty}r_{k})\right\vert\leq& C\left(\left\Vert r_{p}^{0}\right\Vert_{\mathcal{C}_{N}^{\beta}}+\left\Vert \Psi_{p}^{0}\right\Vert_{\mathcal{C}_{N}^{\alpha}}\right)(-u)^{\alpha^{\prime}q_{k}-i-q_{k}j},\label{eq: linearized estimate 1}\\\sum_{0\leq i+j\leq 1}\left\vert\partial_{u}^{i}\partial_{v}^{j}r_{p}\right\vert\leq& C\left(\left\Vert r_{p}^{0}\right\Vert_{\mathcal{C}_{N}^{\beta}}+\left\Vert \Psi_{p}^{0}\right\Vert_{\mathcal{C}_{N}^{\alpha}}\right)(-u)^{\alpha^{\prime}q_{k}-i-q_{k}j},\label{eq: linearized estimate 2}\\\sum_{0\leq i+j\leq 1}\left\vert\partial_{u}^{i}\partial_{v}^{j}m_{p}\right\vert\leq& C\left(\left\Vert r_{p}^{0}\right\Vert_{\mathcal{C}_{N}^{\beta}}+\left\Vert \Psi_{p}^{0}\right\Vert_{\mathcal{C}_{N}^{\alpha}}\right)(-u)^{\alpha^{\prime}q_{k}-i-q_{k}j}.\label{eq: linearized estimate 3}
    \end{align}
    Moreover, for any fixed $k$, we have the following estimate on the constant $c_{\infty}$ in terms of the initial data: \begin{equation}
        \vert c_{\infty}\vert\lesssim_{k}\left\Vert r_{p}^{0}\right\Vert_{\mathcal{C}_{N}^{\beta}}+\left\Vert \Psi_{p}^{0}\right\Vert_{\mathcal{C}_{N}^{\alpha}}.
    \end{equation}
   In the near-axis region $-1\leq z\leq -e^{-2q_{k}}$, we have the following second-order estimates \begin{align}
        \sum_{0\leq i+j\leq 2}\left\vert\partial_{u}^{i}\partial_{v}^{j}(\Psi_{p}-c_{\infty}r_{k})\right\vert\leq& C\left(\left\Vert r_{p}^{0}\right\Vert_{\mathcal{C}_{N}^{\beta}}+\left\Vert \Psi_{p}^{0}\right\Vert_{\mathcal{C}_{N}^{\alpha}}\right)(-u)^{\alpha^{\prime}q_{k}-i-q_{k}j},\label{eq: linearized estimate 1 second order}\\\sum_{0\leq i+j\leq 2}\left\vert\partial_{u}^{i}\partial_{v}^{j}r_{p}\right\vert\leq& C\left(\left\Vert r_{p}^{0}\right\Vert_{\mathcal{C}_{N}^{\beta}}+\left\Vert \Psi_{p}^{0}\right\Vert_{\mathcal{C}_{N}^{\alpha}}\right)(-u)^{\alpha^{\prime}q_{k}-i-q_{k}j},\label{eq: linearized estimate 2 second order}\\\sum_{0\leq i+j\leq 2}\left\vert\partial_{u}^{i}\partial_{v}^{j}m_{p}\right\vert\leq& C\left(\left\Vert r_{p}^{0}\right\Vert_{\mathcal{C}_{N}^{\beta}}+\left\Vert \Psi_{p}^{0}\right\Vert_{\mathcal{C}_{N}^{\alpha}}\right)(-u)^{\alpha^{\prime}q_{k}-i-q_{k}j}.\label{eq: linearized estimate 3 second order}
    \end{align}
\end{theorem}
\begin{remark}
Due to the constraint equations of the Einstein-scalar field equations under the spherically symmetric assumption, the regularity of $r$ is closely related to the regularity of the scalar field $\phi$. In particular, in the nonlinear problem, once the initial regularity of $\phi$ is given, the regularity of both $\phi$ and $r$ is consequently fixed. Therefore, one should distinguish the regularity of the nonlinear problem and the freedom of imposing $r_{p}^{0}$ in the linearized problem. The aforementioned regularity freedom of $r_{p}^{0}$ in the linearized problem should be interpreted as follows: in estimating the decay rate of solutions to the linearized Einstein-scalar field equations, only the $\mathcal{C}_{N}^{\beta}$-norm of $r_{p}$ enters the analysis. This does not imply that, in the nonlinear problem, the regularity of $r_{p}$ can be below $C^{2,k^{2}p_{k}}$.
\end{remark}
\begin{remark}
    The constant $c_{\infty}$ in the above theorem arises from the translation invariance of the Einstein-scalar field equations. In other words, if $(r,\Psi,m)$ is a solution to the Einstein-scalar field equations, then so is $(r,\Psi+cr,m)$ for any $c\in\mathbb{R}$. Therefore, the linearized operator $\mathcal{P}$ admits a kernel $(r,\Psi,m) = (0,cr_{k},0)$, the regularity of which is above the threshold. The presence of this non-trivial kernel obstructs improved decay estimates for $\Psi_{p}$, and cannot be eliminated solely by the regularity constraint of the initial data. One of the main difficulties in the analysis of the linearized operator is to exclude this kernel in terms of the initial data. 
\end{remark}
\section{Main results and the outline of the proof}
\label{sec: main result and proof outline}
\subsection{Nonlinear stability result}
\label{sec: main results}
In this section, we state our nonlinear stability result. Recall that $q_{k} = 1-k^{2}$ and $p_{k} = \frac{1}{q_{k}}$, and the interior initial data of the background $k$-self-similar spacetime $(r_{k},\phi_{k},m_{k})$ belongs to the localized Hölder space $\mathcal{C}_{N}^{2-}\times \mathcal{C}_{N}^{p_{k},\delta_{0}}\times \mathcal{C}_{N}^{p_{k},\delta_{0}}$ for some $\delta_{0}\in(1-O(k),1)$ and any $N\geq 2$. In this paper, we will prove the following theorem.
\begin{theorem}
\label{thm: true main theorem}
    For $k\neq0$ sufficiently small, let $(r,\phi,m)$ be the solution to the spherically symmetric Einstein-scalar field equations \eqref{eq: ESF 1}-\eqref{eq: ESF 2} with the initial data $(r,\phi)|_{\Sigma_{-1}} = (r_{k},\phi_{k}+\phi_{p}^{0})$, where $\phi_{p}^{0}\in\bar{\mathcal{C}}_{N}^{p_{k},\delta}$ for $N\geq 6$ and $0<\delta\leq \delta_{0}$. According to Lemma \ref{lemma: decomposition of the function spaces}, we can decompose $\phi_{p}^{0}$ as \begin{equation*}
        \phi_{p}^{0} = c_{0}\mr{\phi}+\phi_{1},
    \end{equation*}
    with $\phi_{1}\in\bar{\mathcal{C}}_{N}^{p_{k}+\delta^{\prime}}$ for any $\delta^{\prime}\in(0,\delta)$. Then there exists $\epsilon_{0}$ sufficiently small depending on $k$, such that for any $\left\Vert\phi_{p}^{0}\right\Vert_{\bar{\mathcal{C}}_{N}^{p_{k},\delta}}\leq \epsilon_{0}$, the maximal globally hyperbolic development of this initial data is asymptotically flat, contains no trapped surface, and has incomplete future null infinity.
    
    Moreover, there exists a constant $c_{\infty}$ depending on the initial data with the bound \begin{equation}
        \vert c_{\infty}\vert\lesssim_{k}\left\Vert\phi_{p}^{0}\right\Vert_{\bar{\mathcal{C}}_{N}^{p_{k},\delta}},
    \end{equation}
    such that for any $\alpha\in\left(p_{k},\min\{p_{k}+\delta,\frac{3}{2}\}\right)$, the solution $(r,\phi,m)$ satisfies the following quantitative bounds:
    \begin{itemize}
        \item (Bounds in the interior region) The solution asymptotically converges to the $k$-self-similar solutions as $u\rightarrow0$ in the interior region $\{-1\leq u<0,\ -(-u)^{1-k^{2}}\leq v\leq 0\}$ under the renormalized gauge defined in Section~\ref{sec: renormalized gauge}:
        \begin{align}
        \sum_{0\leq i+j\leq 1}\left\vert\partial_{u}^{i}\partial_{v}^{j}(\phi-\phi_{k}-c_{\infty})\right\vert\lesssim_{k}& \left\Vert \phi_{p}^{0}\right\Vert_{\bar{\mathcal{C}}_{N}^{p_{k},\delta}}(-u)^{\alpha q_{k}-1-i-q_{k}j},\\\sum_{0\leq i+j\leq 1}\left\vert\partial_{u}^{i}\partial_{v}^{j}(r-(1+c_{0})r_{k})\right\vert\lesssim_{k}& \left\Vert \phi_{p}^{0}\right\Vert_{\bar{\mathcal{C}}_{N}^{p_{k},\delta}}(-u)^{\alpha q_{k}-i-q_{k}j},\\\sum_{0\leq i+j\leq 1}\left\vert\partial_{u}^{i}\partial_{v}^{j}(m-(1+c_{0})m_{k})\right\vert\lesssim_{k}& \left\Vert \phi_{p}^{0}\right\Vert_{\bar{\mathcal{C}}_{N}^{p_{k},\delta}}(-u)^{\alpha q_{k}-i-q_{k}j},\\\sum_{0\leq i+j\leq 1}\left\vert\partial_{u}^{i}\partial_{v}^{j}(\mu-\mu_{k})\right\vert\lesssim_{k}& \left\Vert \phi_{p}^{0}\right\Vert_{\bar{\mathcal{C}}_{N}^{p_{k},\delta}}(-u)^{\alpha q_{k}-1-i-q_{k}j}.
    \end{align}
    \item (Self-similar rate in the exterior neighborhood of $\{v = 0\}$) In the exterior neighborhood $\left\{0\leq z\lesssim 1\right\}$ of the singular horizon $\{v= 0\}$, the following bounds hold:\begin{align}
        \sum_{0\leq i+j\leq 1}\left\vert\partial_{u}^{i}\partial_{v}^{j}(\phi-\phi_{k}-c_{\infty})\right\vert\lesssim_{k}& \left\Vert \phi_{p}^{0}\right\Vert_{\bar{\mathcal{C}}_{N}^{p_{k},\delta}}(-u)^{\alpha q_{k}-1-i-q_{k}j},\\\sum_{0\leq i+j\leq 1}\left\vert\partial_{u}^{i}\partial_{v}^{j}(r-(1+c_{0})r_{k})\right\vert\lesssim_{k}& \left\Vert \phi_{p}^{0}\right\Vert_{\bar{\mathcal{C}}_{N}^{p_{k},\delta}}(-u)^{\alpha q_{k}-i-q_{k}j},\\\sum_{0\leq i+j\leq 1}\left\vert\partial_{u}^{i}\partial_{v}^{j}(m-(1+c_{0})m_{k})\right\vert\lesssim_{k}& \left\Vert \phi_{p}^{0}\right\Vert_{\bar{\mathcal{C}}_{N}^{p_{k},\delta}}(-u)^{\alpha q_{k}-i-q_{k}j},\\\sum_{0\leq i+j\leq 1}\left\vert\partial_{u}^{i}\partial_{v}^{j}(\mu-\mu_{k})\right\vert\lesssim_{k}& \left\Vert \phi_{p}^{0}\right\Vert_{\bar{\mathcal{C}}_{N}^{p_{k},\delta}}(-u)^{\alpha q_{k}-1-i-q_{k}j}.
    \end{align}
    \item (Convergence to the $k$-self-similar naked singularity spacetime in the region near $\{u = 0\}$) Let $U = -\left\vert u\right\vert^{1-k^{2}},\ V = v^{p_{k}}$, then there exist two constants $0<\eta,s<1$ sufficiently small, the solution asymptotically converges to the $k$-self-similar naked singularity spacetime with the rates: \begin{align}
        \sum_{0\leq i+j\leq 1}\left\vert\partial_{U}^{i}\partial_{V}^{j}(\phi-\phi_{k}-c_{\infty})\right\vert&\lesssim \left\Vert\phi_{p}^{0}\right\Vert_{\bar{\mathcal{C}}_{N}^{p_{k},\delta}}^{1-s}V^{\eta-i-q_{k}j},\\
        \sum_{0\leq i+j\leq 1}\left\vert\partial_{U}^{i}\partial_{V}^{j}(r-(1+c_{0})r_{k})\right\vert&\lesssim\left\Vert\phi_{p}^{0}\right\Vert_{\bar{\mathcal{C}}_{N}^{p_{k},\delta}}^{1-s} V^{1+\eta-i-q_{k}j},\\
        \sum_{0\leq i+j\leq 1}\left\vert\partial_{U}^{i}\partial_{V}^{j}(m-(1+c_{0})m_{k})\right\vert&\lesssim \left\Vert\phi_{p}^{0}\right\Vert_{\bar{\mathcal{C}}_{N}^{p_{k},\delta}}^{1-s}V^{1+\eta-i-q_{k}j},\\\sum_{0\leq i+j\leq 1}\left\vert\partial_{U}^{i}\partial_{V}^{j}\left(\mu-\mu_{k}\right)\right\vert&\lesssim \left\Vert\phi_{p}^{0}\right\Vert_{\bar{\mathcal{C}}_{N}^{p_{k},\delta}}^{1-s}V^{\eta-i-q_{k}j}.
    \end{align}
    \item (Asymptotic flatness) The spacetime $(r,\phi,m)$ is asymptotically flat in the sense that \begin{align*}
        \lim_{v\rightarrow\infty} r(u,v) = \infty,\quad \lim_{v\rightarrow \infty}m(u,v)<\infty,\quad \left\vert\partial_{v}^{i}\phi\right\vert(u,v)\lesssim v^{-1-i},\ 1\leq i\leq2,\quad -1\leq u<0.
    \end{align*}
    \end{itemize}
\end{theorem}
In this paper, we break down the proof of this main theorem into two parts: the proof of the nonlinear result for the interior region (Theorem \ref{thm: nonlinear stability result}), and the gluing for the exterior region (Theorem \ref{thm: nonlinear stability in the exterior region}). In the next few sections, we state the main results for each step.

\subsection{Nonlinear stability result for the interior region}
Writing the nonlinear spherically symmetric Einstein-scalar field equations in the form of $\mathcal{P} = \mathcal{N}$ introduced in Section \ref{sec: setup of the nonlinear stability} and using the linearized result, we can prove the following nonlinear stability result in the interior region:
\begin{theorem}
    \label{thm: nonlinear stability result}
    For $k\neq0$ sufficiently small, let $(r,\phi,m)$ be the solution to the spherically symmetric Einstein-scalar field equations \eqref{eq: ESF 1}-\eqref{eq: ESF 2} with the interior initial data $(r,\phi)|_{\Sigma_{-1}^{(in)}} = (r_{k},\phi_{k}+\phi_{p}^{0})$, where $\phi_{p}^{0}\in{\mathcal{C}}_{N}^{p_{k},\delta}$ for $N\geq 6$ and $0<\delta\leq \delta_{0}$. Assume the decomposition of $\phi_{p}^{0}$ is \begin{equation*}
        \phi_{p}^{0} = c_{0}\mr{\phi}+\phi_{1}
    \end{equation*}
    with $\phi_{1}\in\mathcal{C}^{p_{k}+\delta^{\prime}}$ for any $\delta^{\prime}\in(0,\delta)$. Then there exists $\epsilon_{0}$ sufficiently small depending on $k$, such that for any $\left\Vert\phi_{p}^{0}\right\Vert_{\mathcal{C}_{N}^{p_{k},\delta}}\leq \epsilon_{0}$, the solution $(r,\phi,m)$ satisfies the following in the interior region $\{-1\leq u<0,\ -(-u)^{q_{k}}\leq v\leq0\}$ for $\alpha\in(p_{k},\min\{p_{k}+\delta,\frac{3}{2}\})$
    \begin{align}
        \sum_{0\leq i+j\leq 1}\left\vert\partial_{u}^{i}\partial_{v}^{j}(\phi-\phi_{k}-c_{\infty})\right\vert\lesssim_{k}& \left\Vert \phi_{p}^{0}\right\Vert_{\mathcal{C}_{N}^{p_{k},\delta}}(-u)^{\alpha q_{k}-1-i-q_{k}j},\\\sum_{0\leq i+j\leq 1}\left\vert\partial_{u}^{i}\partial_{v}^{j}(r-(1+c_{0})r_{k})\right\vert\lesssim_{k}& \left\Vert \phi_{p}^{0}\right\Vert_{\mathcal{C}_{N}^{p_{k},\delta}}(-u)^{\alpha q_{k}-i-q_{k}j},\\\sum_{0\leq i+j\leq 1}\left\vert\partial_{u}^{i}\partial_{v}^{j}(m-(1+c_{0})m_{k})\right\vert\lesssim_{k}& \left\Vert \phi_{p}^{0}\right\Vert_{\mathcal{C}_{N}^{p_{k},\delta}}(-u)^{\alpha q_{k}-i-q_{k}j}.
    \end{align}
\end{theorem}
\noindent We will prove this theorem in Section \ref{sec: conclude the nonlinear stability}.
\subsection{Nonlinear stability for the exterior region}
To conclude the nonlinear stability for perturbations above the threshold, we combine the interior stability theorem proved above with the exterior stability result established in \cite{singh1}. We recall here the exterior stability result from \cite{singh1}, adapting to the setting in the present paper.

\begin{theorem}[\textup{\cite{singh1}}]
\label{thm: nonlinear stability in the exterior region}
    Under the renormalized gauge introduced in Section \ref{sec: renormalized gauge}, let $(\bar{r},\bar{\phi},\bar m)$ be the interior solution to the spherically symmetric Einstein-scalar field equations with the gauge choice $\bar r(-1,v) = r_{k}(-1,v)$ and the self-similar bounds\begin{align*}
         \sum_{ i+j =1}\left\vert\partial_{u}^{i}\partial_{v}^{j}(\bar{\phi}-\phi_{k})\right\vert\lesssim& (-u)^{\alpha^{\prime}q_{k}-1-i-q_{k}j},\\\sum_{0\leq i+j\leq 1}\left\vert\partial_{u}^{i}\partial_{v}^{j}(\bar{r}-(1+c_{0})r_{k})\right\vert\lesssim& (-u)^{\alpha^{\prime}q_{k}-i-q_{k}j},\\\sum_{0\leq i+j\leq 1}\left\vert\partial_{u}^{i}\partial_{v}^{j}(\bar m-(1+c_{0})m_{k})\right\vert\lesssim&(-u)^{\alpha^{\prime}q_{k}-i-q_{k}j}.
    \end{align*}
    Let $\phi_{p}^{0}\in\bar{\mathcal{C}}_{N}^{p_{k},\delta}$ for $N\geq 6$ with the decomposition \begin{equation*}
        \phi_{p}^{0} = c_{0}\mr{\phi}+\phi_{1},\quad \phi_{1}\in\bar{\mathcal{C}}^{\alpha+\delta^{\prime}},\quad \delta^{\prime}\in(0,\delta).
    \end{equation*}
    Moreover, assume that $\phi_{p}^{0}$ coincides with $\bar{\phi}(-1,v)-\phi_{k}(-1,v)$ for $v\in[-1,0]$. Then there exists $\epsilon_{0}>0$, such that for characteristic initial data $(r,\phi)|_{\Sigma_{-1}^{(ex)}} = (r_{k},\phi_{k}+\phi_{p}^{0})$ and $(r,\phi)|_{\underline{\Sigma}_{0}} = (\bar{r},\bar{\phi})|_{\underline{\Sigma}_{0}}$ with $\left\Vert\phi_{p}^{0}\right\Vert_{\bar{\mathcal{C}}_{N}^{p_{k},\delta}}\leq \epsilon_{0}$, the maximal globally hyperbolic development of this initial data of the spherically symmetric Einstein-scalar field equations is asymptotically flat, contains no trapped surface, and has incomplete future null infinity. Moreover, the solution is asymptotically $k$-self-similar, in the sense that \begin{itemize}
        \item (Self-similar rate in the exterior neighborhood of $\{v = 0\}$) There exists $\delta$ small such that in the exterior neighborhood $\left\{0\leq z\lesssim 1\right\}$ of the singular horizon $\{v= 0\}$, the following bounds hold:\begin{align}
        \sum_{ i+j= 1}\left\vert\partial_{u}^{i}\partial_{v}^{j}(\phi-\phi_{k})\right\vert&\lesssim_{k}\left\Vert \phi_{p}^{0}\right\Vert_{\bar{\mathcal{C}}_{N}^{p_{k},\delta}}(-u)^{\delta-i-q_{k}j},\\
        \sum_{0\leq i+j\leq 1}\left\vert\partial_{u}^{i}\partial_{v}^{j}(r-(1+c_{0})r_{k})\right\vert&\lesssim_{k}\left\Vert\phi_{p}^{0}\right\Vert_{\bar{\mathcal{C}}_{N}^{p_{k},\delta}}(-u)^{1+\delta-i-q_{k}j},\\
        \sum_{0\leq i+j\leq 1}\left\vert\partial_{u}^{i}\partial_{v}^{j}(m-(1+c_{0})m_{k})\right\vert&\lesssim_{k}\left\Vert\phi_{p}^{0}\right\Vert_{\bar{\mathcal{C}}_{N}^{p_{k},\delta}}(-u)^{1+\delta-i-q_{k}j}.
    \end{align}
    \item (Convergence to the $k$-self-similar naked singularity spacetime in the region near $\{u = 0\}$) Let $U = -\left\vert u\right\vert^{1-k^{2}},\ V = v^{p_{k}}$, then there exist two constants $0<\delta,s<1$ sufficiently small, the solution asymptotically converges to the $k$-self-similar naked singularity spacetime with the rates: \begin{align}
        \sum_{0\leq i+j\leq 1}\left\vert\partial_{U}^{i}\partial_{V}^{j}(\phi-\phi_{k})\right\vert&\lesssim \left\Vert\phi_{p}^{0}\right\Vert_{\bar{\mathcal{C}}_{N}^{p_{k},\delta}}^{1-s}V^{\delta-i-q_{k}j},\\
        \sum_{0\leq i+j\leq 1}\left\Vert\partial_{U}^{i}\partial_{V}^{j}(r-(1+c_{0})r_{k})\right\Vert&\lesssim\left\Vert\phi_{p}^{0}\right\Vert_{\bar{\mathcal{C}}_{N}^{p_{k},\delta}}^{1-s} V^{1+\delta-i-q_{k}j},\\
        \sum_{0\leq i+j\leq 1}\left\Vert\partial_{U}^{i}\partial_{V}^{j}(m-(1+c_{0})m_{k})\right\Vert&\lesssim \left\Vert\phi_{p}^{0}\right\Vert_{\bar{\mathcal{C}}_{N}^{p_{k},\delta}}^{1-s}V^{1+\delta-i-q_{k}j}.
    \end{align}
    \item (Asymptotic flatness) The spacetime $(r,\phi,m)$ is asymptotically flat in the sense that \begin{align*}
        \lim_{v\rightarrow\infty} r(u,v) = \infty,\quad \lim_{v\rightarrow \infty}m(u,v)<\infty,\quad \left\vert\partial_{v}^{i}\phi\right\vert(u,v)\lesssim v^{-1-i},\ 1\leq i\leq2,\quad -1\leq u<0.
    \end{align*}
    \end{itemize}
\end{theorem}
\begin{remark}
    In~\cite{singh2,singhphdtheis}, Singh only proved the asymptotic stability for exterior perturbations above the threshold and orbital stability for exterior perturbations at the threshold. Here, orbital stability means that for perturbations at the threshold, the resulting spacetime converges to another naked singularity spacetime as one approaches the singularity. However, using the method developed in the present paper, one can upgrade this orbital stability to asymptotic stability. Accordingly, in Theorem~\ref{thm: nonlinear stability in the exterior region}, we state asymptotic stability for perturbations at the threshold.
\end{remark}

\subsection{Outline of the proof}
\label{sec: outline of the proof whole sec}
\subsubsection{Nonlinear stability for perturbations above the threshold}
Using Theorem~\ref{thm: nonlinear stability in the exterior region} for exterior perturbations, to show nonlinear stability for perturbations above the threshold, it suffices to consider the nonlinear interior perturbations above the threshold.

Now, assuming we can establish the linearized result (Theorem \ref{thm: linearized result}) in the interior region, then we write the Einstein-scalar field equations under perturbations above the threshold as \begin{align*}
   & \mathcal{P}(r_{p},\Psi_{p},m_{p}) = \mathcal{N}(r_{p},\Psi_{p},m_{p}),\\&
   r_{p}|_{\Sigma_{-1}^{(in)}} = 0,\quad \Psi_{p}|_{\Sigma_{-1}^{(in)}} = \Psi_{p}^{0}\in\mathcal{C}_{N}^{\alpha},\quad \alpha>p_{k},
\end{align*}
where $\mathcal{N}$ is a quadratic nonlinearity with the explicit form computed in Appendix \ref{appendix: nonlinear equations}. At this stage, a naive toy model to capture the quantitative behavior of $\mathcal{N}$ is $\mathcal{N}_{toy} = \frac{1}{r_{k}}(\partial_{z}\Psi_{p})^{2}\approx e^{s}(\partial_{z}\Psi_{p})^{2}$ under the self-similar coordinates. The estimates for solutions to the linearized Einstein-scalar field equations established in \ref{thm: linearized result} can be equivalently stated as \begin{align}
    \sum_{0\leq i+j\leq 1}\left\vert\partial_{s}^{i}\partial_{z}^{j}(\Psi_{p}-c_{\infty}r_{k})\right\vert\lesssim e^{-\alpha^{\prime}q_{k}s},\quad \sum_{0\leq i+j\leq 1}\left\vert\partial_{s}^{i}\partial_{z}^{j}r_{p}\right\vert\lesssim e^{-\alpha^{\prime}q_{k}s},\quad \sum_{0\leq i+j\leq 1}\left\vert\partial_{s}^{j}\partial_{z}^{j}m_{p}\right\vert\lesssim e^{-\alpha^{\prime}q_{k}s},
\end{align} 
with $\alpha^{\prime}>p_{k}$.
Substituting these bounds into the toy nonlinearity $\mathcal{N}_{toy}$, we have \begin{equation*}
    \vert\mathcal{N}_{toy}\vert\lesssim e^{-\alpha^{\prime}q_{k}s}e^{-(\alpha^{\prime}-p_{k})q_{k}s}.
\end{equation*}
Therefore, in the nonlinear analysis of the Einstein-scalar field equations under perturbations above the threshold, the linearized operator dominates, while the quadratic nonlinearities remain sub-leading. We employ a fixed-point argument to obtain estimates for the solution of the nonlinear system.
\subsubsection{Scaling invariance and the nonlinear stability for perturbations at the threshold}
To extend the nonlinear stability result for perturbations above the threshold to those at the threshold, we first note that Einstein-scalar field equations admit the scaling invariance, i.e., if $(r,\Psi,m)$ is a solution, then so is $(ar,a\Psi,am)$ for any $a\in \mathbb{R}\backslash\{0\}$. Therefore, the linearized Einstein-scalar field operator $\mathcal{P}$ admits another kernel $(r,\Psi,m) = (r_{k},\Psi_{k},m_{k})$, arising from the scaling invariance. Although the decay estimate for solutions to the linearized Einstein-scalar field equations only applies to perturbations above the threshold, using the decomposition lemma (Lemma~\ref{lemma: decomposition of the function spaces}) and subtracting the nontrivial kernel $(r_{k},\Psi_{k},m_{k})$ from the translation invariance, one can still obtain the decay estimates for solutions to the linearized Einstein-scalar field equations with initial data of the threshold regularity $\mathcal{C}_{N}^{p_{k},\delta}$. 

To be more precise, assume the initial perturbation $(r_{p},\phi_{p})|_{\Sigma_{-1}^{(in)}} = (0,\phi_{p}^{0}) = (0,c_{0}\mr{\phi}+\phi_{1})$ with $\phi_{1}\in\mathcal{C}^{p_{k}+\delta}$. The the corresponding $\Psi_{p}|_{\Sigma_{-1}^{(in)}} = c_{0}\Psi_{k}|_{\Sigma_{-1}^{(in)}}+\Psi_{1}$ for $\Psi_{1}\in\mathcal{C}^{p_{k}+\delta}$. Therefore, subtracting the non-trivial kernel of $\mathcal{P}$ from the translation invariance, the linearized problem with initial data at the threshold \begin{equation*}
    \mathcal{P}(r_{p},\Psi_{p},m_{p}) = 0,\quad (r_{p},\Psi_{p})|_{\Sigma_{-1}^{(in)}} = (r_{p}^{0},\Psi_{p}^{0})
\end{equation*}
is equivalent to the following linearized problem with initial data above the threshold \begin{equation*}
    \mathcal{P}(r_{p}-c_{0}r_{k},\Psi_{p}-c_{0}\Psi_{k},m_{p}-c_{0}m_{k}) = 0,\quad (r_{p},\Psi_{p})|_{\Sigma_{-1}^{(in)}} = (r_{p}^{0}-c_{0}r_{k},\Psi_{1}).
\end{equation*}
Therefore, for the linearized Einstein-scalar field equations, we can extend the result for perturbations above the threshold to those at the threshold, with $(r_{p},\Psi_{p},m_{p})$ in~\eqref{eq: linearized estimate 1}--\eqref{eq: linearized estimate 3} replaced by $(r_{p}-c_{0}r_{k},\Psi_{p}-c_{0}\Psi_{k},m_{p}-c_{0}m_{k})$ and $\alpha^{\prime}$ replaced by $\widetilde{\alpha}$ for any $\widetilde{\alpha}\in(p_{k},\min\{\frac{3}{2},p_{k}+\delta\})$.

Although the decay estimates for solutions to the linearized equations survive for perturbations at the threshold, the additional $(c_{0}r_{k},c_{0}\Psi_{k},c_{0}m_{k})$ is often believed to destroy the nonlinearities $\mathcal{N}$. For the toy nonlinearity $\mathcal{N}_{toy}$, those additional subtractions of the background will generate terms with threshold behavior, i.e., one can only get $\vert \mathcal{N}_{toy}\vert\lesssim e^{-s}$, which means $\mathcal{N}_{toy}$ will dominate in the nonlinear analysis. Therefore, one usually expects that the linearized estimates fail to control the solutions to the nonlinear problem.

Nonetheless, the toy nonlinearity $\mathcal{N}_{toy}$ does not capture the scaling invariance of the Einstein-scalar field equations. In fact, if one examines the structure of the nonlinearities in the Einstein-scalar field equations carefully, the additional non-trivial kernel will not destroy the nonlinear structure. A simple argument for this is as follows \begin{align*}
    &\mathcal{P}(r_{p},\Psi_{p},m_{p})-\mathcal{N}(r_{p},\Psi_{p},m_{p};r_{k},\Psi_{k},m_{k}) \\=& \mathcal{F}(r_{k}+r_{p},\Psi_{k}+\Psi_{p},m_{k}+m_{p})-\mathcal{F}(r_{k},\Psi_{k},m_{k})\\=&
    \mathcal{F}((1+c_{0})r_{k}+r_{p}-c_{0}r_{k},(1+c_{0})\Psi_{k}+\Psi_{p}-c_{0}\Psi_{k},(1+c_{0})m_{k}+m_{p}-c_{0}m_{k})-\mathcal{F}(r_{k},\Psi_{k},m_{k})\\=&
    \mathcal{F}((1+c_{0})r_{k}+r_{p}-c_{0}r_{k},(1+c_{0})\Psi_{k}+\Psi_{p}-c_{0}\Psi_{k},(1+c_{0})m_{k}+m_{p}-c_{0}m_{k})-\mathcal{F}((1+c_{0})r_{k},(1+c_{0})\Psi_{k},(1+c_{0})m_{k})\\=&
    \mathcal{P}(r_{p}-c_{0}r_{k},\Psi_{p}-c_{0}\Psi_{k},m_{p}-c_{0}m_{k})-\mathcal{N}(r_{p}-c_{0}r_{k},\Psi_{p}-c_{0}\Psi_{k},m_{p}-c_{0}m_{k};(1+c_{0})r_{k},(1+c_{0})\Psi_{k},(1+c_{0})m_{k}),
\end{align*} 
where the third identity above uses the scaling invariance of the Einstein-scalar field equations. Therefore, by the scaling invariance of the Einstein-scalar field equations again, subtracting the kernel of the linearized operator from scaling symmetry preserves the structure of the nonlinear terms, at the cost of rescaling the coefficients of $\mathcal{N}$, which does not affect the decay rate of the nonlinearities.

\section{Nonlinear stability for perturbations above the threshold}
\label{sec: conclude the nonlinear stability}
\subsection{Set-up of the initial perturbations}
Given the gauge choice of the center $\Gamma = \{(-v)= (-u)^{q_{k}}\}$, the free characteristic initial data consist of the gauge choice of $r$ and the scalar field $\phi$. Since in the spherically symmetric Einstein-scalar field equations, the scalar field $\phi$ only appears in the form of derivatives, the Einstein-scalar field equations are invariant under the translation of $\phi$. In the formulation of the stability problem, we fix the translation freedom of the background $k$-self-similar solution by setting $\phi_{k}(-1,0) = 0$.

We use the $(r,\Psi,m)$ formulation. The gauge choice and the initial data of $(r,\Psi)|_{\Sigma_{-1}^{(in)}}$ is \begin{equation}
    r(-1,v) = r_{k}+r_{p}^{0}(-1,v),\quad \Psi(-1,v) = \Psi_{k}(-1,v)+\Psi_{p}^{0},\quad -1\leq v\leq0.\label{eq: gauge choice and the initial data in the nonlinear argument}
\end{equation}
where $(r_{p}^{0},\Psi_{p}^{0})\in \mathcal{C}_{N+1}^{\frac{7}{4}}\times\mathcal{C}^{\alpha}_{N+1}$ for $\alpha\in(p_{k},\frac{3}{2})$.

\begin{remark}
    For initial scalar field perturbations above the threshold, a natural gauge choice of $r|_{\Sigma_{-1}^{(in)}}$ could be $r_{k}|_{\Sigma_{-1}^{(in)}}$, i.e., $r_{p}^{0} = 0$ in \eqref{eq: gauge choice and the initial data in the nonlinear argument}. Nevertheless, to facilitate the proof of nonlinear stability at the threshold, we consider nontrivial $r_{p}^{0}$ in this section.
\end{remark}

In the above formulation, we did not fix the translation freedom of $\Psi$. Hence, one still has the freedom to add $cr$ on $\Psi_{p}^{0}$, which will not change the function space of $\Psi_{p}^{0}$.

Motivated by this observation, we say that two initial perturbations $\Psi_{p}^{0}$ and $\widetilde{\Psi}_{p}^{0}$ are equivalent if they differ by a constant multiple of $r|_{\Sigma_{-1}^{(in)}}$. The solutions $(r,\Psi,m)$ and $(\widetilde{r},\widetilde{\Psi},\widetilde{m})$ to the Einstein-scalar field equations with two equivalent initial perturbations have the following relation:\begin{equation*}
    r = \widetilde{r},\quad \Psi-\widetilde{\Psi} = cr,\quad m = \widetilde{m},
\end{equation*}
which are essentially the same solution. We denote by $[\Psi_{p}^{0}]$ this equivalence class of $\Psi_{p}^{0}$.

Our strategy for demonstrating nonlinear stability is to use the fixed-point argument.
\subsection{Set-up of the fixed point argument}
To set up the fixed-point argument, we express the nonlinear Einstein-scalar field equations in the form of \begin{align*}
    &\mathcal{P}(r_{p},\Psi_{p},m_{p}) =\mathcal{N}(r_{p},\Psi_{p},m_{p}),\\&
    (r_{p},\Psi_{p})|_{\Sigma_{-1}^{(in)}} = (r_{p}^{0},\Psi_{p}^{0}).
\end{align*}
where $r_{p} = r-r_{k}$, $\Psi_{p} = \Psi-\Psi_{k}$, and $m_{p} = m-m_{k}$, and the explicit nonlinear structure $\mathcal{N}$ is derived in Appendix \ref{appendix: nonlinear equations} with $b=0$. In the fixed-point argument, we substitute $r_{p}$, $\Psi_{p}$, and $m_{p}$ in the nonlinearity by $f$, $g$, and $h$ with $(f,g,h)$ in suitably designed function spaces. Precisely, we shall consider the following evolutionary problem:
\begin{equation}
\begin{aligned}
    &\mathcal{P}(r_{p},\Psi_{p},m_{p}) = \mathcal{N}(f,g,h),\\&
    (r_{p},\Psi_{p})|_{\Sigma_{-1}^{(in)}} = (r_{p}^{0},\Psi_{p}^{0}).
\end{aligned}\label{eq: nonlinear formulation}
\end{equation}
We can further decompose the above nonlinear evolutionary equations as 
\begin{equation}
\begin{aligned}
    &\mathcal{P}(r_{p}^{(1)},\Psi_{p}^{(1)},m_{p}^{(1)}) = 0,\\&
   (r_{p}^{(1)},\Psi_{p}^{(1)})|_{\Sigma_{-1}^{(in)}} = (r_{p}^{0},\Psi_{p}^{0}),
\end{aligned}
\label{eq: nonlinear equation 1}
\end{equation}
and 
\begin{equation}
\begin{aligned}
    &\mathcal{P}(r_{p}^{(2)},\Psi_{p}^{(2)},m_{p}^{(2)}) = \mathcal{N}(f,g,h),\\&
    (r_{p}^{(2)},\Psi_{p}^{(2)}) = (0,0).
\end{aligned}
\label{eq: nonlinear equation 2}
\end{equation}
Next, we define the function space we will consider in this section. Let $(f,\partial_{s}f,g,h)(s,\cdot)\in H_{N+1}^{-\frac{1}{4}+\tau}\times H_{N}^{-\frac{1}{4}+\tau}\times H_{N+1}^{\frac{3}{2}-\alpha+\tau}\times H_{N+1}^{\frac{3}{2}-p_{k}+\tau}$ with \begin{equation*}
\left\Vert(f,g,h)\right\Vert_{\mathcal{S}^{a,N,\tau}}: = \sup_{s\geq 0}e^{(1+a)s}\left(\left\Vert f\right\Vert_{H_{N+1}^{-\frac{1}{4}+\tau}}+\left\Vert\partial_{s}f\right\Vert_{H_{N}^{-\frac{1}{4}+\tau}}+\left\Vert g\right\Vert_{H_{N+1}^{\frac{3}{2}-\alpha+\tau}}+\left\Vert h\right\Vert_{H_{N+1}^{\frac{3}{2}-p_{k}+\tau}}\right)<\infty
\end{equation*}
for some suitable $a>0$ and $0<\tau<\alpha-p_{k}$ which will be determined later. We define the function space $\mathcal{S}^{a,N,\tau}$ to be \begin{equation*}
\mathcal{S}^{a,N,\tau}
:=
\left\{
(f,g,h)
\middle|\
\begin{array}{l}
(f,\partial_s f,g,h)(s,\cdot)\in
H_{N+1}^{\tau}\times H_{N}^{\tau}\times
H_{N+1}^{\frac32-\alpha+\tau}\times
H_{N+1}^{\frac32-p_k+\tau},\\[0.3em]
\|(f,g,h)\|_{\mathcal{S}^{a,N,\tau}}<\infty,\\[0.3em]
f(s,-1)=g(s,-1)=h(s,-1)=0
\end{array}
\right\}.
\end{equation*}


Let  $\gamma:= \frac{3}{2}-\alpha+\tau$ and $S_{\delta}^{a,N,\tau}$ denote the ball of radius $\delta$ in $\mathcal{S}^{a,N,\tau}$. By Lemma \ref{lemma: sobolev embedding}, we have that \begin{equation*}
    f\in \mathcal{C}_{N}^{\frac{7}{4}-\tau},\quad \partial_{s}f\in\mathcal{C}_{N-1}^{\frac{7}{4}-\tau},\quad g\in\mathcal{C}^{\alpha-\tau}_{N},\quad h\in\mathcal{C}_{N}^{p_{k}-\tau}.
\end{equation*}

We define the map $\mathfrak{A}(f,g,h)$ as the solution to \eqref{eq: nonlinear formulation}. We use the superscript to denote the component of $\mathfrak{A}$, i.e., $\mathfrak{A}^{(1)}$ corresponds to $r_{p}$, $\mathfrak{A}^{(2)}$ corresponds to $\Psi_{p}$, and $\mathfrak{A}^{(3)}$ corresponds to $m_{p}$. Finally, we define the maps $\mathfrak{A}_{1}(f,g,h)$ as the solution to \eqref{eq: nonlinear equation 1} and $\mathfrak{A}_{2}(f,g,h)$ as the solution to \eqref{eq: nonlinear equation 2}. Then, we have $\mathfrak{A} = \mathfrak{A}_{1}+\mathfrak{A}_{2}$. In the following, we shall prove that $\mathfrak{A}$ is a contraction map from $\mathcal{S}_{\delta}^{a,N,\tau}$ to $\mathcal{S}_{\delta}^{a,N,\tau}$ for suitable $a$, $N$, $\delta$, and $\tau$.

\subsection{Lower-order estimate on $\mathfrak{A}_{1}$}
By Theorem \ref{thm: linearized result}, for the solution $(\mathfrak{A}_{1}^{(1)},\mathfrak{A}_{1}^{(2)},\mathfrak{A}_{1}^{(3)})$ to \eqref{eq: nonlinear equation 1} with the initial data $(r_{p},\Psi_{p})|_{\Sigma^{(in)}_{-1}} = (r_{p}^{0},\Psi_{p}^{0})$, there exist $\alpha^{\prime}\in(p_{k},\alpha)$ and $c_{\infty} = c_{\infty}(\Psi_{p}^{0})$ depending on the initial data, such that we have \begin{equation}
\begin{aligned}
    \sum_{0\leq i+j\leq 1}\left\Vert\partial_{s}^{i}\partial_{z}^{j}\mathfrak{A}_{1}^{(1)}\right\Vert_{L_{z}^{\infty}}\lesssim&\left(\left\Vert r_{p}^{0}\right\Vert_{\mathcal{C}_{N}^{\frac{7}{4}}}+\left\Vert\Psi_{p}^{0}\right\Vert_{\mathcal{C}_{N}^{\alpha}}\right)e^{-\alpha^{\prime}q_{k}s},\\
    \sum_{0\leq i+j\leq 1}\left\Vert\partial_{s}^{i}\partial_{z}^{j}\left(\mathfrak{A}_{1}^{(2)}-c_{\infty}\left(r_{p}^{0},\Psi_{p}^{0}\right)r_{k}\right)\right\Vert_{L_{z}^{\infty}}\lesssim&\left(\left\Vert r_{p}^{0}\right\Vert_{\mathcal{C}_{N}^{\frac{7}{4}}}+\left\Vert\Psi_{p}^{0}\right\Vert_{\mathcal{C}_{N}^{\alpha}}\right)e^{-\alpha^{\prime}q_{k}s},\\
    \sum_{0\leq i+j\leq 1}\left\Vert\partial_{s}^{i}\partial_{z}^{j}\left(\mathfrak{A}^{(3)}_{1}\right)\right\Vert_{L_{z}^{\infty}}\lesssim&\left(\left\Vert r_{p}^{0}\right\Vert_{\mathcal{C}_{N}^{\frac{7}{4}}}+\left\Vert\Psi_{p}^{0}\right\Vert_{\mathcal{C}_{N}^{\alpha}}\right)e^{-\alpha^{\prime}q_{k}s}\\
    \left\vert c_{\infty}(r_{p}^{0},\Psi_{p}^{0})\right\vert\lesssim&\left\Vert r_{p}^{0}\right\Vert_{\mathcal{C}_{N}^{\frac{7}{4}}}+\left\Vert\Psi_{p}^{0}\right\Vert_{\mathcal{C}_{N}^{\alpha}}.
    \end{aligned}
\end{equation} 
Hence, we can modify the initial data $\Psi_{p}^{0}$ of \eqref{eq: nonlinear equation 1} to be $\Psi_{p}^{0}-c_{\infty}(\Psi_{p}^{0})r_{k}$, such that the solution $(\mathfrak{A}_{1}^{(1)},\mathfrak{A}_{1}^{(2)},\mathfrak{A}_{1}^{(3)})$ to \eqref{eq: nonlinear equation 1} with the initial data $(r_{p},\Psi_{p})|_{\Sigma_{-1}^{(in)}} = (r_{p}^{0},\Psi_{p}^{0}-c_{\infty}(\Psi_{p}^{0})r_{k})$ satisfies \begin{equation}
    \sum_{0\leq i+j\leq 1}e^{\alpha^{\prime}q_{k}s}\left(\left\Vert\partial_{s}^{i}\partial_{z}^{j}\mathfrak{A}_{1}^{(1)}\right\Vert_{L_{z}^{\infty}}+\left\Vert\partial_{s}^{i}\partial_{z}^{j}\mathfrak{A}_{1}^{(2)}\right\Vert_{L_{z}^{\infty}}+\left\Vert\partial_{s}^{i}\partial_{z}^{j}\mathfrak{A}_{1}^{(3)}\right\Vert_{L_{z}^{\infty}}\right)\lesssim\left\Vert r_{p}^{0}\right\Vert_{\mathcal{C}_{N}^{\frac{7}{4}}}+\left\Vert\Psi_{p}^{0}\right\Vert_{\mathcal{C}_{N}^{\alpha}}\lesssim \left\Vert r_{p}^{0}\right\Vert_{\mathcal{C}_{N+1}^{\frac{7}{4}}}+\left\Vert\Psi_{p}^{0}\right\Vert_{\mathcal{C}_{N+1}^{\alpha}}.
\end{equation}
\subsection{Lower-order estimate on $\mathfrak{A}_{2}$}
Since the linearized result (Theorem \ref{thm: linearized result}) is only for the homogeneous linear equations $\mathcal{P} = 0$, in this section, we prove the analogous result for the inhomogeneous equations with the trivial initial data by using the Duhamel principle. Under the self-similar coordinates, the equations \eqref{eq: nonlinear equation 2} can be written as \begin{align*}
    \partial_{s}\partial_{z}r_{p}-\frac{\mu_{k}}{1-\mu_{k}}\frac{\partial_{z}r_{k}}{r_{k}}\partial_{s}r_{p}+S_{1}[r_{p},\Psi_{p},m_{p}] &= \mathcal{N}^{(1)}(f,g,h),\\\partial_{s}\partial_{z}\Psi_{p}-\frac{\Psi_{k}}{r_{k}^{2}}\frac{\mu_{k}}{1-\mu_{k}}\partial_{z}r_{k}\partial_{s}r_{p}+S_{2}[r_{p},\Psi_{p},m_{p}]& = \mathcal{N}^{(2)}(f,g,h),\\
    \partial_{z}m_{p}+\frac{r_{k}}{\partial_{z}r_{k}}(\partial_{z}\phi_{k})^{2}m_{p}+S_{3}[r_{p},\Psi_{p}]& = \mathcal{N}^{(3)}(f,g,h),
\end{align*}
where $S_{1}$, $S_{2}$, and $S_{3}$ are three linear operators which only depend on the $(r_{p},\Psi_{p},m_{p})$ and $z$-derivatives of $(r_{p},\Psi_{p},m_{p})$. More precisely, $S_{1}$ and $S_{2}$ take the form of \begin{align*}
    S_{1}[r_{p},\Psi_{p},m_{p}] &= \widetilde{S}_{1}[r_{p},\Psi_{p}]-\frac{2}{(1-\mu_{k})^{2}}\frac{(\partial_{s}r_{k}+q_{k}z\partial_{z}r_{k})\partial_{z}r_{k}}{r_{k}^{2}}m_{p},\\
    S_{2}[r_{p},\Psi_{p},m_{p}]& = \widetilde{S}_{2}[r_{p},\Psi_{p}]-\frac{2\Psi_{k}(\partial_{s}r_{k}+q_{k}z\partial_{z}r_{k})\partial_{z}r_{k}}{r_{k}^{3}(1-\mu_{k})^{2}}m_{p},
\end{align*}
where $\widetilde{S}_{1}$ and $\widetilde{S}_{2}$ have the corresponding explicit form in \eqref{eq:wave eq for rp under self-similar coordinates}-\eqref{eq: wave equation for psi under self-similar coordinates}, respectively. Then, we can express $m_{p}$ as \begin{equation}
    m_{p}:=m_{p}[r_{p},\Psi_{p},\mathcal{N}^{(3)}]= \underbrace{\int_{-1}^{z}-e^{\int_{z}^{\widetilde{z}}\frac{r_{k}}{\partial_{z}r_{k}}(\partial_{z}\phi_{k})^{2}}S_{3}[r_{p},\Psi_{p}]}_{=:(m_{p})_{L}(r_{p},\Psi_{p})}+\int_{-1}^{z}e^{\int_{z}^{\widetilde{z}}\frac{r_{k}}{\partial_{z}r_{k}}(\partial_{z}\phi_{k})^{2}}\mathcal{N}^{(3)}(f,g,h)
    \label{eq: definition of mpl}
\end{equation}
Hence, we can rewrite the equation $\mathcal{P}(r_{p},\Psi_{p},m_{p}) = \mathcal{N}(f,g,h)$ as \begin{equation}
\begin{aligned}
   & \mathcal{P}^{(1)}(r_{p},\Psi_{p},(m_{p})_{L}) =\underbrace{\mathcal{N}^{(1)}(f,g,h)+\frac{2(\partial_{s}r_{k}+q_{k}z\partial_{z}r_{k})\partial_{z}r_{k}}{(1-\mu_{k})^{2}r_{k}^{2}}\int_{-1}^{z}e^{\int_{z}^{\widetilde{z}}\frac{r_{k}}{\partial_{z}r_{k}}(\partial_{z}\phi_{k})^{2}}\mathcal{N}^{(3)}(f,g,h)}_{=:\widetilde{\mathcal{N}}^{(1)}},\\&
   \mathcal{P}^{(2)}(r_{p},\Psi_{p},(m_{p})_{L}) = \underbrace{\mathcal{N}^{(2)}(f,g,h)+\frac{2\Psi_{k}(\partial_{s}r_{k}+q_{k}z\partial_{z}r_{k})\partial_{z}r_{k}}{(1-\mu_{k})^{2}r_{k}^{3}}\int_{-1}^{z}e^{\int_{z}^{\widetilde{z}}\frac{r_{k}}{\partial_{z}r_{k}}(\partial_{z}\phi_{k})^{2}}\mathcal{N}^{(3)}(f,g,h)}_{=:\widetilde{\mathcal{N}}^{(2)}},\\&
   \mathcal{P}^{(3)}(r_{p},\Psi_{p},(m_{p})_{L}) = 0.
    \end{aligned}
\end{equation}
Then, by the standard Duhamel principle, the solution $(r_{p},\Psi_{p},m_{p})$ to \eqref{eq: nonlinear equation 2} can be written in the form of \begin{equation}
\begin{aligned}
    &r_{p}(s,z) = \int_{0}^{s}r_{p}^{s_{1}}(s,z)ds_{1},\quad \Psi_{p}(s,z) = \int_{0}^{s}\Psi_{p}^{s_{1}}(s,z)ds_{1},\\&m_{p}(s,z) = (m_{p})_{L}(s,z)+\int_{-1}^{z}e^{\int_{z}^{\widetilde{z}}\frac{r_{k}}{\partial_{z}r_{k}}(\partial_{z}\phi_{k})^{2}}\mathcal{N}^{(3)}(f,g,h),
    \\
    &\partial_{z}(m_{p})_{L}(s,z)+\frac{r_{k}}{\partial_{z}r_{k}}(\partial_{z}\phi_{k})^{2}(m_{p})_{L}+S_{3}[r_{p},\Psi_{p}] =0.
    \end{aligned}
\end{equation}
where $(r_{p}^{s_{1}},\Psi_{p}^{s_{1}})(s,z)$ is a solution to \begin{align}
    \mathcal{P}(r_{p}^{s_{1}},\Psi_{p}^{s_{1}},(m_{p})_{L}) &= 0,\label{eq: duhamel main equation}\\
    \frac{d}{dz}r_{p}^{s_{1}}(s_{1},z)- \frac{\mu_{k}}{1-\mu_{k}}\frac{\partial_{z}r_{k}}{r_{k}}r_{p}^{s_{1}}(s_{1},z) &= \widetilde{\mathcal{N}}^{(1)},\label{eq: duhamel initial data for rp}\\[1em]
    \frac{d}{dz}\Psi_{p}^{s_{1}}(s_{1},z)-\frac{\Psi_{k}}{r_{k}^{2}}\frac{\mu_{k}}{1-\mu_{k}}\partial_{z}r_{k}r_{p}^{s_{1}}(s_{1},z)& = \widetilde{\mathcal{N}}^{(2)}.\label{eq: duhamel initial data for psip}
\end{align}
Hence, we can express the initial data $r_{p}^{s_{1}}(s_{1},z)$ and $\Psi_{p}^{s_{1}}(s_{1},z)$ of \eqref{eq: duhamel main equation} as \begin{align}
    r_{p}^{s_{1}}(s_{1},z) &= \int_{-1}^{z}e^{\int_{z^{\prime}}^{z}\frac{\mu_{k}}{1-\mu_{k}}\frac{\partial_{z}r_{k}}{r_{k}}d\widetilde{z}}\widetilde{\mathcal{N}}^{(1)}(f,g,h)(s_{1},z^{\prime})dz^{\prime},\label{eq: initial data r for duhamel principle}\\
    \Psi_{p}^{s_{1}}(s_{1},z) &=\int_{-1}^{z}\frac{\Psi_{k}}{r_{k}^{2}}\frac{\mu_{k}}{1-\mu_{k}}\partial_{z}r_{k}r_{p}^{s_{1}}(s_{1},z^{\prime})+\widetilde{\mathcal{N}}^{(2)}(f,g,h)(s_{1},z^{\prime})dz^{\prime}.\label{eq: initial data psi for duhamel principle}
    \end{align}

The following lemmas explore the regularity of $r_{p}^{s_{1}}$ and $\Psi_{p}^{s_{1}}$ in terms of the regularity of $\widetilde{\mathcal{N}}^{(1)}$ and $\widetilde{\mathcal{N}}^{(2)}$.
\begin{lemma}
\label{lemma: Using Talor theory to control the Linfty norm}
    Let $f\in C^{N}([-1,-\frac{1}{2}])$ with $f(-1) = 0$, then we have \begin{equation*}
        \left\Vert\partial_{z}^{n}\left(\frac{1}{(z+1)^{2}}f\right)\right\Vert_{L_{z}^{\infty}([-1,-\frac{1}{2}])}\lesssim \left\Vert \partial_{z}^{n+2}f\right\Vert_{L_{z}^{\infty}([-1,-\frac{1}{2}])},\quad 0\leq n\leq N-2.
    \end{equation*}
\end{lemma}
\begin{proof}
    Using the Taylor expansion, we have \begin{equation*}
        f(z) =\sum_{i = 0}^{n+1}\frac{1}{i!}(\partial_{z}^{i}f)(-1,0)(z+1)^{i}+\frac{1}{(n+1)!}\int_{-1}^{z}(\partial_{z}^{n+2}f)(z^{\prime})(z-z^{\prime})^{n+1}dz^{\prime}.
    \end{equation*}
    Then, we have \begin{align*}
        \left\vert\partial_{z}^{n}\left(\frac{1}{(z+1)^{2}}f\right)\right\vert\lesssim \sum_{i = 0}^{n}(z+1)^{-(2+i)}\int_{-1}^{z}\left\vert\partial_{z}^{n+2}f(z^{\prime})\right\vert(z-z^{\prime})^{i+1}dz^{\prime}\lesssim \left\Vert\partial_{z}^{n+2}f\right\Vert_{L_{z}^{\infty}([-1,-\frac{1}{2}])}.
    \end{align*}
\end{proof}
Using the estimates in Proposition \ref{prop: nonlinear structure equation}, we have the following lemma estimating $\widetilde{\mathcal{N}}^{(1)}$ and $\widetilde{\mathcal{N}}^{(2)}$.
\begin{lemma}
For $(f,g,h)\in\mathcal{S}^{a,N,\tau}$, $\widetilde{\mathcal{N}}^{(1)}$ and $\widetilde{\mathcal{N}}^{(2)}$ have the following estimates \begin{align}
    \left\vert(-z)^{\max\{0,j-\frac{3}{4}+\tau\}}\partial_{z}^{j}\widetilde{\mathcal{N}}^{(1)}(f,g,h)\right\vert\lesssim &e^{-(1+2a)s}\left\Vert(f,g,h)\right\Vert_{\mathcal{S}^{a,N,\tau}}^{2},\quad 0\leq j\leq N-1,\label{eq: pointwise estimate for renormalized N1}\\
     \left\vert(-z)^{\max\{0,j-\frac{3}{4}+\tau\}}\partial_{z}^{j}\widetilde{\mathcal{N}}^{(2)}(f,g,h)\right\vert\lesssim &e^{-(1+2a)s}\left\Vert(f,g,h)\right\Vert_{\mathcal{S}^{a,N,\tau}}^{2},\quad 0\leq j\leq N-1,\label{eq: pointwise estimate for renormalized N2}\\
     \sum_{i = 1}^{N}\int_{-1}^{0}(-z)^{2i-\frac{5}{2}+2\tau}\left\vert\partial_{z}^{i}\widetilde{\mathcal{N}}^{(1)}\right\vert^{2}\lesssim& e^{-2(1+2a)s}\left\Vert(f,g,h)\right\Vert_{\mathcal{S}^{a,N,\tau}}^{4},\label{eq: L2 estimate for renormalized N1}\\
     \sum_{i = 1}^{N}\int_{-1}^{0}(-z)^{2i-\frac{5}{2}+2\tau}\left\vert\partial_{z}^{i}\widetilde{\mathcal{N}}^{(2)}\right\vert^{2}\lesssim& e^{-2(1+2a)s}\left\Vert(f,g,h)\right\Vert_{\mathcal{S}^{a,N,\tau}}^{4}.\label{eq: L2 estimate for renormalized N2}
\end{align}
\end{lemma}
\begin{proof}
    We only prove \eqref{eq: pointwise estimate for renormalized N1} and \eqref{eq: L2 estimate for renormalized N1}. For the zeroth-order pointwise estimate, using \eqref{eq: quadratic structure of the nonlinear terms in the fixed point argument for N1}-\eqref{eq: quadratic structure for N3} for $j = 0$, we have \begin{align*}
        \left\Vert\widetilde{\mathcal{N}}^{(1)}\right\Vert_{L_{z}^{\infty}([-1,0])}\lesssim& \left\Vert\mathcal{N}^{(1)}\right\Vert_{L_{z}^{\infty}([-1,0])}+\left\Vert\frac{1}{\mr{r}^{2}}\int_{-1}^{z}\mr{r}^{2}\frac{\mathcal{N}^{(3)}}{\mr{r}^{2}}\right\Vert_{L_{z}^{\infty}([-1,0])}\\\lesssim&
        \left\Vert\mathcal{N}^{(1)}\right\Vert_{L_{z}^{\infty}([-1,0])}+\left\Vert\frac{\mathcal{N}^{(3)}}{\mr{r}^{2}}\right\Vert_{L_{z}^{\infty}([-1,0])}\lesssim e^{-(1+2a)s}\left\Vert(f,g,h)\right\Vert_{\mathcal{S}^{a,N,\tau}}^{2}.
    \end{align*}
    For the higher-order pointwise estimate ($j\geq 1$), using Proposition~\ref{prop: estimate on the exact k self similar spacetime} on the background $k$-self-similar spacetime and \eqref{eq: quadratic structure of the nonlinear terms in the fixed point argument for N1}-\eqref{eq: quadratic structure for N3}, we have \begin{align*}
        \left\Vert(-z)^{j-\frac{3}{4}+\tau}\partial_{z}^{j}\widetilde{\mathcal{N}}^{(1)}\right\Vert_{L_{z}^{\infty}([-1,0])}\lesssim &\left\Vert(-z)^{j-\frac{3}{4}+\tau}\partial_{z}^{j}{\mathcal{N}}^{(1)}\right\Vert_{L_{z}^{\infty}([-1,0])}+\sum_{0\leq l\leq j-1}\left\Vert(-z)^{\min\{0,l+1-p_{k}\}}\partial_{z}^{l}{\mathcal{N}}^{(3)}\right\Vert_{L_{z}^{\infty}([-1,0])}\\&+\sum_{i = 0}^{j}\left\Vert\partial_{z}^{i}\left(\frac{1}{(z+1)^{2}}\int_{-1}^{z}(\widetilde{z}+1)^{2}\frac{\mathcal{N}^{(3)}}{\mr{r}^{2}}\right)d\widetilde{z}\right\Vert_{L_{z}^{\infty}([-1,-\frac{1}{2}])}.
    \end{align*}
    For the last term in the above inequality, assuming that for $0\leq i\leq j-1$ and using Lemma~\ref{lemma: Using Talor theory to control the Linfty norm}, we have \begin{align*}
        \sum_{i = 0}^{j-1}\left\Vert\partial_{z}^{i}\left(\frac{1}{(z+1)^{2}}\int_{-1}^{z}(\widetilde{z}+1)^{2}\frac{\mathcal{N}^{(3)}}{\mr{r}^{2}}d\widetilde{z}\right)\right\Vert_{L_{z}^{\infty}([-1,-\frac{1}{2}])}\lesssim \sum_{i = 0}^{j}\left\Vert\partial_{z}^{i}\left(\frac{\mathcal{N}^{(3)}}{\mr{r}^{2}}\right)\right\Vert_{L_{z}^{\infty}([-1,-\frac{1}{2}])},
    \end{align*}
    For $i = j$, we have 
    \begin{align*}
        &\left\Vert\partial_{z}^{j}\left(\frac{1}{(z+1)^{2}}\int_{-1}^{z}(\widetilde{z}+1)^{2}\frac{\mathcal{N}^{(3)}}{\mr{r}^{2}}\right)\right\Vert_{L_{z}^{\infty}([-1,-\frac{1}{2}])}\\\lesssim& \left\Vert\partial_{z}^{j-1}\left(\frac{\mathcal{N}^{(3)}}{\mr{r}^{2}}\right)\right\Vert_{L_{z}^{\infty}([-1,-\frac{1}{2}])}+\left\Vert\partial_{z}^{j-1}\left(\frac{1}{(z+1)^{3}}\int_{-1}^{z}(\widetilde{z}+1)^{2}\frac{\mathcal{N}^{(3)}}{\mr{r}^{2}}\right)\right\Vert_{L_{z}^{\infty}([-1,-\frac{1}{2}])}\\\lesssim&
        \left\Vert\partial_{z}^{j-1}\left(\frac{\mathcal{N}^{(3)}}{\mr{r}^{2}}\right)\right\Vert_{L_{z}^{\infty}([-1,-\frac{1}{2}])}+\left\Vert\partial_{z}^{j-1}\left(\frac{1}{(z+1)^{3}}\int_{-1}^{z}(\widetilde{z}+1)^{3}\partial_{z}\left(\frac{\mathcal{N}^{(3)}}{\mr{r}^{2}}\right)\right)\right\Vert_{L_{z}^{\infty}([-1,-\frac{1}{2}])},
    \end{align*}
    where the second inequality follows from the integration by parts. Arguing inductively, we have \begin{align*}
        &\left\Vert\partial_{z}^{j}\left(\frac{1}{(z+1)^{2}}\int_{-1}^{z}(\widetilde{z}+1)^{2}\frac{\mathcal{N}^{(3)}}{\mr{r}^{2}}\right)\right\Vert\\\lesssim &\sum_{i = 0}^{j-1}\left\Vert\partial_{z}^{i}\left(\frac{\mathcal{N}^{(3)}}{\mr{r}^{2}}\right)\right\Vert_{L_{z}^{\infty}([-1,-\frac{1}{2}])}+\left\Vert\frac{1}{(z+1)^{j+2}}\int_{-1}^{z}(z+1)^{j+2}\partial_{z}^{j}\left(\frac{\mathcal{N}^{(3)}}{\mr{r}^{2}}\right)\right\Vert_{L_{z}^{\infty}([-1,-\frac{1}{2}])}\\\lesssim&\sum_{i = 0}^{j}\left\Vert\partial_{z}^{i}\left(\frac{\mathcal{N}^{(3)}}{\mr{r}^{2}}\right)\right\Vert_{L_{z}^{\infty}([-1,-\frac{1}{2}])}.
    \end{align*}
Using \eqref{eq: quadratic structure of the nonlinear terms in the fixed point argument for N1}--\eqref{eq: quadratic structure for N3}, we can conclude the proof of \eqref{eq: pointwise estimate for renormalized N1}.

For the $L^{2}$-estimate \eqref{eq: L2 estimate for renormalized N1}, we have \begin{align*}
        \int_{-1}^{0}(-z)^{2j-\frac{5}{2}+2\tau}\left\vert\partial_{z}^{j}\widetilde{\mathcal{N}}^{(1)}\right\vert^{2}\lesssim&\int_{-1}^{0}(-z)^{2j-\frac{5}{2}+2\tau}\left\vert\partial_{z}^{j}\mathcal{N}^{(1)}\right\vert^{2}+\sum_{i = 0}^{j-1}\int_{-1}^{0}(-z)^{2i+2\tau-\frac{1}{2}}\left\vert\partial_{z}^{i}\mathcal{N}^{(3)}\right\vert^{2}\\&+\sum_{i = 0}^{j}\left\Vert \partial_{z}^{i}\left(\frac{1}{(z+1)^{2}}\int_{-1}^{z}(\widetilde{z}+1)^{2}\frac{\mathcal{N}^{(3)}}{\mr{r}^{2}}\right)d\widetilde{z}\right\Vert_{L_{z}^{\infty}([-1,-\frac{1}{2}])}^{2}\\\lesssim& 
        e^{-2(1+2a)s}\left\Vert(f,g,h)\right\Vert_{\mathcal{S}^{a,N,\tau}}^{4}.
    \end{align*}
\end{proof}
\begin{lemma}
Given $(f,g,h)\in \mathcal{S}^{a,N,\tau}$, then $(r_{p}^{s_{1}},\Psi_{p}^{s_{1}})(s_{1},z)\in \mathcal{C}_{N}^{\frac{7}{4}-\tau}\times \mathcal{C}_{N}^{\alpha-\tau}$. Moreover, we have that \begin{equation}
    \left\Vert r_{p}^{s_{1}}(s_{1},\cdot)\right\Vert_{\mathcal{C}_{N}^{\frac{7}{4}-\tau}}+\left\Vert \Psi_{p}^{s_{1}}(s_{1},\cdot)\right\Vert_{\mathcal{C}_{N}^{\alpha-\tau}}\lesssim \left\Vert(f,g,h)\right\Vert_{\mathcal{S}^{a,N,\tau}}^{2}e^{-(1+2a)s_{1}}.
\end{equation}
\end{lemma}
\begin{proof}
By the explicit form of $r_{p}^{s_{1}}(s_{1},z)$ \eqref{eq: initial data r for duhamel principle} and the estimate \eqref{eq: pointwise estimate for renormalized N1}, we have 
\begin{equation}
\begin{aligned}
    \left\Vert r_{p}^{s_{1}}(s_{1},\cdot)\right\Vert_{L^{\infty}_{z}([-1,0])}\lesssim \left\Vert \widetilde{\mathcal{N}}^{(1)}\right\Vert_{L_{z}^{\infty}([-1,0])}\lesssim  \left\Vert(f,g,h)\right\Vert_{\mathcal{S}^{a,N,\tau}}^{2}e^{-(1+2a)s_{1}},
\end{aligned}
\end{equation}
where we have used Lemma \ref{lemma: Using Talor theory to control the Linfty norm} in the third inequality above.

Using \eqref{eq: duhamel initial data for rp}, we have \begin{equation}
\begin{aligned}
    \left\Vert \frac{d}{dz}r_{p}^{s_{1}}(s_{1},\cdot)\right\Vert_{L^{\infty}_{z}([-1,0])}\lesssim&\left\Vert r_{p}^{s_{1}}(s_{1},\cdot)\right\Vert_{L_{z}^{\infty}([-1,0])}+\left\Vert\widetilde{\mathcal{N}}^{(1)}\right\Vert_{L_{z}^{\infty}([-1,0])}\\\lesssim &\left\Vert(f,g,h)\right\Vert_{\mathcal{S}^{a,N,\tau}}^{2}e^{-(1+2a)s_{1}}.
    \end{aligned}
\end{equation}
For $\frac{d^{2}r_{p}^{s_{1}}}{dz^{2}}(s_{1},z)$, commuting \eqref{eq: duhamel initial data for rp} with $\partial_{z}$, we have \begin{equation}
    \begin{aligned}
        \left\Vert(-z)^{\frac{1}{4}+\tau}\frac{d^{2}}{dz^{2}}r_{p}^{s_{1}}(s_{1},z)\right\Vert_{L_{z}^{\infty}([-1,0])}\lesssim& \sum_{i = 0}^{1}\left\Vert\frac{d^{i}r_{p}^{s_{1}}}{dz^{i}}\right\Vert_{L_{z}^{\infty}([-1,0])}+\left\Vert(-z)^{\frac{1}{4}+\tau}\partial_{z}\widetilde{\mathcal{N}}^{(1)}\right\Vert_{L_{z}^{\infty}([-1,0])}\\\lesssim& \left\Vert(f,g,h)\right\Vert_{\mathcal{S}^{a,N,\tau}}^{2}e^{-(1+2a)s_{1}}
    \end{aligned}
\end{equation}

Arguing inductively and using \eqref{eq: pointwise estimate for renormalized N1}, we have\begin{equation}
    \left\Vert r_{p}^{s_{1}}(s_{1},\cdot)\right\Vert_{\mathcal{C}_{N}^{\frac{7}{4}-\tau}}\lesssim \left\Vert(f,g,h)\right\Vert_{\mathcal{S}^{a,N,\tau}}^{2}e^{-(1+2a)s_{1}}.
\end{equation}
Similarly, for $\Psi_{p}^{s_{1}}$, we have that \begin{equation}
\left\Vert\Psi_{p}^{s_{1}}(s_{1},\cdot)\right\Vert_{\mathcal{C}_{N}^{\alpha-\tau}}\lesssim\left\Vert(f,g,h)\right\Vert_{\mathcal{S}^{a,N,\tau}}^{2}e^{-(1+2a)s_{1}}.
\end{equation}
This concludes the proof.
\end{proof}

Then by Theorem \ref{thm: linearized result}, for $\tau\in(0,\alpha-\alpha^{\prime})$ and $s\geq s_{1}$, we have that \begin{align*}
    \sum_{0\leq i+j\leq1}\left\vert\partial_{s}^{i}\partial_{z}^{j}\left(\Psi_{p}^{s_{1}}-c_{\infty}^{s_{1}}r_{k}\right)\right\vert\leq& C \left(\left\Vert r_{p}^{s_{1}}(s_{1},\cdot)\right\Vert_{\mathcal{C}_{N}^{\frac{7}{4}-\tau}}+\left\Vert\Psi_{p}^{s_{1}}(s_{1},\cdot)\right\Vert_{\mathcal{C}_{N}^{\alpha-\tau}}\right)e^{-\alpha^{\prime}q_{k}(s-s_{1})},\\
    \sum_{0\leq i+j\leq1}\left\vert\partial_{s}^{i}\partial_{z}^{j}r_{p}^{s_{1}}\right\vert\leq& C\left(\left\Vert r_{p}^{s_{1}}(s_{1},\cdot)\right\Vert_{\mathcal{C}_{N}^{\frac{7}{4}-\tau}}+\left\Vert\Psi_{p}^{s_{1}}(s_{1},\cdot)\right\Vert_{\mathcal{C}_{N}^{\alpha-\tau}}\right)e^{-\alpha^{\prime}q_{k}(s-s_{1})},
\end{align*}
where $c_{\infty}^{s_{1}}$ satisfies the estimate \begin{equation}
    \left\vert c_{\infty}^{s_{1}}\right\vert\lesssim \left\Vert r_{p}^{s_{1}}(s_{1},\cdot)\right\Vert_{\mathcal{C}_{N}^{\frac{7}{4}-\tau}}+\left\Vert\Psi_{p}^{s_{1}}(s_{1},\cdot)\right\Vert_{\mathcal{C}_{N}^{\alpha-\tau}}\lesssim \left\Vert(f,g,h)\right\Vert_{\mathcal{S}^{a,N,\tau}}^{2}e^{-(1+2a)s_{1}}.
\end{equation}
Let $1+2a<\alpha^{\prime}q_{k}$. Then by the Duhamel principle,  for $\mathfrak{A}_{2}^{(1)}$, we have \begin{equation}
\begin{aligned}
    \sum_{0\leq i+j\leq1}\left\Vert\partial_{s}^{i}\partial_{z}^{j}\mathfrak{A}_{2}^{(1)}(s,\cdot)\right\Vert_{L^{\infty}_{z}} =&\sum_{0\leq i+j\leq 1}\left\Vert\partial_{s}^{i}\partial_{z}^{j}\int_{0}^{s}r_{p}^{s_{1}}(s,\cdot)ds_{1}\right\Vert_{L_{z}^{\infty}}\\\leq&\sum_{0\leq j\leq 1}\left\Vert\partial_{z}^{j}r_{p}^{s}(s,\cdot)\right\Vert_{L_{z}^{\infty}}+
    \sum_{0\leq i+j\leq 1}\left\Vert\int_{0}^{s}\partial_{s}^{i}\partial_{z}^{j}r_{p}^{s_{1}}(s,\cdot)ds_{1}\right\Vert_{L_{z}^{\infty}}\\\lesssim&\left\Vert(f,g,h)\right\Vert_{\mathcal{S}^{a,N,\tau}}^{2}\left(e^{-(1+2a)s}+\int_{0}^{s}e^{-\alpha^{\prime}q_{k}s}e^{\alpha^{\prime}q_{k}s_{1}}e^{-(1+2a)s_{1}}ds_{1}\right)\\\lesssim& \left\Vert(f,g,h)\right\Vert_{\mathcal{S}^{a,N,\tau}}^{2}e^{-(1+2a)s}.
    \end{aligned}
\end{equation}
For $\mathfrak{A}_{2}^{(2)}$, we have \begin{equation}
    \begin{aligned}
        \sum_{0\leq i+j\leq 1}\left\Vert\partial_{s}^{i}\partial_{z}^{j}\left(\mathfrak{A}_{2}^{(2)}-r_{k}\int_{0}^{\infty}c_{\infty}^{s_{1}}ds_{1}\right)\right\Vert_{L_{z}^{\infty}} \leq& \sum_{0\leq i+j\leq 1}\left\Vert\partial_{s}^{i}\partial_{z}^{j}\int_{0}^{s}\Psi_{p}^{s_{1}}(s,z)-c_{\infty}^{s_{1}}r_{k}(s,z)ds_{1}\right\Vert_{L_{z}^{\infty}}\\&+\sum_{0\leq i+j\leq 1}\left\Vert \partial_{s}^{i}\partial_{z}^{j}\int_{s}^{\infty}c_{\infty}^{s_{1}}r_{k}(s,z)ds_{1}\right\Vert_{L_{z}^{\infty}}\\\lesssim&\left\Vert(f,g,h)\right\Vert_{\mathcal{S}^{a,N,\tau}}^{2}\left(e^{-(1+2a)s}+e^{-(2+2a)s}\right)\\\lesssim&\left\Vert(f,g,h)\right\Vert_{\mathcal{S}^{a,N,\tau}}^{2}e^{-(1+2a)s}.
    \end{aligned}
\end{equation}
Hence, modifying the trivial initial data of \eqref{eq: nonlinear equation 2} to be $(r_{p},\Psi_{p})|_{\Sigma_{-1}^{(in)}} = (0,-r_{k}\int_{0}^{\infty}c_{\infty}^{s_{1}}ds_{1})$, the solution $\mathfrak{A}_{2}$ to \eqref{eq: nonlinear equation 2} has the estimate \begin{equation}
    \sum_{0\leq i+j\leq 1}e^{(1+2a)s}\left\Vert\partial_{s}^{i}\partial_{z}^{j}\mathfrak{A}_{2}^{(1)}\right\Vert_{L_{z}^{\infty}}+\left\Vert\partial_{s}^{i}\partial_{z}^{j}\mathfrak{A}_{2}^{(2)}\right\Vert_{L_{z}^{\infty}}\lesssim \left\Vert(f,g,h)\right\Vert_{\mathcal{S}^{a,N,\tau}}^{2},\quad 1+2a<\alpha^{\prime}q_{k}.
\end{equation}
It remains to estimate $\mathfrak{A}_{2}^{(3)}$. Using the equation \begin{equation*}
    \partial_{z}m_{p}+\frac{r_{k}}{\partial_{z}r_{k}}(\partial_{z}\phi_{k})^{2}m_{p}+S_{3}[r_{p},\Psi_{p}] = \mathcal{N}^{(3)}(f,g,h),
\end{equation*}
we have \begin{equation}
    \sum_{0\leq i+j\leq 1}e^{(1+2a)s}\left\vert\partial_{s}\partial_{z}\mathfrak{A}_{2}^{(3)}\right\vert\lesssim \left\Vert(f,g,h)\right\Vert_{\mathcal{S}^{a,N,\tau}}^{2}.
\end{equation}
We can conclude the discussions in these two sections with the following proposition:
\begin{proposition}
\label{prop: lower order energy estimate}
    Given functions $(r_{p}^{0},\Psi_{p}^{0})\in \mathcal{C}_{N+1}^{\frac{7}{4}}\times\mathcal{C}_{N+1}^{\alpha}$ with $\alpha\in(p_{k},\frac{3}{2})$ and $(f,g,h)\in\mathcal{S}^{a,N,\tau}_{\delta}$ with $a\in\left(0,\frac{\alpha^{\prime}q_{k}-1}{2}\right)$, there exists a constant $\widetilde{c}_{\infty}$ satisfying \begin{equation}
        \vert \widetilde{c}_{\infty}\vert\lesssim\left\Vert r_{p}^{0}\right\Vert_{\mathcal{C}_{N}^{\frac{7}{4}}}+ \left\Vert\Psi_{p}^{0}\right\Vert_{\mathcal{C}_{N}^{\alpha}}+\left\Vert(f,g,h)\right\Vert_{\mathcal{S}^{a,N,\tau}}^{2},
    \end{equation}
    such that the solution $\mathfrak{A}$ to $\mathcal{P}(r_{p},\Psi_{p},m_{p}) = \mathcal{N}(f,g,h)$ with the initial data $(r_{p},\Psi_{p})|_{\Sigma_{-1}^{(in)}} = (r_{p}^{0},\Psi_{p}^{0}-\widetilde{c}_{\infty}r_{k}-\widetilde{c}_{\infty}r_{p}^{0})$ satisfies the estimate \begin{equation}
        \sum_{0\leq i+j\leq 1}\sum_{l = 1}^{3}e^{(1+2a)s}\left\Vert\partial_{s}^{i}\partial_{z}^{j}\mathfrak{A}^{(l)}(s,\cdot)\right\Vert_{L_{z}^{\infty}}\lesssim\left\Vert r_{p}^{0}\right\Vert_{\mathcal{C}_{N}^{\frac{7}{4}}}+ \left\Vert\Psi_{p}^{0}\right\Vert_{\mathcal{C}_{N}^{\alpha}}+\left\Vert(f,g,h)\right\Vert_{\mathcal{S}^{a,N,\tau}}^{2}.\label{eq: total lower order fixed point argument}
    \end{equation}
    Moreover, in the near-center region $\{-1\leq z\leq -e^{-2q_{k}}\}$, we have that \begin{equation*}
         \sum_{0\leq i+j\leq 2}\sum_{l = 1}^{3}e^{(1+2a)s}\left\vert\partial_{s}^{i}\partial_{z}^{j}\mathfrak{A}^{(l)}(s,z)\right\vert\lesssim\left\Vert r_{p}^{0}\right\Vert_{\mathcal{C}_{N}^{\frac{7}{4}}}+ \left\Vert\Psi_{p}^{0}\right\Vert_{\mathcal{C}_{N}^{\alpha}}+\left\Vert(f,g,h)\right\Vert_{\mathcal{S}^{a,N,\tau}}^{2}.
    \end{equation*}
\end{proposition}

\subsection{Second-order energy estimate on $\mathfrak{A}$}
In this section, we establish the second-order energy estimate on $\mathfrak{A}$, using the established lower-order pointwise estimates. First, we establish the estimates for $\frac{(m_{p})_{L}}{\mr{r}^{2}}$, where $(m_{p})_{L}$ is defined in \eqref{eq: definition of mpl}. Let $(\phi_{p})_{L}$ be the linear dependence of $\phi_{p}$ on $r_{p}$ and $\Psi_{p}$.
\begin{proposition}
    Assuming $(r_{p},\Psi_{p})$ satisfies the estimate \eqref{eq: total lower order fixed point argument} in Proposition \ref{prop: lower order energy estimate}, then we have the following two estimates for $(m_{p})_{L}$:\begin{align}
       e^{(1+2a)s} \left\Vert\frac{(m_{p})_{L}}{\mr{r}^{2}}(s,\cdot)\right\Vert_{L_{z}^{\infty}}&\lesssim \left\Vert r_{p}^{0}\right\Vert_{\mathcal{C}_{N}^{\frac{7}{4}}}+\left\Vert\Psi_{p}^{0}\right\Vert_{\mathcal{C}_{N}^{\alpha}}+\left\Vert(f,g,h)\right\Vert_{\mathcal{S}^{a,N,\tau}}^{2},\label{eq: first order mp estimate in the nonlinear problem}\\e^{(1+2a)s}\left\Vert\partial_{z}\left(\frac{(m_{p})_{L}}{\mr{r}^{2}}\right)(s,\cdot)\right\Vert_{L_{z}^{\infty}}&\lesssim \left\Vert r_{p}^{0}\right\Vert_{\mathcal{C}_{N}^{\frac{7}{4}}}+\left\Vert\Psi_{p}^{0}\right\Vert_{\mathcal{C}_{N}^{\alpha}}+\left\Vert(f,g,h)\right\Vert_{\mathcal{S}^{a,N,\tau}}^{2}.\label{eq: second order mp estimate in the nonlinear problem}
    \end{align}
\end{proposition}
\begin{proof}
    Note that $(m_{p})_{L}$ satisfies the equation 
    \begin{equation}
    \begin{aligned}
        &\partial_{z}(m_{p})_{L}+\frac{r_{k}}{\partial_{z}r_{k}}(\partial_{z}\phi_{k})^{2}(m_{p})_{L} \\=&\left[\frac{r_{k}}{\partial_{z}r_{k}}(\partial_{z}\phi_{k})^{2}-\frac{m_{k}(\partial_{z}\phi_{k})^{2}}{\partial_{z}r_{k}}\right]r_{p}+\left[\frac{r_{k}m_{k}}{(\partial_{z}r_{k})^{2}}(\partial_{z}\phi_{k})^{2}-\frac{1}{2}\frac{r_{k}^{2}(\partial_{z}\phi_{k})^{2}}{(\partial_{z}r_{k})^{2}}\right]\partial_{z}r_{p}\\&+\left[\frac{r_{k}^{2}}{\partial_{z}r_{k}}(\partial_{z}\phi_{k})-\frac{2r_{k}m_{k}}{\partial_{z}r_{k}}(\partial_{z}\phi_{k})\right]\partial_{z}\left((\phi_{p})_{L}\right),
    \end{aligned}
    \label{eq: transport equation for mp in the nonlinear setting}
    \end{equation}
    where $(\phi_{p})_{L}: = \frac{\Psi_{p}-r_{p}\mr{\phi}}{r_{k}}$. By Proposition \ref{prop: lower order energy estimate}, we have that \begin{equation*}
        e^{2as}\left\vert\partial_{z}(\phi_{p})_{L}\right\vert\lesssim \left\Vert r_{p}^{0}\right\Vert_{\mathcal{C}_{N}^{\frac{7}{4}}}+\left\Vert\Psi_{p}^{0}\right\Vert_{\mathcal{C}_{N}^{\alpha}}+\left\Vert(f,g,h)\right\Vert_{\mathcal{S}^{a,N,\tau}}^{2}.
    \end{equation*}
Hence, using the integration factor to write down the explicit form of $(m_{p})_{L}$, we have 
\begin{equation}
\label{eq: zero order estimate on mpl}
\begin{aligned}
    e^{(1+2a)s}\left\vert\frac{(m_{p})_{L}}{\mr{r}^{3}}\right\vert\lesssim& e^{(1+2a)s}\left(\left\Vert \frac{r_{p}}{r_{k}}\right\Vert_{L_{z}^{\infty}}+\left\Vert\partial_{z}r_{p}\right\Vert_{L_{z}^{\infty}}+e^{-s}\left\Vert\partial_{z}\left((\phi_{p})_{L}\right)\right\Vert_{L_{z}^{\infty}}\right)\\\lesssim &\left\Vert r_{p}^{0}\right\Vert_{\mathcal{C}_{N}^{\frac{7}{4}}}+\left\Vert\Psi_{p}^{0}\right\Vert_{\mathcal{C}_{N}^{\alpha}}+\left\Vert(f,g,h)\right\Vert_{\mathcal{S}^{a,N,\tau}}^{2}.
\end{aligned}
\end{equation}
Using the equation for $(m_{p})_{L}$, we have \begin{equation}
\label{eq: first order estimate on mpl}
\begin{aligned}
    e^{(1+2a)s}\left\vert\partial_{z}\left(\frac{(m_{p})_{L}}{\mr{r}^{2}}\right)\right\vert\lesssim& e^{(1+2a)s}\left(\left\vert\frac{\partial_{z}(m_{p})_{L}}{\mr{r}^{2}}\right\vert+\left\vert\frac{(m_{p})_{L}}{\mr{r}^{3}}\right\vert\right)\\\lesssim& \left\Vert r_{p}^{0}\right\Vert_{\mathcal{C}_{N}^{\frac{7}{4}}}+\left\Vert\Psi_{p}^{0}\right\Vert_{\mathcal{C}_{N}^{\alpha}}+\left\Vert(f,g,h)\right\Vert_{\mathcal{S}^{a,N,\tau}}^{2}.
\end{aligned}
\end{equation}
\end{proof}
To get the second-order energy estimates for $(r_{p},\Psi_{p})$, we first prove the following Hardy inequalities. \begin{lemma}
\label{lemma: Hardy inequality}
    We have the following estimates on $(r_{p},\Psi_{p})$ for a universal constant $C$ independent of $s$: \begin{align}
        \int_{-1}^{0}r_{p}^{2}(s,z)dz\leq C\int_{-1}^{0}(\partial_{z}r_{p})^{2}(s,z)dz,\quad \int_{-1}^{0}\Psi_{p}^{2}(s,z)dz\leq C\int_{-1}^{0}(\partial_{z}\Psi_{p})^{2}dz.
    \end{align}
\end{lemma}
\begin{proof}
    Since $r_{p}(s,-1) = 0$, by the fundamental theorem of calculus, we have \begin{equation*}
        r_{p}(s,z) = \int_{-1}^{z}\partial_{z}r_{p}(s,z^{\prime})dz^{\prime}.
    \end{equation*}
    Then we have \begin{equation*}
        \int_{-1}^{0}(r_{p})^{2}(s,z)dz\leq C\int_{-1}^{0}\int_{-1}^{0}(\partial_{z}r_{p})^{2}(s,z^{\prime})dz^{\prime}dz\leq C\int_{-1}^{0}(\partial_{z}r_{p})^{2}(s,z)dz.
    \end{equation*}
    The proof of $\Psi_{p}$ will follow similarly. We conclude the proof of this lemma.
\end{proof}

We can prove the following proposition:
\begin{proposition}
\label{prop: second order energy estimate in the nonlinear problem}
    Given functions $(r_{p}^{0},\Psi_{p}^{0})\in\mathcal{C}_{N+1}^{\frac{7}{4}}\times \mathcal{C}_{N+1}^{\alpha}$ and $(f,g,h)\in\mathcal{S}_{\delta}^{a,N,\tau}$ with $a\in\left(0,\frac{\alpha^{\prime}q_{k}-1}{2}\right)$. Define $a_{\lambda} = \lambda a+(1-\lambda) (2a)$. Then the solution $\mathfrak{A} = (r_{p},\Psi_{p},(m_{p})_{L})$ to the equation $\mathcal{P}(r_{p},\Psi_{p},(m_{p})_{L}) = \widetilde{\mathcal{N}}(f,g,h)$ with the initial data $(r_{p},\Psi_{p})|_{\Sigma_{-1}^{(in)}} = (r_{p}^{0},\Psi_{p}^{0}-\widetilde{c}_{\infty}r_{k}-\widetilde{c}_{\infty}r_{p}^{0})$   satisfies the following energy estimate:
    \begin{equation}
        \begin{aligned}
            &e^{2(1+a_{\frac{1}{100}})s_{0}}\left(\int_{s = s_{0}}(-z)^{2\tau-\frac{1}{2}}\left(\partial_{z}^{2}r_{p}\right)^{2}+(-z)^{2\tau-\frac{1}{2}}(\partial_{s}\partial_{z}r_{p})^{2}+(-z)^{2\gamma}\left(\partial_{z}^{2}\Psi_{p}\right)^{2}\right)\\&\hspace{7cm}\lesssim \delta^{-\frac{2}{100}}\left(\left\Vert r_{p}^{0}\right\Vert_{\mathcal{C}_{N+1}^{\frac{7}{4}}}^{2}+\left\Vert\Psi_{p}^{0}\right\Vert_{\mathcal{C}_{N+1}^{\alpha}}^{2}+\left\Vert(f,g,h)\right\Vert_{\mathcal{S}^{a,N,\tau}}^{4}\right).
        \end{aligned}
    \end{equation}
\end{proposition}
\begin{proof}
    We first derive the second-order energy estimate for $r_{p}$. Let $r_{\lambda} = e^{(1+a_{\lambda})s}r_{p}$ and $\Psi_{\lambda} = e^{(1+a_{\lambda})s}\Psi_{p}$. Using the $\mathcal{P}(r_{p},\Psi_{p},(m_{p})_{L}) = \widetilde{\mathcal{N}}$ formulation and the schematic form \eqref{eq:second order wave equation for r rho}, we have \begin{equation}
       \begin{aligned}
&\partial_{s}\partial_{z}^{2}r_{\lambda}+q_{k}z\partial_{z}^{3}r_{\lambda}+\left(2q_{k}-(1+a_{\lambda})+G_{k^{2}}(z)\right)\partial_{z}^{2}r_{\lambda}\\=&F_{k^{2}}(z)\partial_{z}r_{\lambda}+F_{k^{2}}(z)\partial_{s}r_{\lambda}+F_{k^{2}}(z)r_{\lambda}+K(z)\partial_{z}\left(e^{(1+a_{\lambda})s}\frac{(m_{p})_{L}}{\mr{r}^{2}}\right)+\left(K(z)+F_{k^{2}}(z)\right)e^{(1+a_{\lambda})s}\frac{(m_{p})_{L}}{\mr{r}^{2}}\\&+e^{(1+a_{\lambda})s}\partial{z}\widetilde{\mathcal{N}}^{(1)}.
       \end{aligned}
    \end{equation}
    Multiplying by $(-z)^{2\tau-\frac{1}{2}}\partial_{z}^{2}r_{\lambda}$, and applying the integration by parts and the Cauchy--Schwarz inequality, we have \begin{equation}
        \begin{aligned}
            &\frac{1}{2}\int_{s = s_{0}}(-z)^{2\tau-\frac{1}{2}}\left(\partial_{z}^{2}r_{\lambda}\right)^{2}+\frac{1}{2}q_{k}\int_{\Gamma}\left(\partial_{z}^{2}r_{\lambda}\right)^{2}\\&+\left(\left(\frac{7}{4}-\tau\right)q_{k}-(1+a_{\lambda})-Ck^{2}-C\delta^{\frac{1}{100}}\right)\iint_{\mathcal{R}(s_{0})}(-z)^{2\tau-\frac{1}{2}}\left(\partial_{z}^{2}r_{\lambda}\right)^{2}\\\lesssim& k^{2}\delta^{-\frac{1}{100}}\int_{0}^{s_{0}}e^{2(1+a_{\lambda})s}\sum_{0\leq i+j\leq 1}\left\Vert \partial_{s}^{i}\partial_{z}^{j}r_{p}(s,\cdot)\right\Vert_{L_{z}^{\infty}}^{2}+\delta^{-\frac{1}{100}}\iint_{\mathcal{R}(s_{0})}e^{2(1+a_{\lambda})s}\sum_{0\leq j\leq 1}\left\vert\partial_{z}^{j}\left(\frac{(m_{p})_{L}}{\mr{r}^{2}}\right)\right\vert^{2}\\&+\delta^{-\frac{1}{100}}\iint_{\mathcal{R}(s_{0})}e^{2(1+a_{\lambda})s}(-z)^{2\tau-\frac{1}{2}}\left\vert\partial_{z}\widetilde{\mathcal{N}}^{(1)}\right\vert^{2}dzds+\left\Vert r_{p}^{0}\right\Vert_{H_{N+1}^{-\frac{1}{4}+\tau}}^{2}.
        \end{aligned}
    \end{equation}
By Proposition \ref{prop: lower order energy estimate}, we have that \begin{equation*}
    \delta^{-\frac{1}{100}}\int_{0}^{s_{0}}e^{2(1+a_{\lambda})s}\sum_{0\leq i+j\leq 1}\left\Vert\partial_{s}^{i}\partial_{z}^{j}r_{p}(s,\cdot)\right\Vert_{L_{z}^{\infty}}^{2}dz\lesssim \delta^{-\frac{1}{100}}\left(\left\Vert r_{p}^{0}\right\Vert_{\mathcal{C}_{N}^{\frac{7}{4}}}^{2}+\left\Vert\Psi_{p}^{0}\right\Vert_{\mathcal{C}_{N}^{\alpha}}^{2}+\left\Vert(f,g,h)\right\Vert_{\mathcal{S}^{a,N,\tau}}^{4}\right).
\end{equation*}
Using the estimates \eqref{eq: first order mp estimate in the nonlinear problem}-\eqref{eq: second order mp estimate in the nonlinear problem} for $\frac{(m_{p})_{L}}{\mr{r}^{2}}$, we have that \begin{equation*}
    \delta^{-\frac{1}{100}}\sum_{0\leq j\leq 1}\iint_{\mathcal{R}(s_{0})}e^{2(1+a_{\lambda})s}\left\vert\partial_{z}^{j}\left(\frac{(m_{p})_{L}}{\mr{r}^{2}}\right)\right\vert^{2}\lesssim\delta^{-\frac{1}{100}}\left(\left\Vert r_{p}^{0}\right\Vert_{\mathcal{C}_{N}^{\frac{7}{4}}}^{2}+\left\Vert\Psi_{p}^{0}\right\Vert_{\mathcal{C}_{N}^{\alpha}}^{2}+\left\Vert(f,g,h)\right\Vert_{\mathcal{S}^{a,N,\tau}}^{4}\right).
\end{equation*}
For the term with $\partial_{z}\mathcal{N}^{(1)}$, using \eqref{eq: L2 estimate for renormalized N1}, we have \begin{equation}
        \delta^{-\frac{1}{100}}\iint_{\mathcal{R}(s_{0})}e^{2(1+a_{\lambda})s}(-z)^{2\tau-\frac{1}{2}}\left\vert\partial_{z}\widetilde{\mathcal{N}}^{(1)}\right\vert^{2}\lesssim \delta^{-\frac{1}{100}}\left\Vert(f,g,h)\right\Vert_{\mathcal{S}^{a,N,\tau}}^{4}.
    \end{equation}
    Hence, putting everything together, we have \begin{equation}
    \begin{aligned}
       &\int_{s = s_{0}}(-z)^{2\tau-\frac{1}{2}}\left(\partial_{z}^{2}r_{\lambda}\right)^{2}+\left((\frac{7}{2}-2\tau)q_{k}-2(1+a_{\lambda})-Ck^{2}-C\delta^{\frac{1}{100}}\right)\iint_{\mathcal{R}(s_{0})}(-z)^{2\tau-\frac{1}{2}}(\partial_{z}^{2}r_{\lambda})^{2}\\\lesssim& \delta^{-\frac{1}{100}}\left(\left\Vert r_{p}^{0}\right\Vert_{\mathcal{C}_{N}^{\frac{7}{4}}}^{2}+\left\Vert\Psi_{p}^{0}\right\Vert_{\mathcal{C}_{N}^{\alpha}}^{2}+\left\Vert r_{p}^{0}\right\Vert_{H_{N+1}^{-\frac{1}{4}+\tau}}^{2}+\left\Vert(f,g,h)\right\Vert_{\mathcal{S}^{a,N,\tau}}^{4}\right)\\\lesssim& \delta^{-\frac{1}{100}}\left(\left\Vert r_{p}^{0}\right\Vert_{\mathcal{C}_{N+1}^{\frac{7}{4}}}^{2}+\left\Vert\Psi_{p}^{0}\right\Vert_{\mathcal{C}_{N+1}^{\alpha}}^{2}+\left\Vert(f,g,h)\right\Vert_{\mathcal{S}^{a,N,\tau}}^{4}\right).
        \end{aligned}
        \label{eq: nonlinear second order estimate for rp}
    \end{equation}
    Similarly, for the second-order energy estimate of $\Psi_{p}$, using the schematic form \eqref{second order wave equation for psi rho} and multiplying by $(-z)^{2\gamma}w(z)\partial_{z}^{2}\Psi_{\lambda}$, where $w(z) = -2(-z)^{\frac{1}{2}}+3$, we have \begin{equation}
        \begin{aligned}
            &\frac{1}{2}\int_{s = s_{0}}(-z)^{2\gamma}\left(\partial_{z}^{2}\Psi_{\lambda}\right)^{2}+\frac{1}{2}q_{k}\int_{\Gamma}\left(\partial_{z}^{2}\Psi_{\lambda}\right)^{2}\\&+\iint_{\mathcal{R}(s_{0})}\left(\left(\left(\frac{3}{2}-\gamma\right)q_{k}-(1+a_{\lambda})\right)w(z)+\frac{1}{2}q_{k}(-z)^{\frac{1}{2}}-C\delta^{\frac{1}{100}}\right)(-z)^{2\gamma}\left(\partial_{z}^{2}\Psi_{\lambda}\right)^{2}\\\lesssim&\left\Vert\Psi_{p}^{0}\right\Vert_{H_{N+1}^{\gamma}}+
            \delta^{-\frac{1}{100}}\int_{0}^{s_{0}} e^{2(1+a_{\lambda})s}\sum_{0\leq i+j\leq 1}\left(\left\Vert\partial_{s}^{i}\partial_{z}^{j}\Psi_{p}\right\Vert_{L_{z}^{\infty}}^{2}+\left\Vert\partial_{s}^{i}\partial_{z}^{j}r_{p}\right\Vert_{L_{z}^{\infty}}^{2}\right)\\&+\delta^{-\frac{1}{100}}\iint_{\mathcal{R}(s_{0})}e^{2(1+a_{\lambda})s}(-z)^{2\gamma-\frac{1}{2}}\left\vert\partial_{z}^{2}r_{p}\right\vert^{2}+\delta^{-\frac{1}{100}}\int_{0}^{s_{0}}e^{2(1+a_{\lambda})s}\sum_{0\leq j\leq 1}\left\Vert\partial_{z}^{j}\left(\frac{(m_{p})_{L}}{\mr{r}^{2}}\right)\right\Vert_{L_{z}^{\infty}}^{2}\\&+\delta^{-\frac{1}{100}}\iint_{\mathcal{R}(s_{0})}e^{2(1+a_{\lambda})s}(-z)^{2\gamma}\left\vert\partial_{z}\widetilde{\mathcal{N}}^{(2)}\right\vert^{2}.
        \end{aligned}
        \label{second order estimate for nonlinear psip}
    \end{equation}
For the terms on the right-hand side of \eqref{second order estimate for nonlinear psip}, using \eqref{eq: nonlinear second order estimate for rp} to bound the bulk term with $\partial_{z}^{2}r_{p}$, using the estimates \eqref{eq: first order mp estimate in the nonlinear problem}-\eqref{eq: second order mp estimate in the nonlinear problem} to bound the term with $\frac{m_{p}}{\mr{r}^{2}}$, using Proposition \ref{prop: lower order energy estimate} to bound the terms with lower-order derivatives of $r_{p}$ and $\Psi_{p}$, and using \eqref{eq: L2 estimate for renormalized N2} to bound the bulk term with $\mathcal{N}^{(2)}$, we have that \begin{equation}
    \begin{aligned}
        &\frac{1}{2}e^{2(1+a_{\lambda})s_{0}}\left(\int_{s = s_{0}}(-z)^{2\gamma}\left(\partial_{z}^{2}\Psi_{p}\right)^{2}+q_{k}\int_{\Gamma}\left(\partial_{z}^{2}\Psi_{p}\right)^{2}\right)\\&+\left(\left(\frac{3}{2}-\gamma\right)q_{k}-(1+a_{\lambda})-C\delta^{\frac{1}{100}}\right)e^{2(1+a_{\lambda})s_{0}}\iint_{\mathcal{R}(s_{0})}(-z)^{2\gamma}\left(\partial_{z}^{2}\Psi_{p}\right)^{2}\\\lesssim&\delta^{-\frac{2}{100}}\left(\left\Vert\Psi_{p}^{0}\right\Vert_{H_{N+1}^{\gamma}}^{2}+\left\Vert r_{p}^{0}\right\Vert_{\mathcal{C}_{N}^{\frac{7}{4}}}^{2}+\left\Vert\Psi_{p}^{0}\right\Vert_{\mathcal{C}_{N}^{\alpha}}^{2}+\left\Vert(f,g,h)\right\Vert_{\mathcal{S}^{a,N,\tau}}^{4}\right).
    \end{aligned}
    \label{eq: final version of the second order nonlinear estimate for psip}
\end{equation}

    Hence, combining \eqref{eq: nonlinear second order estimate for rp} and \eqref{eq: final version of the second order nonlinear estimate for psip}, we have that \begin{equation}
        \begin{aligned}
           \int_{s = s_{0}}(-z)^{2\tau-\frac{1}{2}}\left(\partial_{z}^{2}r_{\lambda}\right)^{2}+\int_{s = s_{0}}(-z)^{2\gamma}\left(\partial_{z}^{2}\Psi_{\lambda}\right)^{2}\lesssim \delta^{-\frac{2}{100}}\left(\left\Vert r_{p}^{0}\right\Vert_{\mathcal{C}_{N+1}^{\frac{7}{4}}}^{2}+\left\Vert\Psi_{p}^{0}\right\Vert_{\mathcal{C}_{N+1}^{\alpha}}^{2}+\left\Vert(f,g,h)\right\Vert_{\mathcal{S}^{a,N,\tau}}^{4}\right).
        \end{aligned}
    \end{equation}
Using the equation $\mathcal{P}^{(1)} = \widetilde{N}^{(1)}$, we can derive the weighted $L^{2}$-bound for $\partial_{s}\partial_{z}r_{p}$. Taking $\lambda = \frac{1}{100}$ concludes the proof.
\end{proof}

\subsection{Higher-order energy estimate on $\mathfrak{A}$}
\label{sec: higher-order estimate}
In this section, we establish the higher-order energy estimate inductively. We first prove the following two preliminary lemmas.
\begin{lemma}
\label{lemma: inductive hardy}
Given $f\in C^{n}([-1,0))$ with $\lim_{z\rightarrow -1}\frac{f(z)}{\left\vert(z+1)\right\vert^{i-\frac{1}{2}}} = 0$, then for any $\omega>0$ and $1\leq i\in\mathbb{N}$, we have that 
\begin{equation}
    \int_{-1}^{0}(-z)^{2\omega}\frac{f^{2}(z)}{(z+1)^{2i}}dz\lesssim_{\omega,i}\int_{-1}^{0}(-z)^{2\omega+2}\frac{(f^{\prime})^{2}}{(z+1)^{2(i-1)}}dz.
\end{equation}
\end{lemma}
\begin{proof}
    We split the integral into two parts:
    \begin{equation*}
        \int_{-1}^{0}(-z)^{2\omega}\frac{f^{2}}{(z+1)^{2i}}dz = \int_{-1}^{-\frac{1}{2}}(-z)^{2\omega}\frac{f^{2}}{(z+1)^{2i}}+\int_{-\frac{1}{2}}^{0}(-z)^{2\omega}\frac{f^{2}}{(z+1)^{2i}}.
    \end{equation*}
For the integration over $[-1,-\frac{1}{2}]$, we have that \begin{equation}
    \begin{aligned}
        \int_{-1}^{-\frac{1}{2}}(-z)^{2\omega}\frac{f^{2}}{(z+1)^{2i}}dz\lesssim_{\omega,i}&-\int_{-1}^{-\frac{1}{2}}f^{2}d(z+1)^{-2i+1}\\\lesssim_{\omega,i}&\left(\int_{-1}^{-\frac{1}{2}}(-z)^{2\omega}\frac{f^{2}}{(z+1)^{2i}}\right)^{\frac{1}{2}}\left(\int_{-1}^{-\frac{1}{2}}(-z)^{2\omega+2}\frac{(f^{\prime})^{2}}{(z+1)^{2(i-1)}}\right)^{\frac{1}{2}}.
    \end{aligned}
\end{equation}
For the integration over $[-\frac{1}{2},0]$, we have that \begin{equation}
    \begin{aligned}
        \int_{-\frac{1}{2}}^{0}(-z)^{2\omega}\frac{f^{2}}{(z+1)^{2i}}\lesssim_{\omega,i}\int_{-1}^{0}(-z)^{2\omega}f^{2}dz\lesssim_{\omega,i} \left(\int_{-1}^{0}(-z)^{2\omega}f^{2}\right)^{\frac{1}{2}}\left(\int_{-1}^{0}(-z)^{2\omega+2}(f^{\prime})^{2}\right)^{\frac{1}{2}}.
    \end{aligned}
\end{equation}
Hence, we can conclude that \begin{equation*}
    \int_{-1}^{0}(-z)^{2\omega}\frac{f^{2}(z)}{(z+1)^{2i}}dz\lesssim_{\omega,i}\int_{-1}^{0}(-z)^{2\omega+2}\frac{(f^{\prime})^{2}}{(z+1)^{2(i-1)}}dz.
\end{equation*}
\end{proof}

\begin{proposition}
\label{prop: weighted L2 estimate on mpl}
For any $2\leq n\in\mathbb{N}$, we have that \begin{equation}
\label{eq: higher order estimate for mp without singular r weight}
\begin{aligned}
    &\int_{-1}^{0}(-z)^{2(n-1)-\frac{1}{2}}\left(\partial_{z}^{n}(m_{p})_{L}\right)^{2}(s,z)dz\\\leq& C(k)e^{-(1+2a)s}\left(\left\Vert r_{p}^{0}\right\Vert_{\mathcal{C}_{N}^{\frac{7}{4}}}+\left\Vert\Psi_{p}^{0}\right\Vert_{\mathcal{C}_{N}^{\alpha}}+\left\Vert(f,g,h)\right\Vert_{\mathcal{S}^{a,N,\tau}}^{2}\right)\\&+C(k)\sum_{i = 2}^{n-1}\int_{-1}^{0}(-z)^{2i-4+2\tau-\frac{1}{2}}\left(\partial_{z}^{i}r_{p}\right)^{2}(s,z)+(-z)^{2\gamma+2i-4}\left(\partial_{z}^{i}\Psi_{p}\right)^{2}(s,z)dz\\&+Ck\int_{-1}^{0}(-z)^{2\tau+2n-4}\left(\partial_{z}^{n}r_{p}\right)^{2}(s,z)+(-z)^{2\gamma+2n-4+\frac{1}{2}}\left(\partial_{z}^{n}\Psi_{p}\right)^{2}(s,z)dz.
    \end{aligned}
\end{equation}
\end{proposition}
\begin{proof}
We can rewrite the transport equation for $(m_{p})_{L}$ as \begin{equation}
\label{equation for mp under psip formualtion}
\begin{aligned}
    &\partial_{z}(m_{p})_{L}+\frac{r_{k}}{\partial_{z}r_{k}}(\partial_{z}\phi_{k})^{2}(m_{p})_{L} \\= &-\left(\frac{1}{2}\left(\frac{\mr{r}}{\partial_{z}\mr{r}}\right)^{2}(\partial_{z}\mr{\phi})^{2}+\frac{\mr{r}\mr{\phi}}{\partial_{z}\mr{r}}\partial_{z}\mr{\phi}\right)(1-\mu_{k})\partial_{z}r_{p}+\left(\frac{1}{2}\frac{\mr{r}}{\partial_{z}\mr{r}}(\partial_{z}\mr{\phi})^{2}\mu_{k}+(1-\mu_{k})\mr{\phi}\partial_{z}\mr{\phi}\right)r_{p}\\&+(1-\mu_{k})\frac{\mr{r}}{\partial_{z}\mr{r}}\partial_{z}\mr{\phi}\partial_{z}\Psi_{p}-(1-\mu_{k})\partial_{z}\mr{\phi}\Psi_{p}.
\end{aligned}
\end{equation}
For $n = 2$, commuting the equation \eqref{equation for mp under psip formualtion} with $\partial_{z}$, we have \begin{align*}
    \partial_{z}^{2}(m_{p})_{L} =& -\frac{\mr{r}}{\partial_{z}\mr{r}}(\partial_{z}\mr{\phi})^{2}\partial_{z}(m_{p})_{L}-\partial_{z}\left(\frac{\mr{r}}{\partial_{z}\mr{r}}(\partial_{z}\mr{\phi})^{2}\right)(m_{p})_{L}+(1-\mu_{k})\frac{\mr{r}}{\partial_{z}\mr{r}}\partial_{z}\mr{\phi}\partial_{z}^{2}\Psi_{p}+\partial_{z}\left[(1-\mu_{k})\frac{\mr{r}}{\partial_{z}\mr{r}}\partial_{z}\mr{\phi}\right]\partial_{z}\Psi_{p}\\&-(1-\mu_{k})\partial_{z}\mr{\phi}\partial_{z}\Psi_{p}-\partial_{z}\left[(1-\mu_{k})\partial_{z}\mr{\phi}\right]\Psi_{p}-\left(\frac{1}{2}\left(\frac{\mr{r}}{\partial_{z}\mr{r}}\right)^{2}(\partial_{z}\mr{\phi})^{2}+\frac{\mr{r}\mr{\phi}}{\partial_{z}\mr{r}}\partial_{z}\mr{\phi}\right)(1-\mu_{k})\partial_{z}^{2}r_{p}\\&-\partial_{z}\left[\left(\frac{1}{2}\left(\frac{\mr{r}}{\partial_{z}\mr{r}}\right)^{2}(\partial_{z}\mr{\phi})^{2}+\frac{\mr{r}\mr{\phi}}{\partial_{z}\mr{r}}\partial_{z}\mr{\phi}\right)(1-\mu_{k})\right]\partial_{z}r_{p}+\left(\frac{1}{2}\frac{\mr{r}}{\partial_{z}\mr{r}}(\partial_{z}\mr{\phi})^{2}\mu_{k}+(1-\mu_{k})\mr{\phi}\partial_{z}\mr{\phi}\right)\partial_{z}r_{p}\\&+\partial_{z}\left(\frac{1}{2}\frac{\mr{r}}{\partial_{z}\mr{r}}(\partial_{z}\mr{\phi})^{2}\mu_{k}+(1-\mu_{k})\mr{\phi}\partial_{z}\mr{\phi}\right)r_{p}.
\end{align*}
Therefore, using the estimates on the background $k$-self-similar spacetime in Proposition~\ref{prop: estimate on the exact k self similar spacetime}, the estimates on $(m_{p})_{L}$ \eqref{eq: zero order estimate on mpl}--\eqref{eq: first order estimate on mpl}, and Lemma~\ref{lemma: Hardy inequality}, we have \begin{align*}
    \int_{-1}^{0}(-z)^{\frac{3}{2}}\left(\partial_{z}^{2}(m_{p})_{L}\right)^{2}(s,z)dz\leq &Ck\int_{-1}^{0}(-z)^{2\tau}(\partial_{z}^{2}r_{p})^{2}(s,z)+(-z)^{2\gamma+\frac{1}{2}}(\partial_{z}^{2}\Psi_{p})^{2}(s,z)dz\\&+C(k)\int_{-1}^{0}(\partial_{z}r_{p})^{2}(s,z)+(\partial_{z}\Psi_{p})^{2}(s,z)dz.
\end{align*}
The case $n\geq 3$ will follow inductively. Hence, we conclude the proof of this proposition.
\end{proof}

Next, we establish the higher-order estimates for $\frac{m_{p}}{\mr{r}^{2}}$.
\begin{proposition}
\label{prop: higher-order estimate for mpl}
    We have the following estimate for $(m_{p})_{L}$ for $n\geq 2$: \begin{equation}
    \begin{aligned}
        \int_{-1}^{0}(-z)^{2(n-1)-\frac{1}{2}}\left\vert\partial_{z}^{n}\left(\frac{(m_{p})_{L}}{\mr{r}^{2}}\right)\right\vert^{2}\leq& C(k)e^{-(1+2a)s}\left(\left\Vert r_{p}^{0}\right\Vert_{\mathcal{C}_{N}^{\frac{7}{4}}}+\left\Vert\Psi_{p}^{0}\right\Vert_{\mathcal{C}_{N}^{\alpha}}+\left\Vert(f,g,h)\right\Vert_{\mathcal{S}^{a,N,\tau}}^{2}\right)\\&+C(k)\sum_{i = 2}^{n}\int_{-1}^{0}(-z)^{2\tau+2i-4-\frac{1}{2}}\left(\partial_{z}^{i}r_{p}\right)^{2}+(-z)^{2\gamma+2i-4}\left(\partial_{z}^{i}\Psi_{p}\right)^{2}\\&+Ck\int_{-1}^{0}(-z)^{2n+2\tau-2}\left(\partial_{z}^{n+1}r_{p}\right)^{2}\\&+Ck\int_{-1}^{0}(-z)^{2n+2\gamma-2+\frac{1}{2}}\left(\partial_{z}^{n+1}\Psi_{p}\right)^{2}.
        \end{aligned}
        \label{eq: higher-order estimate for mp in the nonlinear problem}
    \end{equation}
\end{proposition}
\begin{proof}
For $\left\Vert(-z)^{n-1-\frac{1}{4}}\partial_{z}^{n}\left(\frac{(m_{p})_{L}}{\mr{r}^{2}}\right)(s,\cdot)\right\Vert_{L_{z}^{2}}$ with $n\geq 1$, we have that \begin{equation}
    \begin{aligned}
        &\int_{-1}^{0}\left(-z\right)^{2(n-1)-\frac{1}{2}}\left(\partial_{z}^{n}\left(\frac{\left(m_{p}\right)_{L}}{\mr{r}^{2}}\right)\right)^{2}dz\\=&
        \int_{-1}^{0}(-z)^{2(n-1)-\frac{1}{2}}\left(\partial_{z}^{n}\left(\frac{\left(m_{p}\right)_{L}}{(z+1)^{2}}\frac{(z+1)^{2}}{\mr{r}^{2}}\right)\right)^{2}\\\leq&C(k)\sum_{i = 0}^{1}\int_{-1}^{0}\left(\partial_{z}^{i}\left(\frac{\left(m_{p}\right)_{L}}{(z+1)^{2}}\right)\right)^{2}+C(k)\sum_{i=2}^{n-1}\int_{-1}^{0}(-z)^{2(i-1)-\frac{1}{2}}\left(\partial_{z}^{i}\left(\frac{\left(m_{p}\right)_{L}}{(z+1)^{2}}\right)\right)^{2}\\&+C\int_{-1}^{0}(-z)^{2(n-1)-\frac{1}{2}}\left(\partial_{z}^{n}\left(\frac{(m_{p})_{L}}{(z+1)^{2}}\right)\right)^{2}.
    \end{aligned}
\end{equation}
Hence, it suffices to establish the estimates for $\left\Vert(-z)^{n-1-\frac{1}{4}}\partial_{z}^{n}\left(\frac{(m_{p})_{L}}{(z+1)^{2}}\right)(s,\cdot)\right\Vert_{L_{z}^{2}}$. By \eqref{eq: zero order estimate on mpl}, we have that \begin{equation*}
    \int_{-1}^{0}\left(\frac{\left(m_{p}\right)_{L}}{(z+1)^{2}}\right)^{2}\lesssim e^{-(1+2a)s}\left(\left\Vert r_{p}^{0}\right\Vert_{\mathcal{C}_{N}^{\frac{7}{4}}}+\left\Vert\Psi_{p}^{0}\right\Vert_{\mathcal{C}_{N}^{\alpha}}+\left\Vert(f,g,h)\right\Vert_{\mathcal{S}^{a,N,\tau}}^{2}\right)
\end{equation*} 
By \eqref{eq: first order estimate on mpl}, we have that \begin{align*}
    \int_{-1}^{0}\partial_{z}\left(\frac{(m_{p})_{L}}{(z+1)^{2}}\right)^{2}\lesssim e^{-(1+2a)s}\left(\left\Vert r_{p}^{0}\right\Vert_{\mathcal{C}_{N}^{\frac{7}{4}}}+\left\Vert\Psi_{p}^{0}\right\Vert_{\mathcal{C}_{N}^{\alpha}}+\left\Vert(f,g,h)\right\Vert_{\mathcal{S}^{a,N,\tau}}^{2}\right).
\end{align*}
Directly integrating the equation \eqref{eq: transport equation for mp in the nonlinear setting}, we have \begin{align*}
    &\frac{(m_{p})_{L}}{(z+1)^{2}} \\=& -\frac{1}{(z+1)^{2}}\left(\int_{-1}^{z}\frac{\mr{r}}{\partial_{z}\mr{r}}(\partial_{z}\mr{\phi})^{2}(m_{p})_{L}-\int_{-1}^{z}\left(1-\frac{1}{2}\mu_{k}\right)\frac{(\partial_{z}\mr{\phi})^{2}}{\partial_{z}\mr{r}}\mr{r}^{2}\frac{r_{p}}{\mr{r}}-\int_{-1}^{z}\frac{1}{2}\left(\mu_{k}-1\right)\frac{(\partial_{z}\mr{\phi})^{2}}{(\partial_{z}\mr{r})^{2}}\mr{r}^{2}\partial_{z}r_{p}\right)\\&+\frac{1}{(z+1)^{2}}\int_{-1}^{z}\frac{1}{2}(1-\mu_{k})\frac{\partial_{z}\mr{\phi}}{\partial_{z}\mr{r}}\mr{r}^{2}\left(\partial_{z}\left(\frac{\Psi_{p}}{\mr{r}}\right)-\partial_{z}\left(\mr{\phi}\frac{r_{p}}{\mr{r}}\right)\right)\\=&:\frac{1}{(z+1)^{2}}\left(\mathfrak{m}_{1}+\mathfrak{m}_{2}\right)+\frac{1}{2}\left(\mu_{k}-1\right)\frac{(\partial_{z}\mr{\phi})^{2}}{(\partial_{z}\mr{r})^{2}}\frac{\mr{r}^{2}}{(z+1)^{2}}r_{p}\bigg|_{-1}^{z}-\frac{1}{2(z+1)^{2}}\int_{-1}^{z}\frac{d}{dz}\left((\mu_{k}-1)\frac{(\partial_{z}\mr{\phi})^{2}}{(\partial_{z}\mr{r})^{2}}\mr{r}^{2}\right)r_{p}\\&+\frac{1}{(z+1)^{2}}\int_{-1}^{z}\frac{1}{2}(1-\mu_{k})\frac{\partial_{z}\mr{\phi}}{\partial_{z}\mr{r}}\mr{r}^{2}\left(\partial_{z}\left(\frac{\Psi_{p}}{\mr{r}}\right)-\partial_{z}\left(\mr{\phi}\frac{r_{p}}{\mr{r}}\right)\right)\\=&:\frac{1}{(z+1)^{2}}\left(\mathfrak{m}_{1}+\mathfrak{m}_{2}\right)+\mathfrak{m}_{3}+\frac{1}{(z+1)^{2}}\mathfrak{m}_{4}+\frac{1}{(z+1)^{2}}\mathfrak{m}_{5}.
\end{align*}

Using the Taylor expansion, we can express $\mathfrak{m}_{i}$ in the following form for $n\geq 1$:
\begin{equation}
\mathfrak{m}_{i}(s,z) = \sum_{i = 0}^{n+1}\frac{\partial_{z}^{i}\mathfrak{m}_{i}(s,-1)}{i!}(z+1)^{i}+\frac{1}{(n+1)!}\int_{-1}^{z}\partial_{z}^{n+2}\mathfrak{m}_{i}(s,z^{\prime})(z-z^{\prime})^{n+1}dz^{\prime}.
\end{equation} 
Hence, using Lemma \ref{lemma: inductive hardy} repeatedly, for $i = 1,2,4$, we have \begin{equation}
    \begin{aligned}
        &\int_{-1}^{0}(-z)^{2(n-1)-\frac{1}{2}}\left\vert\partial_{z}^{n}\left(\frac{\mathfrak{m}_{i}(s,z)}{(z+1)^{2}}\right)\right\vert^{2} dz\\\lesssim&\int_{-1}^{0}(-z)^{2(n-1)-\frac{1}{2}}\left\vert\partial_{z}^{n}\left(\frac{1}{(z+1)^{2}}\int_{-1}^{z}\partial_{z}^{n+2}\mathfrak{m}_{i}(s,z^{\prime})(z-z^{\prime})^{n+1}dz^{\prime}\right)\right\vert^{2}dz\\\lesssim&\sum_{j = 0}^{n}\int_{-1}^{0}(-z)^{2(n-1)-\frac{1}{2}}(z+1)^{-2(2+j)}\left(\int_{-1}^{z}\partial_{z}^{n+2}\mathfrak{m}_{i}(s,z^{\prime})(z-z^{\prime})^{j+1}dz^{\prime}\right)^{2}dz\\\lesssim&\sum_{j = 0}^{n}\int_{-1}^{0}(-z)^{2(n-1)-\frac{1}{2}+2(j+2)}\left\vert\partial_{z}^{n+2}\mathfrak{m}_{i}\right\vert^{2}(s,z)dz\\\lesssim&\int_{-1}^{0}(-z)^{2(n+1)-\frac{1}{2}}\left\vert\partial_{z}^{n+2}\mathfrak{m}_{i}\right\vert^{2}(s,z)dz.
    \end{aligned}
    \label{eq: Taylor expansion estimate}
\end{equation}
Then for $\mathfrak{m}_{1}$, we have that \begin{equation}
\begin{aligned}
    &\int_{-1}^{0}(-z)^{2(n+1)-\frac{1}{2}}\left\vert\partial_{z}^{n+2}\mathfrak{m}_{1}\right\vert^{2}(s,z)dz\\\leq& C(k)\int_{-1}^{0}(-z)^{2(n+1)-\frac{1}{2}}\left\vert\partial_{z}^{n+1}\left(\frac{\mr{r}}{\partial_{z}\mr{r}}(\partial_{z}\mr{\phi})^{2}(m_{p})_{L}\right)\right\vert^{2}\\\leq&C(k)e^{-(1+2a)s}\left(\left\Vert r_{p}^{0}\right\Vert_{\mathcal{C}_{N}^{\frac{7}{4}}}+\left\Vert\Psi_{p}^{0}\right\Vert_{\mathcal{C}_{N}^{\alpha}}+\left\Vert(f,g,h)\right\Vert_{\mathcal{S}^{a,N,\tau}}^{2}\right)\\&+C(k)\sum_{i  =2}^{n}\int_{-1}^{0}(-z)^{2\tau+2i-4-\frac{1}{2}}(\partial_{z}^{i}r_{p})^{2}(s,z)+(-z)^{2\gamma+2i-4}(\partial_{z}^{i}\Psi_{p})^{2}(s,z)dz\\&+Ck\int_{-1}^{0}(-z)^{2n+2\tau-2}\left(\partial_{z}^{n+1}r_{p}\right)^{2}(s,z)+(-z)^{2n+2\gamma-2+\frac{1}{2}}\left(\partial_{z}^{n+1}\Psi_{p}\right)^{2}(s,z)dz.
    \end{aligned}
\end{equation}
For $\mathfrak{m}_{2}$, we have that \begin{equation}
    \begin{aligned}
        \int_{-1}^{0}(-z)^{2(n+1)-\frac{1}{2}}\left\vert\partial_{z}^{n+2}\mathfrak{m}_{2}\right\vert^{2}(s,z)dz\leq&C(k)\int_{-1}^{0}(-z)^{2(n+1)-\frac{1}{2}}\left\vert\partial_{z}^{n+1}\left(\left(1-\frac{1}{2}\mu_{k}\right)\frac{(\partial_{z}\mr{\phi})^{2}}{\partial_{z}\mr{r}}\mr{r}^{2}\frac{r_{p}}{\mr{r}}\right)\right\vert^{2}\\\leq&C(k)e^{-(1+2a)s}\left(\left\Vert r_{p}^{0}\right\Vert_{\mathcal{C}_{N}^{\frac{7}{4}}}+\left\Vert\Psi_{p}^{0}\right\Vert_{\mathcal{C}_{N}^{\alpha}}+\left\Vert(f,g,h)\right\Vert_{\mathcal{S}^{a,N,\tau}}^{2}\right)\\&+C(k)\sum_{i = 2}^{n}\int_{-1}^{0}(-z)^{2\tau+2i-4-\frac{1}{2}}(\partial_{z}^{i}r_{p})^{2}dz+Ck\int_{-1}^{0}(-z)^{2n+2\tau-2}\left(\partial_{z}^{n+1}r_{p}\right)^{2}dz.
    \end{aligned}
\end{equation}
Similarly, for $\mathfrak{m}_{4}$, we have that 
\begin{align}
        \int_{-1}^{0}(-z)^{2(n+1)-\frac{1}{2}}\left\vert\partial_{z}^{n+2}\mathfrak{m}_{4}\right\vert^{2}dz\leq&C(k)\int_{-1}^{0}(\partial_{z}r_{p})^{2}dz+C(k)\sum_{i = 2}^{n}\int_{-1}^{0}(-z)^{2\tau+2i-4-\frac{1}{2}}(\partial_{z}^{i}r_{p})^{2}dz\nonumber\\&+Ck\int_{-1}^{0}(-z)^{2n+2\tau-2}\left(\partial_{z}^{n+1}r_{p}\right)^{2}dz.
\end{align}
For $\mathfrak{m}_{3}$, we have that \begin{equation}
    \begin{aligned}
        \int_{-1}^{0}(-z)^{2(n-1)-\frac{1}{2}}\left\vert\partial_{z}^{n}\mathfrak{m}_{3}\right\vert^{2}(s,z)dz\leq&C(k)\int_{-1}^{0}(\partial_{z}r_{p})^{2}dz+C(k)\sum_{i = 2}^{n}\int_{-1}^{0}(-z)^{2\tau+2i-4-\frac{1}{2}}(\partial_{z}^{i}r_{p})^{2}dz.
    \end{aligned}
\end{equation}
For $\mathfrak{m}_{5}$, we first split the integral into two parts \begin{align*}
    \int_{-1}^{0}(-z)^{2(n-1)-\frac{1}{2}}\left\vert\partial_{z}^{n}\left(\frac{\mathfrak{m}_{5}}{(z+1)^{2}}\right)\right\vert^{2}dz = &\int_{-1}^{-\frac{1}{2}}(-z)^{2(n-1)-\frac{1}{2}}\left\vert\partial_{z}^{n}\left(\frac{\mathfrak{m}_{5}}{(z+1)^{2}}\right)\right\vert^{2}dz\\&+\int_{-\frac{1}{2}}^{0}(-z)^{2(n-1)-\frac{1}{2}}\left\vert\partial_{z}^{n}\left(\frac{\mathfrak{m}_{5}}{(z+1)^{2}}\right)\right\vert^{2}dz.
\end{align*}
For the integration over $[-1,-\frac{1}{2}]$, using the integration by parts, we have that \begin{equation}
    \begin{aligned}
        &\int_{-1}^{-\frac{1}{2}}(-z)^{2(n-1)-\frac{1}{2}}\left\vert\partial_{z}^{n}\left(\frac{\mathfrak{m}_{5}}{(z+1)^{2}}\right)\right\vert^{2}\\\leq& C\int_{-1}^{-\frac{1}{2}}\left\vert\partial_{z}^{n}\left(\frac{\mathfrak{m}_{5}}{(z+1)^{2}}\right)\right\vert^{2}\\\leq&C\int_{-1}^{-\frac{1}{2}}\left\vert\partial_{z}^{n}\left(\frac{1}{2}(1-\mu_{k})\frac{\mr{r}^{2}}{(z+1)^{2}}\frac{\partial_{z}\mr{\phi}}{\partial_{z}\mr{r}}\left(\frac{\Psi_{p}}{\mr{r}}-\mr{\phi}\frac{r_{p}}{\mr{r}}\right)\right)\right\vert^{2}dz\\&+C\int_{-1}^{-\frac{1}{2}}\left\vert\partial_{z}^{n}\left(\frac{1}{(z+1)^{2}}\int_{-1}^{z}\partial_{z}\left(\frac{1}{2}(1-\mu_{k})\frac{\partial_{z}\mr{\phi}}{\partial_{z}\mr{r}}\mr{r}^{2}\right)\left(\frac{\Psi_{p}}{\mr{r}}-\mr{\phi}\frac{r_{p}}{\mr{r}}\right)d\widetilde{z}\right)\right\vert^{2}dz\\\leq&C(k)\sum_{i = 0}^{n-1}\int_{-1}^{-\frac{1}{2}}\left(\partial_{z}^{i}\left(\frac{\Psi_{p}}{(z+1)}\right)\right)^{2}+\left(\partial_{z}^{i}\left(\frac{r_{p}}{(z+1)}\right)\right)^{2}\\&+Ck\int_{-1}^{-\frac{1}{2}}\left(\partial_{z}^{n}\left(\frac{\Psi_{p}}{(z+1)}\right)\right)^{2}+\left(\partial_{z}^{n}\left(\frac{r_{p}}{(z+1)}\right)\right)^{2}+C\int_{-1}^{-\frac{1}{2}}\left\vert\partial_{z}^{n}\left(\frac{1}{(z+1)^{2}}\widetilde{\mathfrak{m}}_{5}\right)\right\vert^{2}dz,
    \end{aligned}
\end{equation}
where $\widetilde{\mathfrak{m}}_{5}$ is defined to be \begin{equation*}
    \widetilde{\mathfrak{m}}_{5}:=\int_{-1}^{z}\partial_{z}\left(\frac{1}{2}(1-\mu_{k})\frac{\partial_{z}\mr{\phi}}{\partial_{z}\mr{r}}\mr{r}^{2}\right)\left(\frac{\Psi_{p}}{\mr{r}}-\mr{\phi}\frac{r_{p}}{\mr{r}}\right)d\widetilde{z}.
\end{equation*}
Using the Taylor expansion as in \eqref{eq: Taylor expansion estimate}, we have that \begin{equation}
    \begin{aligned}
        \int_{-1}^{-\frac{1}{2}}(-z)^{2(n-1)-\frac{1}{2}}\left\vert\partial_{z}^{n}\left(\frac{\mathfrak{m}_{5}}{(z+1)^{2}}\right)\right\vert^{2}\leq& C(k)\int_{-1}^{0}\left(\partial_{z}r_{p}\right)^{2}+\left(\partial_{z}\Psi_{p}\right)^{2}dz\\&+C(k)\sum_{i = 2}^{n}\int_{-1}^{0}(-z)^{2i-4+2\tau-\frac{1}{2}}\left(\partial_{z}^{i}r_{p}\right)^{2}+(-z)^{2i-4+2\gamma}\left(\partial_{z}^{i}\Psi_{p}\right)^{2}\\&+Ck\int_{-1}^{0}(-z)^{2n+2\tau-2}\left(\partial_{z}^{n+1}r_{p}\right)^{2}\\&+Ck\int_{-1}^{0}(-z)^{2n+2\gamma-2+\frac{1}{2}}\left(\partial_{z}^{n+1}\Psi_{p}\right)^{2}.
    \end{aligned}
\end{equation}
For the integration over $[-\frac{1}{2},0]$, we have that \begin{equation}
    \begin{aligned}
        \int_{-\frac{1}{2}}^{0}(-z)^{2(n-1)-\frac{1}{2}}\left\vert\partial_{z}^{n}\left(\frac{\mathfrak{m}_{5}}{(z+1)^{2}}\right)\right\vert^{2}\leq &C(k)\int_{-1}^{0}\left(\partial_{z}r_{p}\right)^{2}+\left(\partial_{z}\Psi_{p}\right)^{2}\\&+C(k)\sum_{i = 2}^{n}\int_{-1}^{0}(-z)^{2\tau+2i-4-\frac{1}{2}}\left(\partial_{z}^{i}r_{p}\right)^{2}+(-z)^{2\gamma+2i-4}\left(\partial_{z}^{i}\Psi_{p}\right)^{2}.
    \end{aligned}
\end{equation}
Putting everything together concludes the proof of this proposition.
\end{proof}
Now, we can prove the following higher-order energy estimate for $\mathfrak{A}$.
\begin{proposition}
\label{prop: higher order energy estimate on nonlinear problem}
    Given functions $(r_{p}^{0},\Psi_{p}^{0})\in\mathcal{C}_{N+1}^{\frac{7}{4}}\times \mathcal{C}_{N+1}^{\alpha}$ and $(f,\partial_{s}f,g)\in\mathcal{S}_{\delta}^{a,N,\tau}$ with $\tau$ sufficiently small and $a\in\left(0,\frac{\alpha^{\prime}q_{k}-1}{2}\right)$, then the solution $\mathfrak{A} = (r_{p},\Psi_{p},(m_{p})_{L})$ to the equation $\mathcal{P}(r_{p},\Psi_{p},(m_{p})_{L}) = \widetilde{\mathcal{N}}(f,g,h)$ with the initial data $(r_{p},\Psi_{p})|_{\Sigma_{-1}^{(in)}} = (r_{p}^{0},\Psi_{p}^{0}-\widetilde{c}_{\infty}r_{k}-\widetilde{c}_{\infty}r_{p}^{0})$ satisfies the following energy estimate for $2\leq n\leq N+1$:
    \begin{equation}
        \begin{aligned}
            &e^{2(1+a_{\frac{n-1}{100}})s_{0}}\left(\int_{s = s_{0}}(-z)^{2n-4+2\tau-\frac{1}{2}}\left(\partial_{z}^{n}r_{p}\right)^{2}+(-z)^{2\tau-\frac{1}{2}+\max\{2n-6,0\}}\left(\partial_{s}\partial_{z}^{n-1}r_{p}\right)^{2}+(-z)^{2\gamma+2n-4}\left(\partial_{z}^{n}\Psi_{p}\right)^{2}\right)\\&\hspace{7cm}\lesssim\delta^{-\frac{n}{100}}\left(\left\Vert r_{p}^{0}\right\Vert_{\mathcal{C}_{N+1}^{\frac{7}{4}}}^{2}+\left\Vert\Psi_{p}^{0}\right\Vert_{\mathcal{C}_{N+1}^{\alpha}}^{2}+\left\Vert(f,g,h)\right\Vert_{\mathcal{S}^{a,N,\tau}}^{4}\right)
        \end{aligned}
        \label{eq: higher-order energy estimate}
    \end{equation}
\end{proposition}
\begin{proof}
We prove the estimate \eqref{eq: higher-order energy estimate} inductively. The case $n = 2$ has been proved in Proposition \ref{prop: second order energy estimate in the nonlinear problem}. Assuming the estimate \eqref{eq: higher-order energy estimate} holds for $2\leq n\leq i$, then it suffices to establish \eqref{eq: higher-order energy estimate} for $n = i+1$. 

We first commute the wave equation for $\Psi_{\lambda}$ with $\partial_{z}^{i}$ and multiply by $$(-z)^{2(i+1)+2\gamma-4}w(z)\partial_{z}^{i+1}\Psi_{\lambda},$$ where $w(z) = -2(-z)^{\frac{1}{2}}+3$. We have that \begin{equation}
\label{eq: higher-order energy estimate in the induction process}
    \begin{aligned}
        &\frac{1}{2}\int_{s_{0}}(-z)^{2(i+1)+2\gamma-4}w(z)\left(\partial_{z}^{i+1}\Psi_{\lambda}\right)^{2}+\frac{1}{2}q_{k}\int_{\Gamma}(\partial_{z}^{i+1}\Psi_{\lambda})^{2}\\&+\iint_{\mathcal{R}(s_{0})}\left[\left(q_{k}\left(\frac{3}{2}-\gamma\right)-(1+a_{\lambda})\right)w(z)+\frac{1}{2}q_{k}(-z)^{\frac{1}{2}}\right](-z)^{2(i+1)+2\gamma-4}\left(\partial_{z}^{i+1}\Psi_{\lambda}\right)^{2}\\\leq&\iint_{\mathcal{R}(s_{0})}G_{k}(-z)^{2(i+1)+2\gamma-4}\left\vert\partial_{z}^{i+1}r_{\lambda}\right\vert\left\vert\partial_{z}^{i+1}\Psi_{\lambda}\right\vert+G_{k^{2}}(-z)^{2(i+1)+2\gamma-3}\left\vert\partial_{z}^{i+1}r_{{\lambda}}\right\vert\left\vert\partial_{z}^{i+1}\Psi_{{\lambda}}\right\vert\\&+\iint_{\mathcal{R}(s_{0})}G_{k}(-z)^{2(i+1)+2\gamma-4}e^{(1+a_{\lambda})s}\left\vert\partial_{z}^{i}\left(\frac{(m_{p})_{L}}{\mr{r}^{2}}\right)\right\vert\left\vert\partial_{z}^{i+1}\Psi_{{\lambda}}\right\vert\\&+\iint_{\mathcal{R}(s_{0})}(-z)^{2(i+1)+2\gamma-4}e^{(1+a_{\lambda})s}\left\vert\partial_{z}^{i}\widetilde{\mathcal{N}}^{(2)}\right\vert\left\vert\partial_{z}^{i+1}\Psi_{{\lambda}}\right\vert\\&+C(k)\iint_{\mathcal{R}(s_{0})}(-z)^{2(i+1)+2\gamma-4}\left\vert\text{l.o.t}\right\vert\left\vert\partial_{z}^{i+1}\Psi_{{\lambda}}\right\vert,
    \end{aligned}
\end{equation}
where  by our induction hypothesis, the regularity of the background spacetime, Proposition \ref{prop: higher-order estimate for mpl}, and the Cauchy--Schwarz inequality, for $\lambda = \frac{i}{100}$, we have the following estimate for the lower-order terms:\begin{equation}
    \begin{aligned}
        &\iint_{\mathcal{R}(s_{0})}(-z)^{2(i+1)+2\gamma-4}\left\vert\text{l.o.t}\right\vert\left\vert\partial_{z}^{i+1}\Psi_{{\lambda}}\right\vert\\\lesssim&\delta^{\frac{1}{100}}\iint_{\mathcal{R}(s_{0})}(-z)^{2(i+1)+2\gamma-4}\left\vert\partial_{z}^{i+1}\Psi_{{\lambda}}\right\vert^{2}+\delta^{-\frac{i+1}{100}}\left(\left\Vert r_{p}^{0}\right\Vert_{\mathcal{C}_{N+1}^{\frac{7}{4}}}^{2}+\left\Vert\Psi_{p}^{0}\right\Vert_{\mathcal{C}_{N+1}^{\alpha}}^{2}+\left\Vert(f,g,h)\right\Vert_{\mathcal{S}^{a,N,\tau}}^{4}\right).
    \end{aligned}
\end{equation}
For the first and second terms on the right-hand side of \eqref{eq: higher-order energy estimate in the induction process}, we have that \begin{equation}
    \begin{aligned}
        \iint_{\mathcal{R}(s_{0})}G_{k}(-z)^{2(i+1)+2\gamma-4}\left\vert\partial_{z}^{i+1}r_{{\lambda}}\right\vert\left\vert\partial_{z}^{i+1}\Psi_{{\lambda}}\right\vert\lesssim& k\iint_{\mathcal{R}(s_{0})}(-z)^{2(i+1)+2\gamma-4+\frac{1}{2}}\left\vert\partial_{z}^{i+1}\Psi_{\lambda}\right\vert^{2}\\&+k\iint_{\mathcal{R}(s_{0})}(-z)^{2(i+1)+2\gamma-4-\frac{1}{2}}\left\vert\partial_{z}^{i+1}r_{{\lambda}}\right\vert^{2}\\\lesssim&k\iint_{\mathcal{R}(s_{0})}(-z)^{2(i+1)+2\gamma-4+\frac{1}{2}}\left\vert\partial_{z}^{i+1}\Psi_{{\lambda}}\right\vert^{2}\\&+k\iint_{\mathcal{R}(s_{0})}(-z)^{2(i+1)+2\tau-4-\frac{1}{2}}\left(\partial_{z}^{i+1}r_{{\lambda}}\right)^{2}.
    \end{aligned}
\end{equation}
For the third term on the right-hand side of \eqref{eq: higher-order energy estimate in the induction process}, we have that \begin{equation}
    \begin{aligned}
        &\iint_{\mathcal{R}(s_{0})}G_{k}(-z)^{2i+2\gamma-2}e^{(1+a_{\lambda})s}\left\vert\partial_{z}^{i}\left(\frac{(m_{p})_{L}}{\mr{r}^{2}}\right)\right\vert\left\vert\partial_{z}^{i+1}\Psi_{{\lambda}}\right\vert\\\leq&Ck\iint_{\mathcal{R}(s_{0})}(-z)^{2i+2\gamma-2+\frac{1}{2}}\left\vert\partial_{z}^{i+1}\Psi_{{\lambda}}\right\vert^{2}+Ck\iint_{\mathcal{R}(s_{0})}(-z)^{2i+2\gamma-2-\frac{1}{2}}e^{2(1+a_{\lambda})s}\left\vert\partial_{z}^{i}\left(\frac{(m_{p})_{L}}{\mr{r}^{2}}\right)\right\vert^{2}\\\leq&
        Ck\iint_{\mathcal{R}(s_{0})}(-z)^{2i+2\gamma-2+\frac{1}{2}}\left\vert\partial_{z}^{i+1}\Psi_{{\lambda}}\right\vert^{2}+Ck\iint_{\mathcal{R}(s_{0})}(-z)^{2(i-1)-\frac{1}{2}}e^{2(1+a_{\lambda})s}\left\vert\partial_{z}^{i}\left(\frac{(m_{p})_{L}}{\mr{r}^{2}}\right)\right\vert^{2}\\\leq&C(k)\left(\left\Vert r_{p}^{0}\right\Vert_{\mathcal{C}_{N}^{\frac{7}{4}}}^{2}+\left\Vert\Psi_{p}^{0}\right\Vert_{\mathcal{C}_{N}^{\alpha}}^{2}+\left\Vert(f,g,h)\right\Vert_{\mathcal{S}^{a,N,\tau}}^{4}\right)\\&+C(k)\sum_{j = 2}^{i}\iint_{\mathcal{R}(s_{0})}(-z)^{2j+2\tau-4-\frac{1}{2}}\left(\partial_{z}^{j}r_{{\lambda}}\right)^{2}+(-z)^{2j+2\gamma-4}\left(\partial_{z}^{j}\Psi_{{\lambda}}\right)^{2}\\&+Ck\iint_{\mathcal{R}(s_{0})}(-z)^{2i+2\tau-2}\left(\partial_{z}^{i+1}r_{{\lambda}}\right)^{2}+(-z)^{2i+2\gamma-2+\frac{1}{2}}\left(\partial_{z}^{i+1}\Psi_{{\lambda}}\right)^{2}.
    \end{aligned}
\end{equation}
For the last term on the right-hand side of \eqref{eq: higher-order energy estimate in the induction process}, we have that \begin{equation}
    \begin{aligned}
        &\iint_{\mathcal{R}(s_{0})}(-z)^{2(i+1)+2\gamma-4}e^{(1+a_{\lambda})s}\left\vert\partial_{z}^{i}\widetilde{\mathcal{N}}^{(2)}\right\vert\left\vert\partial_{z}^{i+1}\Psi_{{\lambda}}\right\vert^{2}\\\lesssim&\delta^{\frac{1}{100}}\iint_{\mathcal{R}(s_{0})}(-z)^{2(i+1)+2\gamma-4}\left\vert\partial_{z}^{i+1}\Psi_{{\lambda}}\right\vert^{2}+\delta^{-\frac{1}{100}}\iint_{\mathcal{R}(s_{0})}(-z)^{2(i+1)+2\gamma-4}e^{2(1+a_{\lambda})s}\left\vert\partial_{z}\widetilde{\mathcal{N}}^{(2)}\right\vert^{2}\\\lesssim&\delta^{\frac{1}{100}}\iint_{\mathcal{R}(s_{0})}(-z)^{2(i+1)+2\gamma-4}\left\vert\partial_{z}^{i+1}\Psi_{{\lambda}}\right\vert^{2}+\delta^{-\frac{1}{100}}\left(\left\Vert r_{p}^{0}\right\Vert_{\mathcal{C}_{N+1}^{\frac{7}{4}}}^{2}+\left\Vert\Psi_{p}^{0}\right\Vert_{\mathcal{C}_{N+1}^{\alpha}}^{2}+\left\Vert(f,g,h)\right\Vert_{\mathcal{S}^{a,N,\tau}}^{4}\right).
    \end{aligned}
\end{equation}
Putting everything together and using the induction assumption, since $k$ is sufficiently small, we can choose $\delta$ to be small enough, such that for $\lambda = \frac{i}{100}$ we have that \begin{equation}
    \begin{aligned}
        &\int_{s = s_{0}}(-z)^{2(i+1)+2\gamma-4}\left(\partial_{z}^{i+1}\Psi_{{\frac{i}{100}}}\right)^{2}\\&+\iint_{\mathcal{R}(s_{0})}\left[\left(q_{k}\left(\frac{3}{2}-\gamma\right)-(1+a_{\frac{i}{100}})\right)w(z)+\frac{1}{4}q_{k}(-z)^{\frac{1}{2}}\right](-z)^{2(i+1)+2\gamma-4}\left(\partial_{z}^{i+1}\Psi_{\frac{i}{100}}\right)^{2}\\\lesssim& \delta^{-\frac{i+1}{100}}\left(\left\Vert r_{p}^{0}\right\Vert_{\mathcal{C}_{N+1}^{\frac{7}{4}}}^{2}+\left\Vert\Psi_{p}^{0}\right\Vert_{\mathcal{C}_{N+1}^{\alpha}}^{2}+\left\Vert(f,g,h)\right\Vert_{\mathcal{S}^{a,N,\tau}}^{4}\right)+k\iint_{\mathcal{R}(s_{0})}(-z)^{2i+2\tau-2}\left\vert\partial_{z}^{i+1}r_{{\frac{i}{100}}}\right\vert^{2}.
    \end{aligned}
    \label{eq: higher-order energy estimate for psip}
\end{equation}

Commuting the wave equation for $r_{{\lambda}}$ with $\partial_{z}^{i}$ and multiplying by $(-z)^{2(i+1)+2\tau-4-\frac{1}{2}}\partial_{z}^{i+1}r_{{\lambda}}$, we have that \begin{equation}
    \begin{aligned}
    &\frac{1}{2}\int_{s = s_{0}}(-z)^{2i+2\tau-2-\frac{1}{2}}\left(\partial_{z}^{i+1}r_{{\lambda}}\right)^{2}+\frac{1}{2}q_{k}\int_{\Gamma}\left(\partial_{z}^{i+1}r_{{\lambda}}\right)^{2}\\&+\left(q_{k}\left(\frac{7}{4}-\tau\right)-(1+a_{\lambda})-Ck^{2}\right)\iint_{\mathcal{R}(s_{0})}(-z)^{2i+2\tau-2-\frac{1}{2}}\left(\partial_{z}^{i+1}r_{{\lambda}}\right)^{2}\\\leq&C\iint_{\mathcal{R}(s_{0})}(-z)^{2i+2\tau-2-\frac{1}{2}}e^{(1+a_{\lambda})s}\left\vert\partial_{z}^{i}\left(\frac{(m_{p})_{L}}{\mr{r}^{2}}\right)\right\vert\left\vert\partial_{z}^{i+1}r_{{\lambda}}\right\vert+(-z)^{2i+2\tau-2-\frac{1}{2}}e^{(1+a_{\lambda})s}\left\vert\partial_{z}^{i}\widetilde{\mathcal{N}}^{(1)}\right\vert\left\vert\partial_{z}^{i+1}r_{{\lambda}}\right\vert\\&+\iint_{\mathcal{R}(s_{0})}(-z)^{2i+2\tau-2-\frac{1}{2}}\left\vert\text{l.o.t}\right\vert\left\vert\partial_{z}^{i+1}r_{{\lambda}}\right\vert+\frac{1}{2}\int_{s = 0}(-z)^{2i+2\tau-2-\frac{1}{2}}(\partial_{z}^{i+1}r_{\lambda})^{2}
    \end{aligned}
    \label{eq: higher-order energy estimate for rp in the induction process}
\end{equation}
where, by our induction hypothesis, the regularity of the background spacetime, and Proposition \ref{prop: higher-order estimate for mpl}, for $\lambda = \frac{i}{100}$ we have the following estimate for the lower-order terms above:\begin{equation}
    \begin{aligned}
        \iint_{\mathcal{R}(s_{0})}(-z)^{2i+2\tau-2-\frac{1}{2}}\left\vert\text{l.o.t}\right\vert\left\vert\partial_{z}^{i+1}r_{{\frac{i}{100}}}\right\vert\lesssim& \delta^{-\frac{i+1}{100}}\left(\left\Vert r_{p}^{0}\right\Vert_{\mathcal{C}_{N+1}^{\frac{7}{4}}}^{2}+\left\Vert\Psi_{p}^{0}\right\Vert_{\mathcal{C}_{N+1}^{\alpha}}^{2}+\left\Vert(f,g,h)\right\Vert_{\mathcal{S}^{a,N,\tau}}^{4}\right)\\&+\delta^{\frac{1}{100}}\iint_{\mathcal{R}(s_{0})}(-z)^{2i+2\tau-2-\frac{1}{2}}\left\vert\partial_{z}^{i+1}r_{{\frac{i}{100}}}\right\vert^{2}.
    \end{aligned}
\end{equation}
For the first term on the right-hand side of \eqref{eq: higher-order energy estimate for rp in the induction process}, using Proposition \ref{prop: higher-order estimate for mpl}, we have that \begin{equation}
    \begin{aligned}
        &\iint_{\mathcal{R}(s_{0})}(-z)^{2i+2\tau-2-\frac{1}{2}}e^{(1+a_{\lambda})s}\left\vert\partial_{z}^{i}\left(\frac{(m_{p})_{L}}{\mr{r}^{2}}\right)\right\vert\left\vert\partial_{z}^{i+1}r_{{\lambda}}\right\vert\\\leq&Ck^{\frac{1}{2}}\iint_{\mathcal{R}(s_{0})}(-z)^{2i+2\tau-2-\frac{1}{2}}\left\vert\partial_{z}^{i+1}r_{{\lambda}}\right\vert^{2}+Ck^{-\frac{1}{2}}\iint_{\mathcal{R}(s_{0})}(-z)^{2i+2\tau-2-\frac{1}{2}}e^{2(1+a_{\lambda})s}\left\vert\partial_{z}^{i}\left(\frac{(m_{p})_{L}}{\mr{r}^{2}}\right)\right\vert^{2}\\\leq&Ck^{\frac{1}{2}}\iint_{\mathcal{R}(s_{0})}(-z)^{2i+2\tau-2-\frac{1}{2}}\left\vert\partial_{z}^{i+1}r_{{\lambda}}\right\vert^{2}+(-z)^{2i+2\gamma-2+\frac{1}{2}}\left(\partial_{z}^{i+1}\Psi_{{\lambda}}\right)^{2}\\&+C(k)\sum_{j = 2}^{i}\iint_{\mathcal{R}(s_{0})}(-z)^{2j+2\tau-4-\frac{1}{2}}\left\vert\partial_{z}^{j}r_{{\lambda}}\right\vert^{2}+(-z)^{2j+2\gamma-4}\left(\partial_{z}^{j}\Psi_{{\lambda}}\right)^{2}\\&+C(k)\left(\left\Vert r_{p}^{0}\right\Vert_{\mathcal{C}_{N}^{\frac{7}{4}}}+\left\Vert\Psi_{p}^{0}\right\Vert_{\mathcal{C}_{N}^{\alpha}}+\left\Vert(f,g,h)\right\Vert_{\mathcal{S}^{a,N,\tau}}^{2}\right).
    \end{aligned}
\end{equation}
For the second term on the right-hand side of \eqref{eq: higher-order energy estimate for rp in the induction process}, we have that \begin{equation}
    \begin{aligned}
        \iint_{\mathcal{R}(s_{0})}(-z)^{2i+2\tau-2-\frac{1}{2}}e^{(1+a_{\lambda})s}\left\vert\partial_{z}^{i}\widetilde{\mathcal{N}}^{(1)}\right\vert\left\vert\partial_{z}^{i+1}r_{{\lambda}}\right\vert\lesssim&\delta^{\frac{1}{100}}\iint_{\mathcal{R}(s_{0})}(-z)^{2i+2\tau-2-\frac{1}{2}}\left\vert\partial_{z}^{i+1}r_{{\lambda}}\right\vert^{2}\\&+\delta^{-\frac{1}{100}}\iint_{\mathcal{R}(s_{0})}(-z)^{2i+2\tau-2-\frac{1}{2}}e^{2(1+a_{\lambda})s}\left\vert\partial_{z}^{i}\widetilde{\mathcal{N}}^{(1)}\right\vert^{2}\\\lesssim&\delta^{\frac{1}{100}}\iint_{\mathcal{R}(s_{0})}(-z)^{2i+2\tau-2-\frac{1}{2}}\left\vert\partial_{z}^{i+1}r_{{\lambda}}\right\vert^{2}\\&+\delta^{-\frac{1}{100}}\left(\left\Vert r_{p}^{0}\right\Vert_{\mathcal{C}_{N+1}^{\frac{7}{4}}}^{2}+\left\Vert\Psi_{p}^{0}\right\Vert_{\mathcal{C}_{N+1}^{\alpha}}^{2}+\left\Vert(f,g,h)\right\Vert_{\mathcal{S}^{a,N,\tau}}^{4}\right).
    \end{aligned}
\end{equation}
Putting everything together, since $k$ is sufficiently small, we can choose $\delta$ to be small enough, such that for $\lambda = \frac{i}{100}$, we have that \begin{equation}
    \begin{aligned}
        &\int_{s = s_{0}}(-z)^{2i+2\tau-2-\frac{1}{2}}\left(\partial_{z}^{i+1}r_{{\frac{i}{100}}}\right)^{2}\\&+\left(q_{k}\left(\frac{7}{4}-\tau\right)-(1+a_{\frac{i}{100}})-Ck^{\frac{1}{2}}-C\delta^{\frac{1}{100}}\right)\iint_{\mathcal{R}(s_{0})}(-z)^{2i+2\tau-2-\frac{1}{2}}\left(\partial_{z}^{i+1}r_{{\frac{i}{100}}}\right)^{2}\\\leq&\delta^{-\frac{i+1}{100}}\left(\left\Vert r_{p}^{0}\right\Vert_{\mathcal{C}_{N+1}^{\frac{7}{4}}}^{2}+\left\Vert\Psi_{p}^{0}\right\Vert_{\mathcal{C}_{N+1}^{\alpha}}^{2}+\left\Vert(f,g,h)\right\Vert_{\mathcal{S}^{a,N,\tau}}^{4}\right)+Ck^{\frac{1}{2}}\iint_{\mathcal{R}(s_{0})}(-z)^{2i+2\gamma-2+\frac{1}{2}}\left\vert\partial_{z}^{i+1}\Psi_{{\frac{i}{100}}}\right\vert^{2}.
    \end{aligned}
    \label{eq: higher-order energy estimate for rp}
\end{equation}
Hence, combining \eqref{eq: higher-order energy estimate for rp} and \eqref{eq: higher-order energy estimate for psip} concludes the proof of this proposition.
\end{proof}

\subsection{Proof concluded}

Combining the lower-order point-wise estimate and the higher-order energy estimate up to $n = N+1$, we have the following proposition:
\begin{proposition}
\label{prop: map to itself}
    Given functions $(r_{p}^{0},\Psi_{p}^{0})\in \mathcal{C}_{N+1}^{\frac{7}{4}}\times\mathcal{C}_{N+1}^{\alpha}$ and $(f,g,h)\in\mathcal{S}_{\delta}^{a,N,\tau}$ with $\tau$ sufficiently small, $\alpha\in(p_{k},\frac{3}{2})$, and $a\in(0,\alpha^{\prime}q_{k})$, there exist constants $\widetilde{c}_{\infty}$ and $C(k)$ satisfying \begin{equation*}
        \left\vert \widetilde{c}_{\infty}\right\vert\leq C(k)\left(\left\Vert r_{p}^{0}\right\Vert_{\mathcal{C}_{N+1}^{\frac{7}{4}}}+\left\Vert\Psi_{p}^{0}\right\Vert_{\mathcal{C}_{N+1}^{\alpha}}+\left\Vert(f,g,h)\right\Vert_{\mathcal{S}^{a,N,\tau}}^{2}\right),
    \end{equation*}
    such that the solution $\mathfrak{A}$ to $\mathcal{P}(r_{p},\Psi_{p},m_{p}) = \mathcal{N}(f,g,h)$ with the initial data $(r_{p},\Psi_{p})|_{\Sigma_{-1}^{(in)}} = (r_{p}^{0},\Psi_{p}^{0}-\widetilde{c}_{\infty}r_{k}-\widetilde{c}_{\infty}r_{p}^{0})$ satisfies the estimate \begin{equation}
        \begin{aligned}
\left\Vert(\mathfrak{A}^{(1)},\mathfrak{A}^{(2)},\mathfrak{A}^{(3)})\right\Vert_{\mathcal{S}^{a_{\frac{N}{100}},N,\tau}}^{2}\leq C(k)\delta^{-\frac{N+1}{100}}\left(\left\Vert r_{p}^{0}\right\Vert_{\mathcal{C}_{N+1}^{\frac{3}{2}}}^{2}+\left\Vert\Psi_{p}^{0}\right\Vert_{\mathcal{C}_{N+1}^{\alpha}}^{2}+\left\Vert(f,g,h)\right\Vert_{\mathcal{S}^{a,N,\gamma}}^{4}\right).
        \end{aligned}
    \end{equation}
\end{proposition}
\begin{proof}
    By Proposition \ref{prop: higher order energy estimate on nonlinear problem} and Proposition \ref{prop: lower order energy estimate}, it remains to estimate $\mathfrak{A}^{(3)}$. Since $$m_{p}: =\mathfrak{A}^{(3)} = (m_{p})_{L}+\int_{-1}^{z}e^{\int_{z}^{\widetilde{z}}\frac{r_{k}}{\partial_{z}r_{k}}(\partial_{z}\phi_{k})^{2}}\mathcal{N}^{(3)}d\widetilde{z}.$$
    Commuting the $z$-transport equation \eqref{eq: transport equation for mp in the nonlinear setting} for $m_{p}$ with $\partial_{z}^{i}$ and arguing as in Proposition \ref{prop: weighted L2 estimate on mpl}, we have \begin{equation*}
        \left\Vert(m_{p})_{L}\right\Vert_{H_{N+1}^{\frac{3}{2}-p_{k}+\tau}}^{2}\lesssim \left\Vert\mathfrak{A}^{(1)}\right\Vert_{H_{N+1}^{-\frac{1}{4}+\tau}}^{2}+\left\Vert\mathfrak{A}^{(2)}\right\Vert_{H_{N+1}^{\frac{3}{2}-\alpha+\tau}}^{2}\lesssim \delta^{-\frac{N+1}{100}}\left(\left\Vert r_{p}^{0}\right\Vert_{\mathcal{C}_{N+1}^{\frac{7}{4}}}^{2}+\left\Vert\Psi_{p}^{0}\right\Vert_{\mathcal{C}_{N+1}^{\alpha}}^{2}+\left\Vert(f,g,h)\right\Vert_{\mathcal{S}^{a,N,\tau}}^{4}\right).
    \end{equation*}
    For the term with $\mathcal{N}^{(3)}$, we have \begin{align*}
        \left\Vert\int_{-1}^{z}e^{\int_{z}^{\widetilde{z}}\frac{r_{k}}{\partial_{z}r_{k}}(\partial_{z}\phi_{k})^{2}}\mathcal{N}^{(3)}\right\Vert_{H_{N+1}^{\frac{3}{2}-p_{k}+\tau}}^{2}\lesssim \left\Vert\partial_{z}\mathcal{N}^{(3)}\right\Vert_{H_{N}^{\frac{3}{2}-p_{k}+\tau}}^{2}+\left\Vert\mathcal{N}^{(3)}\right\Vert_{L^{2}}^{2}\lesssim \left\Vert(f,g,h)\right\Vert_{\mathcal{S}^{a,N,\tau}}^{4}.
    \end{align*}
    This concludes the proof.
\end{proof}
Hence, for any $a\in\left(0,\frac{\alpha^{\prime}q_{k}-1}{2}\right)$ and $0<\tau<\alpha-\alpha^{\prime}$, we can conclude that $\mathfrak{A}$ is a contraction.
\begin{corollary}
    Given $k\neq0$ sufficiently small, $\alpha\in(p_{k},\frac{3}{2})$, the initial data $(r_{p}^{0},\Psi_{p}^{0})\in\mathcal{C}_{N+1}^{\frac{7}{4}}\times\mathcal{C}_{N+1}^{\alpha}$, and functions $(f,g,h)\in\mathcal{S}_{\delta}^{a,N,\tau}$ with $a\in\left(0,\frac{\alpha^{\prime}q_{k}-1}{2}\right)$, $N\geq 5$, and $0<\tau<\alpha-\alpha^{\prime}$, there exist $\epsilon_{0}$ sufficiently small, $\delta$ small enough depending on $\epsilon_{0}$ and $k$, and $\widetilde{c}_{\infty}$ satisfying \begin{equation*}
        \left\vert\widetilde{c}_{\infty}\right\vert\lesssim_{k}\left(\left\Vert r_{p}^{0}\right\Vert_{\mathcal{C}_{N+1}^{\frac{7}{4}}}+\left\Vert\Psi_{p}^{0}\right\Vert_{\mathcal{C}_{N+1}^{\alpha}}+\left\Vert(f,g,h)\right\Vert_{\mathcal{S}^{a,N,\tau}}^{2}\right),
    \end{equation*} 
    such that for any $\left\Vert\Psi_{p}^{0}\right\Vert_{\mathcal{C}_{N+1}^{\alpha}}\leq \epsilon_{0}$, the solution map \begin{equation*}
        (f,g,h)\rightarrow(\mathfrak{A}^{(1)},\mathfrak{A}^{(2)},\mathfrak{A}^{(3)})
    \end{equation*}
    with the initial data $(r_{p},\Psi_{p})|_{\Sigma_{-1}^{(in)}} = (0,\Psi_{p}^{0}-\widetilde{c}_{\infty}r_{k}-\widetilde{c}_{\infty}r_{p}^{0})$ is a contraction in the function space $\mathcal{S}_{\delta}^{a,N,\tau}$.
\end{corollary}
\begin{proof}
    To show that $\mathfrak{A}$ maps $\mathcal{S}_{\delta}^{a,N,\tau}$ to $\mathcal{S}_{\delta}^{a,N,\tau}$, we use Proposition \ref{prop: map to itself} and take $\delta$ and $\epsilon$ to be sufficiently small. To show that \begin{align*}
        &\left\Vert\mathfrak{A}(f_{1},g_{1},h_{1})-\mathfrak{A}(f_{2},g_{2},h_{2})\right\Vert_{\mathcal{S}^{a,N,\tau}}\leq q\left\Vert(f_{1},g_{1},h_{1})-(f_{2},g_{2},h_{2})\right\Vert_{\mathcal{S}^{a,N,\tau}},\quad 0<q<1
    \end{align*} 
    we can argue similarly as in Proposition \ref{prop: map to itself}. This concludes the proof.
\end{proof}
A fixed-point argument will conclude the proof of Theorem \ref{thm: nonlinear stability result}.
\section{Nonlinear stability for perturbations at the threshold}
\subsection{Set-up of the initial perturbations}
In this section, we establish the nonlinear stability of the $k$-self-similar naked singularity spacetime under the initial perturbation on $\Sigma_{-1}^{(in)}$ at the threshold. We fix the gauge choice of the center $\Gamma$ to be $\Gamma = \{(-v) = (-u)^{q_{k}}\}$ and the gauge choice of the initial hypersurface $\Sigma^{(in)}_{-1}$ to be $r = r_{k}$. The initial data of $\phi$ is chosen to be \begin{equation*}
    \phi(-1,v) = \phi_{k}(-1,v)+\phi_{p},
\end{equation*}
where $\phi_{p}\in\mathcal{C}_{N+1}^{p_{k},\delta}$. Recall that $p_{k}: = \frac{1}{1-k^{2}}$. By Lemma \ref{lemma: decomposition of the function spaces}, we have that \begin{equation*}
    \phi_{p} = c_{0}(\phi_{p})\mr{\phi}+\bar{\phi}_{p},
\end{equation*}
where $\bar{\phi}_{p}\in\mathcal{C}_{N+1}^{p_{k}+\delta^{\prime}}$ for $0<\delta^{\prime}<\delta$ and $\mr{\phi} = \phi_{k}+k\log(-u)$. Under the $(r,\Psi,m)$-formulation, it is equivalent to consider the equations $\mathcal{P}(r_{p},\Psi_{p},m_{p}) = \mathcal{N}(r_{p},\Psi_{p},m_{p})$ with the initial data and the gauge choice \begin{equation*}
    r_{p} = 0,\quad \Psi_{p} = c_{0}(\phi_{p})\Psi_{k}+r_{k}\bar{\phi}_{p}.
\end{equation*}
Since the Einstein-scalar field equations have the scaling symmetry, i.e., if $(r,\Psi,m)$ is a solution, then $(\lambda r,\lambda\Psi,\lambda m)$ is also a solution, we have that $(r_{k},\Psi_{k},m_{k})$ is a kernel to the linearized operator $\mathcal{P}$. By \eqref{eq: scaling symmetry of the equations}, it is equivalent to consider the equations \begin{equation}
\begin{aligned}
    &\mathcal{P}(r_{p}-c_{0}r_{k},\Psi_{p}-c_{0}\Psi_{k},m_{p}-c_{0}m_{k}) = \mathcal{N}_{c_{0}}(r_{p}-c_{0}r_{k},\Psi_{p}-c_{0}\Psi_{k},m_{p}-c_{0}m_{k}),\\[1em]&
    r_{p}-c_{0}r_{k}|_{\Sigma_{-1}^{(in)}} = -c_{0}r_{k},\quad \Psi_{p}-c_{0}\Psi_{k} = r_{k}\bar{\phi}_{p},
\end{aligned}
\label{eq: initial value problem for the threshold case}
\end{equation}
where $\mathcal{N}_{c_{0}}$ is defined in Appendix \ref{appendix: nonlinear equations} for $b = c_{0}$.
To simply the notation, let \begin{equation*}
    \widetilde{r}_{p}: = r_{p}-c_{0}r_{k},\quad \widetilde{\Psi}_{p}: = \Psi_{p}-c_{0}\Psi_{k},\quad \widetilde{m}_{p}: = m_{p}-c_{0}m_{k}.
\end{equation*}
Then the initial value problem \eqref{eq: initial value problem for the threshold case} can be written as \begin{equation}
\label{eq: final reduction of the threshold case}
    \begin{aligned}
        &\mathcal{P}(\widetilde{r}_{p},\widetilde{\Psi}_{p},\widetilde{m}_{p})=\mathcal{N}_{c_{0}}(\widetilde{r}_{p},\widetilde{\Psi}_{p},\widetilde{m}_{p}),\\&
        (\widetilde{r}_{p},\widetilde{\Psi}_{p})|_{\Sigma_{-1}^{(in)}} = (-c_{0}r_{k},r_{k}\bar{\phi}_{p})\in \mathcal{C}^{\frac{7}{4}}\times \mathcal{C}^{p_{k}+\delta^{\prime}}.
    \end{aligned}
\end{equation}
\subsection{Proof concluded}
Since the problem has been reduced to \eqref{eq: final reduction of the threshold case} with the initial data of $\widetilde{\Psi}_{p}$ of higher regularity, we can apply the argument in Section \ref{sec: conclude the nonlinear stability} with $r_{p}^{0} = -c_{0}r_{k}$ and $\Psi_{p}^{0} = r_{k}\bar{\phi}_{p}$. We can show that \begin{proposition}
Given $\delta>0$ and $c_{0}\in \mathbb{R}$, for any $0<\delta^{\prime}<\delta$, there exists $\epsilon$ sufficiently  small, such that there exists a solution to \eqref{eq: final reduction of the threshold case} solving the initial data $\bar{\phi}_{p}\in\mathcal{C}_{N+1}^{p_{k}+\delta}$ with $\Vert\bar{\phi}_{p}\Vert_{\mathcal{C}_{N+1}^{p_{k}+\delta^{\prime}}}\leq \epsilon$. Moreover, we have the estimates \begin{align}
    \sum_{0\leq i+j\leq 1}\left\vert\partial_{s}^{i}\partial_{z}^{j}\widetilde{r}_{p}\right\vert\lesssim&\left(\left\vert c_{0}\right\vert\left\Vert r_{k}\right\Vert_{\mathcal{C}_{N+1}^{\frac{7}{4}}}+\left\Vert r_{k}\bar{\phi}_{p}\right\Vert_{\mathcal{C}_{N+1}^{p_{k}+\delta^{\prime}}}\right)e^{-(1+\delta^{\prime}q_{k})s}\\\sum_{0\leq i+j\leq 1}\left\vert\partial_{s}^{i}\partial_{z}^{j}(\widetilde{\Psi}_{p}-c_{\infty}r_{k})\right\vert\lesssim &\left(\left\vert c_{0}\right\vert\left\Vert r_{k}\right\Vert_{\mathcal{C}_{N+1}^{\frac{7}{4}}}+\left\Vert r_{k}\bar{\phi}_{p}\right\Vert_{\mathcal{C}_{N+1}^{p_{k}+\delta^{\prime}}}\right)e^{-(1+\delta^{\prime}q_{k})s},\\\sum_{0\leq i+j\leq 1}\left\vert\partial_{s}^{i}\partial_{z}^{j}\widetilde{m}_{p}\right\vert\lesssim&\left(\left\vert c_{0}\right\vert\left\Vert r_{k}\right\Vert_{\mathcal{C}_{N+1}^{\frac{7}{4}}}+\left\Vert r_{k}\bar{\phi}_{p}\right\Vert_{\mathcal{C}_{N+1}^{p_{k}+\delta^{\prime}}}\right)e^{-(1+\delta^{\prime}q_{k})s},
\end{align}
where $c_{\infty}$ satisfies the estimate \begin{equation}
    \left\vert c_{\infty}\right\vert\lesssim \left\vert c_{0}\right\vert\left\Vert r_{k}\right\Vert_{\mathcal{C}_{N+1}^{\frac{7}{4}}}+\left\Vert r_{k}\bar{\phi}_{p}\right\Vert_{\mathcal{C}_{N+1}^{p_{k}+\delta^{\prime}}}.
\end{equation}
\end{proposition}
\noindent Transforming back to the original $(r_{p},\Psi_{p},m_{p})$ concludes the proof of the threshold case.
\appendix
\section{The structure of nonlinear terms}
\label{appendix: nonlinear equations}
Recall that the spherically symmetric Einstein-scalar field equations under the self-similar coordinates $(s,z)$:
\begin{align}
\partial_{s}\partial_{z}r+q_{k}z\partial_{z}^{2}r+q_{k}\partial_{z}r &= \frac{\mu}{1-\mu}\frac{(\partial_{s}+q_{k}z\partial_{z})r\ \partial_{z}r}{r},\label{eq: nonlinear r equation}
    \\\partial_{s}\partial_{z}\Psi+q_{k}z\partial_{z}^{2}\Psi+q_{k}\partial_{z}\Psi+k\partial_{z}r&=\frac{\mu}{1-\mu}\frac{\partial_{z}r(\partial_{s}+q_{k}z\partial_{z})r}{r^{2}}\Psi,\label{eq: A2 equation}\\
    \partial_{z}m+\frac{r}{\partial_{z}r}\left(\partial_{z}\left(\frac{\Psi}{r}\right)\right)^{2}m &= \frac{1}{2}\frac{r^{2}}{\partial_{z}r}\left(\partial_{z}\left(\frac{\Psi}{r}\right)\right)^{2},\\
    \mu &= \frac{2m}{r}.\label{eq: nonlinear m equation}
\end{align}
Since the Einstein-scalar field equations have the translation symmetry: if $(r,\Psi,m)$ is a solution, then $(r,\Psi+cr,m)$ is also a solution. Hence, we fix the gauge choice of the background $k$-self-similar naked singularity solution to be $\Psi_{k}(s,0) = 0$. Let $r_{p} = r-r_{k}$, $\Psi_{p} = \Psi-\Psi_{k}$, and $m_{p} = m-m_{k}$. Then we can schematically write the Einstein-scalar field equations $\mathcal{F}(r,\Psi,m)= 0 $ as \begin{equation*}
    \mathcal{P}_{(r_{k},\Psi_{k},m_{k})}(r_{p},\Psi_{p},m_{p}) = \mathcal{N}_{(r_{k},\Psi_{k},m_{k})}(r_{p},\Psi_{p},m_{p}),
\end{equation*}
where $\mathcal{P}$ is the linear operator defined in Section~\ref{sec: setup of the nonlinear stability}, and $\mathcal{N}$ depends on $(r_{k},\Psi_{k},m_{k})$ and quadratically on $(r_{p},\Psi_{p},m_{p})$.

Let \begin{align*}
    &r_{p}^{b}: = r-r_{k}^{b}: = r-(1+b)r_{k},\quad \Psi_{p}^{b}: = \Psi-\Psi_{k}^{b}: = \Psi-(1+b)\Psi_{k},\quad \phi_{p}: = \phi-\phi_{k},\\&
    m_{p}: =m-m_{k}^{b}: = m-(1+b)m_{k},\quad \mu_{p}: = \mu-\mu_{k}.
\end{align*} 
Since the Einstein-scalar field equations are scaling invariant, we have that \begin{equation}
\label{eq: scaling symmetry of the equations}
    \begin{aligned}
       &\mathcal{P}_{(r_{k},\Psi_{k},m_{k})}(r_{p},\Psi_{p},m_{p})-\mathcal{N}_{(r_{k},\Psi_{k},m_{k})}(r_{p},\Psi_{p},m_{p})\\=&\mathcal{F}(r_{k}+r_{p},\Psi_{k}+\Psi_{p},m_{k}+m_{p})-\mathcal{F}(r_{k},\Psi_{k},m_{k}) \\=& \mathcal{F}(r_{k}^{b}+r_{p}^{b},\Psi_{k}^{b}+\Psi_{p}^{b},m_{k}^{b}+m_{p}^{b})-\mathcal{F}(r_{k}^{b},\Psi_{k}^{b},m_{k}^{b})\\=&\mathcal{P}_{(r_{k}^{b},\Psi_{k}^{b},m_{k}^{b})}(r_{p}^{b},\Psi_{p}^{b},m_{p}^{b})-\mathcal{N}_{(r_{k}^{b},\Psi_{k}^{b},m_{k}^{b})}(r_{p}^{b},\Psi_{p}^{b},m_{p}^{b}).
    \end{aligned}
\end{equation}
Since the linearized operator is invariant under the scaling of the background spacetime, we have \begin{equation*}
    \mathcal{P}_{(r_{k},\Psi_{k},m_{k})} = \mathcal{P}_{(r_{k}^{b},\Psi_{k}^{b},m_{k}^{b})}.
\end{equation*}
Therefore, we usually suppress the subscript of $\mathcal{P}$.

We have the following proposition.\begin{proposition}
\label{prop: nonlinear structure equation}
The nonlinear Einstein-scalar field equations \eqref{eq: nonlinear r equation}-\eqref{eq: nonlinear m equation} can be schematically written as \begin{equation}
    \mathcal{P}(r_{p}^{b},\Psi_{p}^{b},m_{p}^{b}) = \mathcal{N}_{(r_{k}^{b},\Psi_{k}^{b},m_{k}^{b})}\left(r_{p}^{b},\Psi_{p}^{b},m_{p}^{b}\right),\label{eq: schematic form of nonlinear things}
\end{equation}
where $\mathcal{P}$ is the linearized operator introduced in Section~\ref{sec: setup of the nonlinear stability}, and the nonlinearity $\mathcal{N}(r_{p}^{b},\Psi_{p}^{b},m_{p}^{b})$ has three components $\mathcal{N}^{(1)}$, $\mathcal{N}^{(2)}$, and $\mathcal{N}^{(3)}$. For $(f,g,h)\in \mathcal{S}^{a,N,\tau}$, $\mathcal{N}_{(r_{k}^{b},\Psi_{k}^{b},m_{k}^{b})}$ has the following estimates
\begin{align}
\left\vert(-z)^{\max\{0,j-\frac{3}{4}+\tau\}}\partial_{z}^{j}\mathcal{N}^{(1)}(f,g,h)\right\vert\lesssim& e^{-(1+2a)s}\left\Vert(f,g,h)\right\Vert_{\mathcal{S}^{a,N,\tau}}^{2},\quad 0\leq j\leq N-1,\label{eq: quadratic structure of the nonlinear terms in the fixed point argument for N1}\\
\left\vert(-z)^{\max\{0,j-\frac{3}{4}+\tau\}}\partial_{z}^{j}\mathcal{N}^{(2)}(f,g,h)\right\vert\lesssim& e^{-(1+2a)s}\left\Vert(f,g,h)\right\Vert_{\mathcal{S}^{a,N,\tau}}^{2},\quad 0\leq j\leq N-1,\label{eq: quadratic structure for N2}\\
\left\vert(-z)^{\max\{0,j+1-p_{k}\}}\partial_{z}^{j}\left(\frac{\mathcal{N}^{(3)}(f,g,h)}{\mr{r}^{2}}\right)\right\vert\lesssim& e^{-(1+2a)s}\left\Vert(f,g,h)\right\Vert_{\mathcal{S}^{a,N,\tau}}^{2},\quad 0\leq j\leq N-1,\label{eq: quadratic structure for N3}\\
\sum_{i = 1}^{N}\int_{-1}^{0}(-z)^{2i-\frac{5}{2}+2\tau}\left\vert\partial_{z}^{i}\mathcal{N}^{(1)}\right\vert^{2}\lesssim& e^{-(1+2a)s}\left\Vert(f,g,h)\right\Vert_{\mathcal{S}^{a,N,\tau}}^{4},\label{eq: L2 bound for the nonlinear term 1}\\\sum_{i = 1}^{N}
\int_{-1}^{0}(-z)^{2i-\frac{5}{2}+2\tau}\left\vert\partial_{z}^{i}\mathcal{N}^{(2)}\right\vert^{2}\lesssim& e^{-(1+2a)s}\left\Vert(f,g,h)\right\Vert_{\mathcal{S}^{a,N,\tau}}^{4},\label{eq: L2 bound for the nonlinear term 2}\\
\sum_{i = 1}^{N}\int_{-1}^{0}(-z)^{2i+1-2p_{k}+2\tau}\left\vert\partial_{z}^{i}\left(\frac{\mathcal{N}^{(3)}}{\mr{r}^{2}}\right)\right\vert^{2}\lesssim&e^{-(1+2a)s}\left\Vert(f,g,h)\right\Vert_{\mathcal{S}^{a,N,\tau}}^{4}.\label{eq: L2 bound for the nonlinear term 3}
\end{align}
\end{proposition}
\begin{proof}
    We can write the equations for $r_{p}$ and $\Psi_{p}$ as \begin{align}
        &\partial_{s}\partial_{z}r_{p}^{b}+q_{k}z\partial_{z}^{2}r_{p}^{b}+q_{k}\partial_{z}r_{p}^{b}\nonumber\\=&\frac{1}{r_{k}^{b}}
    \left[\frac{\mu_{k}}{1-\mu_{k}}+\frac{1}{1-\mu_{k}}\sum_{n = 1}^{\infty}\left(\frac{\mu_{p}}{1-\mu_{k}}\right)^{n}\right]\left(1+\sum_{n=1}^{\infty}(-1)^{n}\left(\frac{r_{p}^{b}}{r_{k}^{b}}\right)^{n}\right)(\partial_{z}r_{k}^{b}+\partial_{z}r_{p}^{b})\nonumber\\&\times \left((\partial_{s}+q_{k}z\partial_{z})r_{k}^{b}+(\partial_{s}+q_{k}z\partial_{z})r_{p}^{b}\right)-\frac{\mu_{k}}{1-\mu_{k}}\frac{(\partial_{s}+q_{k}z\partial_{z})r_{k}^{b}\partial_{z}r_{k}^{b}}{r_{k}^{b}}, \allowdisplaybreaks\label{eq: nonlineae wave equation for rp}
    \\[1em]&
\partial_{s}\partial_{z}\Psi_{p}^{b}+q_{k}z\partial_{z}^{2}\Psi_{p}^{b}+q_{k}\partial_{z}\Psi_{p}^{b}+k\partial_{z}r_{p}^{b}\nonumber\\=&\frac{1}{(r_{k}^{b})^{2}}\left(\frac{\mu_{k}}{1-\mu_{k}}+\frac{1}{1-\mu_{k}}\sum_{n = 1}^{\infty}\left(\frac{\mu_{p}}{1-\mu_{k}}\right)^{n}\right)\left(\sum_{n = 0}^{\infty}(-1)^{n}\left(\frac{r_{p}^{b}}{r_{k}^{b}}\right)^{n}\left(2+\frac{r_{p}^{b}}{r_{k}^{b}}\right)^{n}\right)(\partial_{z}r_{k}^{b}+\partial_{z}r_{p}^{b})\nonumber\\&\times\left((\partial_{s}+q_{k}z\partial_{z})r_{k}^{b}+(\partial_{s}+q_{k}z\partial_{z})r_{p}^{b}\right)(\Psi_{k}^{b}+\Psi_{p}^{b})-\frac{\mu_{k}}{1-\mu_{k}}\frac{\partial_{z}r_{k}^{b}(\partial_{s}+q_{k}z\partial_{z})r_{k}^{b}\cdot \Psi_{k}^{b}}{(r_{k}^{b})^{2}},\allowdisplaybreaks\label{eq: nonlinear wave equation for psip}\\&\mu_{p} = \frac{2m_{p}^{b}}{r_{k}^{b}}-\frac{2m_{k}^{b}}{r_{k}^{b}}\frac{r_{p}^{b}}{r_{k}^{b}}-2\frac{r_{p}^{b}}{r_{k}^{b}}\frac{m_{p}^{b}}{r_{k}^{b}}+\left(\frac{2m_{k}^{b}}{r_{k}^{b}}+\frac{2m_{p}^{b}}{r_{k}^{b}}\right)\sum_{n = 2}^{\infty}(-1)^{n}\left(\frac{r_{p}^{b}}{r_{k}^{b}}\right)^{n}.\label{eq: nonlinear mup equation}
    \end{align}
    Taking out the linear part of the right-hand sides of \eqref{eq: nonlineae wave equation for rp}-\eqref{eq: nonlinear wave equation for psip}, we can rewrite the above equations as \begin{equation*}
        \mathcal{P}^{(1)} = \mathcal{N}^{(1)},\quad \mathcal{P}^{(2)} = \mathcal{N}^{(2)},
    \end{equation*}
    where $\mathcal{N}^{(1)}$ depends quadratically on $(r_{p},\partial r_{p},m_{p})$, independent of $\Psi_{p}$, and $\mathcal{N}^{(2)}$ depends quadratically on $(r_{p},\partial r_{p},m_{p},\Psi_{p})$. For $\mathcal{N}^{(1)}$, we can estimate \begin{align*}
        \left\vert\mathcal{N}^{(1)}(f,g,h)\right\vert\lesssim& e^{-s}\left(\left\Vert e^{s}f(s,\cdot)\right\Vert_{L_{z}^{\infty}}+\left\Vert e^{s}\partial f(s,\cdot)\right\Vert_{L_{z}^{\infty}}+\left\Vert e^{s}h(s,\cdot)\right\Vert_{L_{z}^{\infty}}\right)^{2}\\\lesssim& e^{-s}e^{-2as}\left\Vert(f,g,h)\right\Vert_{\mathcal{S}^{a,N,\tau}}^{2} = e^{-(1+2a)s}\left\Vert(f,g,h)\right\Vert_{\mathcal{S}^{a,N,\tau}}^{2}.
    \end{align*}
    For the higher-order pointwise estimates on $\mathcal{N}^{(1)}$, arguing inductively and using Lemma \ref{lemma: sobolev embedding}, we have \begin{equation*}
        (-z)^{j-\frac{3}{4}+\tau}\left\vert\partial_{z}^{j}\mathcal{N}^{(1)}(f,g,h)\right\vert\lesssim e^{-(1+2a)s}\left\Vert(f,g,h)\right\Vert_{\mathcal{S}^{a,N,\tau}}^{2},\quad 1\leq j\leq N-1.
    \end{equation*}
    Similarly, we can prove \eqref{eq: quadratic structure for N2}.

For the $L^{2}$-estimates \eqref{eq: L2 bound for the nonlinear term 1}-\eqref{eq: L2 bound for the nonlinear term 2}, using $L^{2}$-norm to control the top-order derivatives and using $L^{\infty}$-norm to control the lower-order derivatives, we can derive \eqref{eq: L2 bound for the nonlinear term 1}-\eqref{eq: L2 bound for the nonlinear term 2} inductively.

Next, we consider the nonlinear transport equation \eqref{eq: nonlinear m equation}. We can write this equation as \begin{equation}
\begin{aligned}
    &\partial_{z}m_{p}^{b}+\frac{1}{\partial_{z}r_{k}^{b}}(r_{k}^{b}+r_{p}^{b})\left[\sum_{n = 0}^{\infty}(-1)^{n}\left(\frac{\partial_{z}r_{p}^{b}}{\partial_{z}r_{k}^{b}}\right)^{n}\right]\left(\partial_{z}\left(\frac{\Psi_{k}^{b}+\Psi_{p}^{b}}{r_{k}^{b}}\sum_{n = 0}^{\infty}(-1)^{n}\left(\frac{r_{p}^{b}}{r_{k}^{b}}\right)^{n}\right)\right)^{2}(m_{k}^{b}+m_{p}^{b})\\=&\frac{1}{2}\frac{(r_{k}^{b}+r_{p}^{b})^{2}}{\partial_{z}r_{k}^{b}}\left(\sum_{n = 0}^{\infty}(-1)^{n}\left(\frac{\partial_{z}r_{p}^{b}}{\partial_{z}r_{k}^{b}}\right)^{n}\right)\left(\partial_{z}\left(\frac{\Psi_{k}^{b}+\Psi_{p}^{b}}{r_{k}^{b}}\sum_{n = 0}^{\infty}(-1)^{n}\left(\frac{r_{p}^{b}}{r_{k}^{b}}\right)^{n}\right)\right)^{2}\\&+\frac{r_{k}^{b}}{\partial_{z}r_{k}^{b}}\left(\partial_{z}\left(\frac{\Psi_{k}^{b}}{r_{k}^{b}}\right)\right)^{2}m_{k}^{b}-\frac{1}{2}\frac{(r_{k}^{b})^{2}}{\partial_{z}r_{k}^{b}}\left(\partial_{z}\left(\frac{\Psi_{k}^{b}}{r_{k}^{b}}\right)\right)^{2}.
\end{aligned}
\end{equation}
Then the above equation can be written as \begin{equation*}
    \mathcal{P}^{(3)} = \mathcal{N}^{(3)},
\end{equation*}
where $\mathcal{N}^{(3)}$ depends quadratically on $(r_{p},\partial r_{p},\Psi_{p},\partial \Psi_{p},m_{p})$. For $(f,g,h)\in\mathcal{S}^{a,N,\tau}$, we have \begin{equation}
\begin{aligned}
\left\vert\frac{\mathcal{N}^{(3)}}{\mr{r}^{2}}\right\vert\lesssim& e^{-s}\left(\left\vert e^{s}\partial_{z}f\right\vert^{2}+\left\vert e^{s}\frac{f}{\mr{r}}\right\vert^{2}+\left\vert e^{s}\partial_{z}\left(\frac{f}{\mr{r}}\right)\right\vert^{2}+\left\vert e^{s}\partial_{z}\left(\frac{g}{\mr{r}}\right)\right\vert^{2}+\left\vert e^{s}\frac{g}{\mr{r}}\right\vert^{2}+\left\vert e^{s}\frac{h}{\mr{r}}\right\vert^{2}\right)\\\lesssim&
e^{-s}\left(\left\vert e^{s}\partial_{z}f\right\vert^{2}+\left\vert e^{s}f\right\vert^{2}+\left\vert e^{s}\partial_{z}g\right\vert^{2}+\left\vert e^{s}g\right\vert^{2}+\left\vert e^{s}\partial_{z}h\right\vert^{2}+\left\vert e^{s} h\right\vert^{2}\right)\\\lesssim& e^{-(1+2a)s}\left\Vert(f,g,h)\right\Vert_{\mathcal{S}^{a,N,\tau}}^{2}.
\end{aligned}
\end{equation}
Arguing similarly to the proof of \eqref{eq: quadratic structure of the nonlinear terms in the fixed point argument for N1} and \eqref{eq: L2 estimate for renormalized N1}, we can prove \eqref{eq: quadratic structure for N3} and \eqref{eq: L2 bound for the nonlinear term 3}.
    
\end{proof}
\section{Black hole formation under large smooth perturbations}
\label{appdx: large smooth perturbation}
In this section, we prove Theorem~\ref{thm: smooth perturbation instability}, which shows that for smooth perturbations that are large in the threshold topology $\mathcal{C}_{N}^{p_{k},\delta}$ while remaining small in the BV norm, the resulting spacetime will form a trapped surface. First, we recall the following criterion for the formation of trapped surfaces to spherically symmetric Einstein-scalar field equations, established in~\cite{chris91}.
\begin{theorem}[\cite{chris91}]
\label{thm: trapped surface}
Let $(\mathcal{M},g,\phi)$ be a solution to the spherically symmetric Einstein-scalar field equations \eqref{eq: ESF 1}--\eqref{eq: ESF 2}, $C_{o}^{+}$ be an outgoing null hypersurface emanating from the center, and $S_{1}$ and $S_{2}$ be two spheres on $C_{o}^{+}$ with $r|_{S_{2}}>r|_{S_{1}}$, i.e., $S_{2}$ lies in the exterior of $S_{1}$. Let $r_{1}: = r|_{S_{1}}$, $r_{2}: = r|_{S_{2}}$, $m_{1}: = m|_{S_{1}}$, and $m_{2}: = m|_{S_{2}}$, where $m$ is the Hawking mass $1-\frac{2m}{r}: = g(\nabla r,\nabla r)$. Denote \begin{equation*}
    \delta = \frac{r_{2}}{r_{1}}-1,\quad \eta = \frac{2(m_{2}-m_{1})}{r_{2}}.
\end{equation*}
Then there exist two positive constants $c_{0}$ and $c_{1}$, such that if \begin{equation}
    \delta\leq c_{0},\quad \eta>c_{1}\delta\log\left(\frac{1}{\delta}\right),
    \label{eq: trapped surface formation}
\end{equation}
then the ingoing light cone emanating from $S_{2}$ will intersect the apparent horizon and there exists a trapped region in $\mathcal{M}$.
\end{theorem}
Therefore, for the $k$-self-similar naked singularity, it suffices to verify the condition~\eqref{eq: trapped surface formation} in an exterior small self-similar neighborhood $\mathcal{U}_{s} = \{-1\leq u<0,\ 0\leq \frac{v}{(-u)^{q_{k}}}\leq \kappa\}$ of the ingoing cone $\{v= 0\}$ emanating from the naked singularity. Now we define a step function $\chi(v)$ which is zero for $v\leq 0$ and is $1$ for $v>0$. Let $\rho(v)$ be a smooth cut-off function with \begin{equation*}
    \text{supp} \rho\in[0,1],\quad \int_{0}^{1}\rho(v)dv=1.
\end{equation*}
Let $\rho_{\delta} = \frac{1}{\delta}\rho(\frac{v}{\delta})$ and $\chi_{\delta} = \rho_{\delta}*\chi$. Then $\chi_{\delta}$ is a smooth function with $\chi_{\delta}(v)\equiv0$ for $v\leq 0$ and $\chi_{\delta}\equiv 1$ for $v\geq \delta$. 

We impose the following gauge choice on the characteristic initial hypersurface $\{-1\leq u<0,\ v = 0\}\cup \{u = -1,\ 0\leq v\leq\kappa\}$:\begin{equation}
    r(u,0) = (-u),\quad r(-1,v) = r_{k}(-1,v),\quad \Omega^{2}(-1,0) = 1.\label{eq: exterior gauge choice}
\end{equation}
We impose the following initial data on $\phi$ on the same characteristic hypersurface: \begin{equation}
    \partial_{u}\phi(u,0) = \partial_{u}\phi_{k} = \frac{k}{(-u)},\quad \partial_{v}\phi(-1,v) = \partial_{v}\phi_{k}(-1,v)+\epsilon\chi_{\delta}(v),\label{eq: exterior initial condition}
\end{equation}
where $\epsilon$ will be chosen to be small later.

To show the formation of trapped surfaces, one cannot expect $\partial_{v}r>0$ in the whole region of the self-similar neighborhood $\mathcal{U}_{s}$. However, one can still manage to show that a suitable subset of $\mathcal{U}_{s}$ remains untrapped. We have the following proposition.
\begin{proposition}
For the spherically symmetric Einstein-scalar field equations \eqref{eq: ESF 1}-\eqref{eq: ESF 2} with the gauge choice \eqref{eq: exterior gauge choice} and characteristic initial data \eqref{eq: exterior initial condition}, the region \begin{equation*}
    \widetilde{\mathcal{U}}_{s}: = \left\{0\leq \frac{v}{(-u)^{1+k^{2}}}\leq \kappa\right\}
\end{equation*}
is free of trapped surfaces for $0<k^{2}<\frac{1}{3}$ and $\kappa,\epsilon$ sufficiently small.
\end{proposition}
\begin{proof}
Before we state our bootstrap assumptions, we first introduce the following renormalized quantities. Let $f_{\delta} = \int_{0}^{v}\chi_{\delta}(\widetilde{v})d\widetilde{v}$ and $\phi_{S}(u,v)= \epsilon\frac{f_{\delta}(v)}{(-u)}$. Then we decompose $(r,\phi,\mu)$ as \begin{align}
    &\phi(u,v) =\phi_{k}(u,v)+\phi_{p}(u,v) = \phi_{k}(u,v)+\phi_{S}(u,v)+\widetilde{\phi}(u,v),\\& r(u,v) = r_{k}(u,v)+r_{p}(u,v),\quad \mu(u,v) = \mu_{k}(u,v)+\mu_{p}(u,v).
\end{align}
We note that if $f_{\delta}\equiv0$, then necessarily we have $\phi_{S} = \widetilde{\phi} = 0$. We make the following bootstrap assumption in the region $\widetilde{\mathcal{U}}_{s}(u_{0},v_{0}): = \{-1\leq u\leq u_{0},\ 0\leq v\leq v_{0}\}\cap\widetilde{\mathcal U}_{s}$
\begin{equation}
\label{eq: bootstrap assumptions for appendix}
\begin{aligned}
   \mathcal{A}(u_{0},v_{0}) = &\sup_{\widetilde{\mathcal{U}}_{s}(u_{0},v_{0})}\left\vert(-u)^{-k^{2}}\partial_{v}r_{p}\right\vert+\sup_{\widetilde{\mathcal{U}}_{s}(u_{0},v_{0})}\left\vert\partial_{u}r_{p}\right\vert+\left\vert\mu\right\vert+\sup_{\widetilde{\mathcal{U}}_{s}(u_{0},v_{0})}\left\vert(-u)\partial_{u}\widetilde{\phi}\right\vert+\sup_{\widetilde{\mathcal{U}}_{s}(u_{0},v_{0})}\left\vert(-u)^{q_{k}}\partial_{v}\widetilde{\phi}\right\vert\\\leq&C\epsilon,
\end{aligned}
\end{equation}
for some universal constant $C$. Using the gauge choice $r_{p}(u,0) = 0$, we have \begin{equation}
    \vert r_{p}(u,v)\vert\leq \int_{0}^{v}\vert\partial_{v}r\vert(u,\widetilde{v})d\widetilde{v}\leq C\epsilon(-u) \frac{v}{(-u)^{q_{k}}}\leq C\epsilon\kappa (-u).\label{eq: improved bound for rp}
\end{equation}
To improve the bound for $\partial_{u}r_{p}$, we have \begin{equation*}
    \partial_{v}\partial_{u}r_{p} = \frac{\mu}{1-\mu}\frac{\partial_{u}r\partial_{v}r}{r}-\frac{\mu_{k}}{1-\mu_{k}}\frac{\partial_{u}r_{k}\partial_{v}r_{k}}{r_{k}} = O(\epsilon (-u)^{-q_{k}}),
\end{equation*}
where the second identity follows from the bootstrap assumption \eqref{eq: bootstrap assumptions for appendix}. Directly integrating the above equation, we have \begin{equation}
    \vert\partial_{u}r_{p}\vert(u,v) \leq C_{0}C\epsilon\frac{v}{(-u)^{q_{k}}}\leq C_{0}C\epsilon\kappa,\label{eq: improved bound for durp}
\end{equation}
where $C_{0}$ is a constant depending on the background $k$-self-similar spacetime.

To estimate $\mu_{p}$, we use the $v$-transport equations for $\mu$:\begin{equation*}
    \partial_{v}\mu+\left(\frac{\partial_{v}r}{r}+\frac{r}{\partial_{v}r}(\partial_{v}\phi)^{2}\right)\mu = \frac{r}{\partial_{v}r}(\partial_{v}\phi)^{2}.
\end{equation*}
Taking the difference, we have \begin{align*}
    \partial_{v}\mu_{p}+\underbrace{\left(\frac{\partial_{v}r}{r}+\frac{r}{\partial_{v}r}(\partial_{v}\phi)^{2}\right)}_{G(u,v)}\mu_{p} = \underbrace{-\mu_{k}\left(\frac{\partial_{v}r}{r}+\frac{r}{\partial_{v}r}(\partial_{v}\phi)^{2}\right)_{p}+\left(\frac{r}{\partial_{v}r}(\partial_{v}\phi)^{2}\right)_{p}}_{F(u,v)}.
\end{align*}
Using the integration factor and the positivity of $G$, we have \begin{equation*}
    \vert\mu_{p}\vert \leq\int_{0}^{v}\vert F(u,\widetilde{v})\vert e^{\int_{v}^{\widetilde{v}}G(u,v^{\prime})dv^{\prime}}d\widetilde{v}\leq \int_{0}^{v}\vert F(u,\widetilde{v})\vert d\widetilde{v}
\end{equation*}
For $F$, by our bootstrap assumptions~\eqref{eq: bootstrap assumptions for appendix}, the choice of the initial perturbation $\chi_{\epsilon}$ and Proposition \ref{prop: estimate on the exact k self similar spacetime} on the background spacetime, we have the estimates\begin{equation*}
   \vert F(u,v)\vert\leq C_{0}C\epsilon (-u)^{-1-k^{2}}.
\end{equation*}
Therefore, we have \begin{equation}
    \vert \mu_{p}\vert\leq C_{0}C\epsilon\frac{v}{(-u)^{1+k^{2}}}\leq C_{0}C\epsilon\kappa.\label{eq: improved estimate for mup}
\end{equation}
To estimate $\partial_{u}\widetilde{\phi}$, it suffices to estimate $\partial_{u}\phi_{p}$. Taking the difference between the equations for $\partial_{u}\phi$ and $\partial_{u}\phi_{k}$, we have \begin{equation*}
    \partial_{v}(r\partial_{u}\phi_{p}) = -\partial_{u}r\partial_{v}\widetilde{\phi}-\epsilon(-u)^{-1}\partial_{u}r\chi_{\delta}-r\left(\frac{\partial_{u}r}{r}\right)_{p}\partial_{v}\phi_{k}-\left(\frac{\partial_{v}r}{r}\right)_{p}\partial_{u}\phi_{k} = O(\epsilon (-u)^{-1}\chi_{\delta})+O(\epsilon(-u)^{-1+k^{2}}).
\end{equation*}
Therefore, we can get the estimate for $\partial_{u}\phi_{p}$: \begin{equation}
    \vert (-u)\partial_{u}\phi_{p}\vert\leq C_{0}C\epsilon\frac{v}{(-u)}\leq C_{0}C\epsilon \kappa.\label{eq: improved bound for duphip}
\end{equation}
Next, we turn to estimating quantities that need integration in the $u$-direction. For $\partial_{v}r_{p}$, we have \begin{equation}
    \partial_{u}\left(\partial_{v}r_{p}\right)-\left(\frac{\mu}{1-\mu}\frac{\partial_{u}r}{r}\right)\partial_{v}r_{p} = \partial_{v}r_{k}\left(\frac{\mu}{1-\mu}\frac{\partial_{u}r}{r}\right)_{p}.\label{du dv perturbation equation}
\end{equation}
Using the bootstrap assumption \eqref{eq: bootstrap assumptions for appendix} and improved bounds \eqref{eq: improved bound for rp}-\eqref{eq: improved estimate for mup}, we can control the right-hand side of \eqref{du dv perturbation equation}\begin{equation*}
    \left\vert\partial_{v}r_{k}\left(\frac{\mu}{1-\mu}\frac{\partial_{u}r}{r}\right)_{p}\right\vert\leq C_{0}C\epsilon \frac{v}{(-u)^{2}}.
\end{equation*}
Since $ -\frac{\mu}{1-\mu}\frac{\partial_{u}r}{r} $ is positive, we have the following improved bound for $\partial_{v}r_{p}$
\begin{equation}
    \vert\partial_{v}r_{p}\vert(u,v)\leq \int_{-1}^{u}C_{0}C\epsilon\frac{v}{(-\widetilde u)}d\widetilde u\leq C_{0}C\epsilon\frac{v}{(-u)}\leq C_{0}C\epsilon\kappa(-u)^{k^{2}}.
\end{equation}
It remains to control $\partial_{v}\widetilde{\phi}$. Note that $\phi_{S}$ is the solution to the equation \begin{equation*}
    \partial_{u}\partial_{v}\phi_{S}(u,v)+\frac{\partial_{u}r}{r}(u,0)\partial_{v}\phi_{S}(u,v) = 0.
\end{equation*}
Then we can write down the equation for $\partial_{v}\widetilde{\phi}$\begin{align}
    \partial_{u}(r\partial_{v}\widetilde{\phi}) = -r\left(\frac{\partial_{v}r}{r}\partial_{u}\phi\right)_{p}-r\left(\frac{\partial_{u}r}{r}\right)_{p}\partial_{v}\phi_{k}-r\left(\frac{\partial_{u}r}{r}(u,v)-\frac{\partial_{u}r}{r}(u,0)\right)\partial_{v}\phi_{S}.\label{eq: equation for dvphi tilde}
\end{align}
For the first term and the second term on the right-hand side of \eqref{eq: equation for dvphi tilde}, using the improved bounds \eqref{eq: improved bound for rp}--\eqref{eq: improved bound for duphip}, we have \begin{equation*}
    \left\vert r\left(\frac{\partial_{v}r}{r}\partial_{u}\phi\right)_{p}\right\vert\leq C_{0}C\epsilon\frac{v}{(-u)^{2-k^{2}}},\quad \left\vert r\left(\frac{\partial_{u}r}{r}\right)_{p}\partial_{v}\phi_{k}\right\vert\leq C_{0}C\epsilon\frac{v}{(-u)^{2-2k^{2}}}.
\end{equation*}
For the third term on the right-hand side of \eqref{eq: equation for dvphi tilde}, using Proposition \ref{prop: estimate on the exact k self similar spacetime}, we have \begin{equation*}
    \left\vert r\left(\frac{\partial_{u}r}{r}(u,v)-\frac{\partial_{u}r}{r}(u,0)\right)\partial_{v}\phi_{S}\right\vert\leq C_{0}C\epsilon \frac{v}{(-u)^{q_{k}}}\frac{\chi_{\delta}}{(-u)}.
\end{equation*}
Therefore, we have \begin{equation*}
    \left\vert(-u)^{q_{k}}\partial_{v}\widetilde{\phi}\right\vert\leq C_{0}C\epsilon\frac{v\chi_{\delta}}{(-u)^{q_{k}}}(-u)^{-k^{2}}\leq C_{0}C\epsilon\kappa.
\end{equation*}
Therefore, by choosing $\kappa$ sufficiently small, we can improve the bootstrap assumption \eqref{eq: bootstrap assumptions for appendix}. This concludes the proof.
\end{proof}

Now we are ready to close the proof of Theorem~\ref{thm: smooth perturbation instability}. Recall the $v$-transport equation for $m$:\begin{equation*}
    \partial_{v}m = \frac{1}{2}\frac{r^{2}}{\partial_{v}r}(\partial_{v}\phi)^{2}(1-\mu)\geq \frac{1}{4}\frac{r^{2}}{\partial_{v}r}\epsilon^{2}\chi_{\delta}^{2}(-u)^{-2}(1-\mu)-\frac{1}{2}\frac{r^{2}}{\partial_{v}r}(\partial_{v}\phi_{k}+\partial_{v}\widetilde{\phi})^{2}(1-\mu).
\end{equation*}
Let $m_{1} = m(u,0)$, $m_{2} = m(u,v)$, $r_{1} = r(u,0)$, and $r_{2} = r(u,v)$. Then by the bootstrap bound \eqref{eq: bootstrap assumptions for appendix}, for $\delta\leq v\leq \kappa(-u)^{1+k^{2}}$, we have \begin{align*}
    m_{2}-m_{1} = m(u,v)-m(u,0)\geq C\epsilon^{2}(-u)^{-k^{2}}\int_{0}^{v}\chi_{\delta}(\widetilde{v})d\widetilde{v}-C\kappa (-u)^{1+2k^{2}}\geq C\epsilon^{2}(-u)^{-k^{2}}(v-\delta)-C\kappa(-u)^{1+2k^{2}}.
\end{align*}
On the other hand, we have \begin{equation*}
    \delta = \frac{r_{2}-r_{1}}{r_{1}} = \frac{r_{2}-r_{1}}{(-u)} = \frac{1}{(-u)}\int_{0}^{v}\partial_{v}r\leq Cv(-u)^{-q_{k}}\leq C\kappa(-u)^{2k^{2}}\rightarrow0.
\end{equation*}
Therefore, we have \begin{equation*}
    \delta\log\frac{1}{\delta}\rightarrow 0,\quad u\rightarrow 0.
\end{equation*}
For fixed $\epsilon$ sufficiently small, we can choose $\delta = \epsilon^{100}$. Then, for $u = -\epsilon^{10}$ and $\epsilon^{100}\leq v=\kappa \epsilon^{10(1+k^{2})}$, we have \begin{equation*}
    \frac{m_{2}-m_{1}}{r_{1}}\geq \frac{1}{2}\kappa\epsilon^{2}\epsilon^{10}\geq c_{1}\kappa\epsilon^{20k^{2}}(-\log(20k^{2})),
\end{equation*}
for $\epsilon$ sufficiently small. Therefore, by Theorem~\ref{thm: trapped surface}, we can conclude that a trapped surface will form in the future of the ingoing cone emanating from $S_{2}$.

\bibliographystyle{plain}
\bibliography{references}

\end{document}

%% file: figures/characteristic_naked_singularity.tex


\begin{tikzpicture}[scale=1.5]

\coordinate (n1) at (0,0);
\coordinate (n2) at (45:6);
\draw[thick] (n1) -- (n2);

\coordinate (n3) at ($(n2) + (135:2) $);
\coordinate (b1) at ($(n3) + (-135:2) $);
\node[circle, inner sep = 0, minimum size = .1cm, label = {[xshift=1mm, yshift=-.5mm]above:{$i^+$}}, draw] (node1) at (n3) {};
\draw[dashed] (n2) -- (node1);

\node[label = left:{$\{r=0\}$}] at ($(0,1.5) + (.1,0)$) {};

\coordinate (n8) at ($(node1) + (-135:4)$) ;
\draw[dashed] (node1) --+ (-135:4);

\node[circle, draw, inner sep = 0, minimum size = .1cm, label = left:{$\mathcal{O}$}] (Onode) at (n8) {};

\node[label = left:{``naked singularity"}] (Bnode) at ($(n8)+(0,.2) + (.3,0)$) {};

\draw[thick] (n1) -- (Onode);
\draw[dashed] (Onode) -- ($(Onode) + (-45:2)$);
\coordinate (Hend) at ($(Onode)+(-45:2)$);
\draw[dashed] (Onode) -- (Hend);

\coordinate (Hmid) at ($(Onode)!0.75!(Hend)$);

\draw[->]
($(Hmid)+(45:2.2)$)
-- node[midway, above=2pt, sloped] {singular horizon $\underline{C}_{o}^{-}$}
(Hmid);
\coordinate (n7) at (45:4) ;
\node[label = {[rotate = 45]below:{initial data hypersurface $C_{out}$}}] at (45:3.4) {};

\coordinate (n95) at ($(0,1.9) + (45:4.67) $ );
\coordinate (n10) at ($(n95) + (-45:.4) $);
\coordinate (n11) at ($(n95) + (-135:.4) $);
\fill[blue!20] ($(0,1.9) + (45:4.67) $ ) -- (n10) -- (n11);
\node[circle, blue, fill, draw, inner sep = 0, minimum size = .1cm, label = {[yshift = 1mm, xshift = -.7mm]right:{distant observer}}] (obs) at (n95) {};

\draw[->, blue] (obs) -- (n10) ;
\draw[->, blue] (obs) -- (n11) ;

\coordinate (n12) at ($(n95) + (0,-.28) $);
\coordinate (n13) at ($(n12) + (-135:.5) $ );

\draw[->] (n13) to[out=10, in = -110] (n12);

\node[label = {[yshift = 2mm]below:{causal past}}] at (n13) {};

\coordinate (b2) at ($(n2) + (-135:.7)  $ ) ;
\coordinate (b3) at ($(n3) + (-135:.7)  $ ) ;
\coordinate (b4) at ($(Onode) + (-45:1)  $ ) ;
\node[label = {[rotate = 45, yshift=-1mm]above:{ $ r\rightarrow \infty $ }}] at (b2) {};
\node[label = {[rotate = 45,yshift=-1mm]above:{ $ r\rightarrow \infty $ }}] at (b3) {};
\coordinate (p1) at (60:3);

\node at (65:3.8) {exterior};
\coordinate (p2) at (70:1.5);
\node[label = above:{interior}] at (p2){};
\end{tikzpicture}
